\documentclass[12pt,twoside]{amsart}
\usepackage[margin=3cm]{geometry}
\usepackage[colorlinks=false]{hyperref}
\usepackage[english]{babel}
\usepackage{graphicx,titling}
\usepackage{float}
\usepackage{amsmath,amsfonts,amssymb,amsthm}
\usepackage{lipsum}
\usepackage[T1]{fontenc}
\usepackage{fourier}
\usepackage{color}
\usepackage[latin1]{inputenc}
\usepackage{esint}
\usepackage{caption}
\usepackage{ccicons}

\makeatletter
\def\blfootnote{\xdef\@thefnmark{}\@footnotetext}
\makeatother

\newcommand\ccnote{
    \blfootnote{\copyright\,\, Esteban Cárdenas, Joseph K. Miller, and Nata\v{s}a Pavlovi\'c}
    \blfootnote{\ccLogo\, \ccAttribution\,\, Licensed under a \href{https://creativecommons.org/licenses/by/4.0/}{Creative Commons Attribution License (CC-BY)}.}
}

\usepackage[export]{adjustbox}
\numberwithin{equation}{section}
\usepackage{setspace}

\renewcommand{\leq}{\leqslant}

\renewcommand{\geq}{\geqslant}
\renewcommand{\mathbb}{\varmathbb}
\usepackage{fancyhdr}
\newtheorem{theorem}{Theorem}[section]
\newtheorem{lemma}[theorem]{Lemma}

\newtheorem{proposition}[theorem]{Proposition}
\newtheorem{definition}[theorem]{Definition}
\newtheorem{remark}[theorem]{Remark}
\usepackage{mathrsfs}               
\usepackage{dsfont}
\usepackage{braket}

\usepackage{enumitem}						

\renewcommand{\d}{\mathrm{d}}							
\renewcommand{\i}{ \mathrm{ i }  }

\newcommand{\N}{\mathcal{N}} 		
\newcommand{\C}{\mathbb{C}}			 
\newcommand{\R}{\mathbb{R}} 			

\newcommand{\U}{\mathcal{U}}				
\newcommand{\calD}{\mathcal{D}}

\renewcommand{\L}{ \mathcal{L} }				 
\newcommand{\F}{\mathscr{F}}

\newcommand{\calS}{\mathcal{S}}

\renewcommand{\O}{\mathcal O}
\newcommand{\calC}{\mathcal{C}}

\newcommand{\calN}{\mathcal N}
\newcommand{\W}{  \mathcal{W}  }

\newcommand{\calR}{\mathcal{R}}

\renewcommand{\H}{\mathcal{H} }  			 
\newcommand{\h}{\mathfrak{h}}

\newcommand{\tr}{\mathrm{Tr}}
\newcommand{\bx}{		 \boldsymbol{x}}
\newcommand{\by}{		 \boldsymbol{y}}

\newcommand{\1}{\mathds{1}}
\newcommand{\0}{ \bf 0 }
\newcommand{\<}{\left\langle}							
\renewcommand{\>}{\right\rangle}

\renewcommand{\leq}{\leqslant}

\renewcommand{\geq}{\geqslant}     
\newcommand{\vp}{\varphi}

\renewcommand{\t}[1]{\textnormal{#1}}

\newtheorem{assumption}{Assumption}

\address{Esteban Cárdenas, The University of Texas at Austin, Department of Mathematics, 2515 Speedway, Austin TX, 78712, USA}
\email{eacardenas@utexas.edu}
\medskip
\address{Joseph K. Miller, Stanford University,  Department of Mathematics, 450 Jane Stanford Way Building 380, Stanford CA, 94305, USA} 
\email{jkm314@stanford.edu}
\medskip
\address{Nata\v{s}a Pavlovi\'{c}, The University of Texas at Austin, Department of Mathematics, 2515 Speedway, Austin TX, 78712, USA}
\email{natasa@math.utexas.edu}

\begin{document}

\thispagestyle{empty}

\begin{minipage}{0.28\textwidth}
\begin{figure}[H]
\includegraphics[width=2.5cm,height=2.5cm,left]{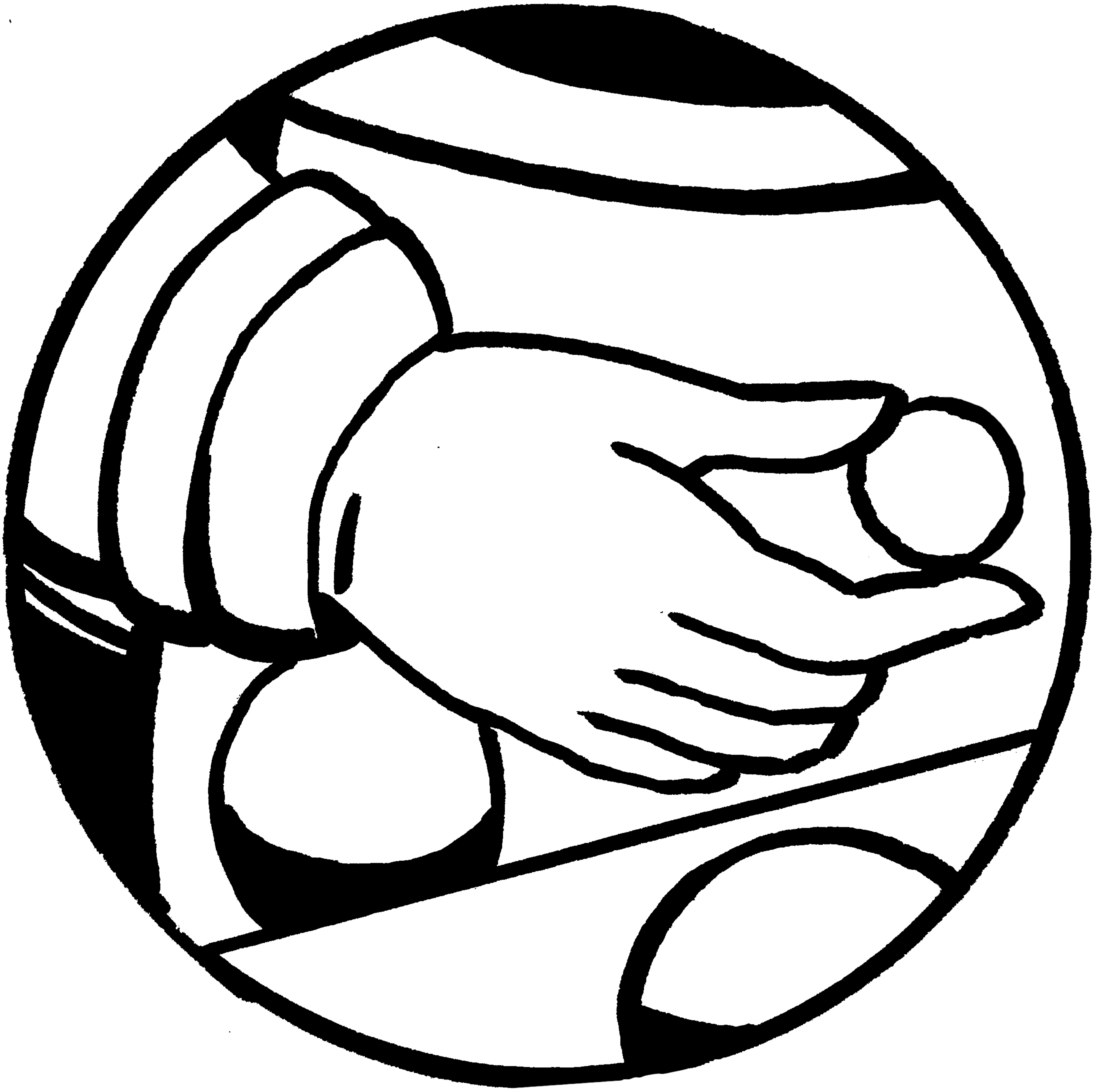}
\end{figure}
\end{minipage}
\begin{minipage}{0.7\textwidth} 
\begin{flushright}
Ars Inveniendi Analytica (2025), Paper No. 5, 54 pp.
\\
DOI 10.15781/av4e-gx71
\\
ISSN: 2769-8505
\end{flushright}
\end{minipage}

\ccnote

\vspace{1cm}


\begin{center}
\begin{huge}
\textit{On the effective dynamics}

\textit{of Bose-Fermi mixtures}

\end{huge}
\end{center}

\vspace{1cm}


\begin{minipage}[t]{.28\textwidth}
\begin{center}
{\large{\bf{Esteban Cárdenas}}} \\
\vskip0.15cm
\footnotesize{The University of Texas at Austin}
\end{center}
\end{minipage}
\hfill
\noindent
\begin{minipage}[t]{.28\textwidth}
\begin{center}
{\large{\bf{Joseph K. Miller}}} \\
\vskip0.15cm
\footnotesize{Stanford University}
\end{center}
\end{minipage}
\hfill
\noindent
\begin{minipage}[t]{.28\textwidth}
\begin{center}
{\large{\bf{Nata\v{s}a Pavlovi\'c}}} \\
\vskip0.15cm
\footnotesize{The University of Texas at Austin}
\end{center}
\end{minipage}

\vspace{1cm}


\begin{center}
\noindent \em{Communicated by Jean Dolbeault}
\end{center}
\vspace{1cm}


\noindent \textbf{Abstract.} \textit{In this work, we describe the dynamics of	a Bose-Einstein condensate interacting with a degenerate Fermi gas at zero temperature. 
	First, we analyze the mean-field approximation of the 
	many-body Schr\"odinger dynamics and prove   
the	emergence of a  coupled Hartree-type system of equations. 
	We obtain rigorous error control that yields a non-trivial scaling window in which the approximation 
	is meaningful. 	Second, starting from this Hartree system,
	we identify a  novel scaling regime in which the fermion  distribution
	behaves semi-classically, but the boson field remains quantum-mechanical; 
	this is  one of the main contributions of the present article. 
	In this regime, the  bosons are much lighter and more numerous than the fermions.
	We then prove convergence to a coupled Vlasov-Hartree system of equations
	with an explicit convergence rate. }
\vskip0.3cm

\noindent \textbf{Keywords.} Bose-Fermi mixtures, mean-field equations, semi-classical limit. 
\vspace{0.5cm}


	\section{Introduction}
\label{section introduction}

In this work, we study the dynamics of a gas composed of $M$ identical fermions and $N$ identical bosons
moving in   Euclidean space $\R^d$, 
for spatial dimensions $d\geq2.$ 
The Hilbert space for the system is the tensor product 
\begin{equation}
	\mathscr{H} 
	\equiv 
	L^2_a( \R^{dM}) \otimes L^2_s( \R^{dN}) \  , 
\end{equation}
where  $L_a^2$ and $L_s^2$ correspond to the subspaces of antisymmetric and symmetric functions, respectively.
We neglect any internal degrees of freedom the particles may have. 
We assume  that the two systems are non-relativistic
and interact by means of a   two-body potential $V : \R^d \rightarrow \R$. 
Thus, we  study the   Hamiltonian in 
dimensionless variables 
\begin{equation}
	\label{hamiltonian HNM}
	H 
	\equiv 
	\frac{\hbar^2}{2 m_F} \sum_{i =1}^M  (- \Delta_{x_i}) \otimes \1 
	+ 
	\frac{\hbar^2}{2 m_B} \sum_{j =1 }^N \1 \otimes   (- \Delta_{y_j})   
	+ 
	\lambda 
	\sum_{i,j =1 }^{N,M}
	V(x_i - y_j) \    . 
\end{equation}
Here,  all the physical parameters are effective (dimensionless) quantities
and for simplicity we shall make no distinction in the notation.
Thus, 
$\hbar>0$  corresponds to  Planck's   constant,   
$\lambda>0$ is the strength   of the interaction, and $m_F$ and $m_B$ are the respective masses of fermions and bosons. 
The first two terms in \eqref{hamiltonian HNM}
correspond to the kinetic energies of the two subsystems, 
whereas the third one corresponds to the potential energy associated to their interaction. 
In particular,    
$V$ will always be assumed to be regular enough so that the Hamiltonian 
$H$ is self-adjoint in its natural domain, and
its time evolution is well-defined. 

\vspace{2mm}

Since Bose-Fermi mixtures   have not been extensively studied in the mathematical literature, 
to better focus on the interplay between bosons and fermions through the  potential $ V $, 
this work intentionally omits two-body boson-boson and fermion-fermion interactions
\begin{equation}
	\qquad 
	\mathbb W _{BB}
	= 
	\mu_{BB}
	\sum_{j_1 < j_2}
	W(y_{j_1}  - y_{j_2} )
	\qquad 
	\t{and}
	\qquad
	\mathbb W _{FF} = 
	\mu_{FF}
	\sum_{i_1 < i_2} 
	W (x_{i_1}- x_{i_2})
\end{equation}
in the Hamiltonian \eqref{hamiltonian HNM}. 
On the other hand, systems of interacting bosons and interacting fermions have been broadly studied in the literature in the last few decades and we now give a brief historical review of the corresponding mathematical results.

\subsection{Historical background on single species systems}
The main goal of this paper is to study the   dynamics generated by the Hamiltonian 
$H$. 	
As is often the case in many-body systems,   the associated   Schr\"odinger equation 
\begin{eqnarray} \label{schrodinger intro}
	i \hbar \partial_t \psi(t) = H \psi(t) \ , \qquad \psi (t) \in \mathscr H 
\end{eqnarray}
is impossible to solve exactly 
and one must rely on effective approximations 
to understand the  physical   behavior of the system. 
In this context, one of the best understood approximations is the mean-field approximation.
In broad terms, 
one assumes here that the two-body interaction potential    between particles is weak, 
but the number of particles is large--one then   replaces
the total   interaction  by its average over position densities. 
Let us briefly describe what this idea leads to   in the context of single species gases. 

\vspace{2mm}

For cold systems of interacting bosons, 
assuming that the  potential $W(y_i - y_j)$  
that mediates their interaction 
varies over length scales
comparable to the size of the gas,  
the mean-field approximation leads 
to the derivation of the Hartree equation 
\begin{eqnarray}
	\label{hartree}
	&  i \partial_t \vp(t) = -  \frac{1}{2}
	\Delta \vp(t) + 
	\big(W* |\vp (t)|^2 \big) \,  \vp(t)  \ , 
\end{eqnarray}
where the solution $ \vp(t) \in L^2(\R^d)$  corresponds to the wave-function 
of a Bose-Einstein condensate; 
if $ W(x)$ varies on much smaller scales, one obtains 
the Gross-Pitaevskii or nonlinear Schr\"{o}dinger (NLS) equation. 
The  Hartree equation \eqref{hartree}
was first rigorously derived in \cite{Spohn1980} for bounded potentials, 
and later in \cite{ErdosYau2001}
for Coulomb systems. 
More recently, 
the second quantization formalism
was employed  in \cite{RodnianskiSchlein2009,ChenLeeSchlein2011}--inspired
by studies on the fluctuation dynamics originated in
\cite{Hepp74,GinibreVelo1974,GinibreVelo79}--to obtain a quantitative convergence rate.
See also \cite{DietzeLee2022} where uniform-in-time bounds for error estimates are proven. 
We  refer the reader to the following 
non-exhaustive list of references
on related works \cite{FrohlichKnowlesSchwarz2009,AmmariNier2011,KnowlesPickl2010,Pick2011,GrillakisMachedon2010,GrillakisMachedonMargetis2011,GrillakisMachedon2012} 
on the derivation of the Hartree equation, 
and to 
\cite{ErdosSchleinYau2006,ErdosSchleinYau2010,ErdosSchleinYau2007,ErdosSchleinYau2009,AdamiGolseTeta2007,KirkPatrickSchleinStaffilani2011, Chen15} on the derivation of the Gross-Pitaevskii equation.

\vspace{2mm}

For  cold gases of $M$ interacting fermions,  one obtains the Hartree-Fock equation 
\begin{equation}
	\label{hartree fock}
	i \hbar  \partial_t \omega(t) 
	=
	\Big[
	-\frac{\hbar^2}{2} \Delta + W*\rho (t) - X(t)  , \omega(t) 
	\Big] \ , 
\end{equation}
where 
$ \rho(t,x)= M^{-1} \omega(t;x,x)$
is the density of particles,  
and $X(t) $ is the so-called 
exchange term. 
Here, the solution 
$\omega(t) $
is a positive,  trace-class operator on $ L^2(\R^d)$
whose trace is equal to $M$; 
it ought to describe an interacting    Fermi gas of $M$ particles.
The Hartree-Fock equation
has been historically studied in two scaling regimes.
The first derivation in the ``microscopic regime" (namely, physical scales for which $\hbar=1$)
was carried out in 
\cite{BardosGolseGottliebMauser2003} for regular interactions, 
and later improved in
\cite{FrohlichKnowles2011}
for Coulomb systems. 
In the ``macroscopic regime" 
(namely, physical scales for which
$ \hbar = M^{-1/d}$), 
the first derivation was carried out in \cite{ElgartErdosSchleinYau2004} for real analytic potentials, yielding an optimal convergence rate for short macroscopic times. 
More recently, the derivation was revisited in 
\cite{BenedikterPortaSchlein2014}
using second quantization methods, significantly relaxing the regularity of the potentials and extending the time validity of the derivation--as a tradeoff, here one requires additional semiclassical structure on the initial data. 
This inspired substantial work in the literature; 
see for instance
\cite{BenedikterPortaSchlein2014,BenedikterPortaSchlein2014-2,BenedikterJaksicPortaSaffirioSchlein2016,PortaRademacherSaffirio2017,FrestaPortaSchlein2023}. 

\vspace{2mm}

On the other hand, the $\hbar\downarrow 0 $ limit of the Hartree-Fock equation \eqref{hartree fock}
leads to the Vlasov equation
\begin{align}
	\label{vlasov}
	& (\partial_t + p \cdot \nabla_x + F_f(t) \cdot \nabla_p) f (t,x,p)= 0 
\end{align}
where $ F_f (t,x) = - \int \nabla W(x-y ) f(t,y,p) d y dp$
is a mean-field force  
and $f(t )\in L^1_+ (\R_{x,p}^{2d} )$ is a  macroscopic phase-space distribution function.
In particular, the Pauli exclusion principle, $ 0 \leq \omega \leq 1$ viewed in the sense of quadratic forms, still holds in the macroscopic limit   $ 0 \leq f \leq 1$, in the pointwise sense. 
One can therefore understand
the solution of the Vlasov equation \eqref{vlasov} as the description
of a macroscopic gas with quantum features. 
As for  the  derivation of the Vlasov equation
from interacting quantum systems, 
the first works on the subject are
\cite{NarnhoferSewell1981,Spohn1981}.
Here, the derivation is carried out 
in the   macroscopic regime, 
by studying directly 
the BBGKY hierarchy
associated to 
the many-body Schr\"odinger dynamics.
The convergence from the
Hartree/Hartree-Fock equation 
to the Vlasov equation
was later analyzed in 
\cite{LionsPaul1993,MarkowichMauser1993,GasserIllnerMarkowichSchmeiser1998}, 
although providing no convergence rate.
The first work
to provide a convergence rate for regular potentials
was \cite{AthanassoulisPaulPezzottiPulvirenti2011}.
Later, the derivation of a convergence rate
from the Hartree to the Vlasov dynamics was revisited
and established in
\cite{BenedikterPortaSaffirioSchlein2016, Chong21} 
for a larger class of potentials.

\subsection{Bose-Fermi mixtures}
Investigating degenerate mixtures of bosons and fermions is an extremely active area of research in experimental physics for constructing and understanding novel quantum bound states such as those in superconductors, superfluids, and supersolids \cite{FB14, Schreck01, Duda23}. These ultra-cold Bose-Fermi mixtures are fundamentally different from degenerate gases with only bosons or fermions. They not only show an enriched phase diagram, but also a fundamental instability due to energetic considerations coming from the Pauli exclusion principle \cite{Molmer1998}. In particular, the fermionic particles maintain a higher energy than the bosonic particles in the ground state, causing a physical instability due to the energetic difference. This difference bounds from above the number of fermions allowed to exist in these doubly degenerate mixtures \cite{DeSalvo17}. On the other hand, by varying the ratio of masses of bosons and fermions in these mixtures, experimentalists have studied Bose-Einstein condensates with bose-bose interactions mediated by fermions \cite{DeSalvo19}. 

\vspace{2mm }

Inspired by this activity in the physics community, in this paper we start exploring
the mathematical theory of Bose-Fermi mixtures
by studying the mean-field dynamics of the Hamiltonian introduced in \eqref{hamiltonian HNM}.
Here, one of the main challenges is understanding the physical scales of the system
that allow for suitable analysis. 
Indeed, the Pauli Exclusion Principle implies that for  confined gases of fermions at low temperatures, 
fermions have a characteristic energy 
that varies in 
a scale 
$ 
\hbar^2 M^{\frac{2}{d}}/ m_F,
$
whereas 
for bosons
this is only of order 
$\hbar^2 / m_B$. 
Thus, finding a   scaling regime in which one can capture
the effective dynamics of the system presents a challenge in itself that we have to address. 

\vspace{2mm}

Let us informally describe the main results of this paper, stated rigorously in Section \ref{section main results}. The first result, formulated in Theorem \ref{thm1}, 
contains 
a quantum mean-field approximation
of the many-body Schr\"odinger dynamics. 
We prove that 
the one-particle reduced density matrices for the corresponding fermionic and
bosonic subsystems (see \eqref{definition kernels} for the definition) are
effectively described by a pair of  interacting variables
\begin{equation}
	(\omega, \vp )
	: \R 
	\rightarrow 
	\mathscr{L}^1 
	( L^2 ( \R^d)   ) 
	\times L^2 ( \R^d  )   \ , 
\end{equation} 
satisfying the system of self-consistent equations, 
which we  shall refer to 
as the  {Hartree-Hartree equation}. 
\begin{align}
	\label{hartree hartree intro}
	\begin{cases}
		& i \hbar \partial_t \omega   
		= 
		[   - (\hbar^2 / 2m_F)\Delta 
		+ 
		\lambda N   \, 
		(V  * \rho_B )  , 
		\omega 
		]			\\
		&   i \hbar \partial_t  \vp 
		= 
		- (\hbar^2 / 2m_B)\Delta  \vp 
		+ 
		\lambda  M    \, 
		(V  *  \rho_F)   
		\vp  	
	\end{cases}   \   .
\end{align}
Here $\rho_F (t , x) = \frac{1}{M } \omega( t; x , x )$
and $\rho_B(t,x) = |\vp(t,x)|^2$
are    the 
fermionic and bosonic position densities, respectively.  
A few comments are in order. 

\begin{itemize}[leftmargin=* , label=$\square$]

	\item 
	The derivation of the above equation can be heuristically justified as follows. 
	We assume that at time $t=0$, 
	the system
	has been externally confined by means of an harmonic trap, 
	at zero temperature.
		The many-body wave function is   then
		assumed to be of the form
		$\psi (0) =     \psi_F(0) \otimes \psi_B(0) \in \mathscr H ,$
		where the fermionic component $\psi_F(0)$ is a 
		Slater determinant, 
		and the bosonic component $\psi_B(0)$ is a fully factorized state--for a precise meaning, see Assumption \ref{assumption 1}. 
		These assumptions 
		have been employed in  single-species systems across various works (see e.g. \cite{BenedikterPortaSchlein2014} and \cite{RodnianskiSchlein2009}).
		Additionally,    this condition has recently been  verified 
		for ground states of 
		Bose-Fermi mixtures in our collaboration with D. Mitrouskas 
		\cite[Corollary 1]{CardenasMillerMitrouskasPavlovic2025}, although in a somewhat different scaling. 
	If the interactions between particles are weak enough, 
	then the above structure for the wave function is approximately 
	valid also for later times.
	In particular, a direct calculation shows that plugging the formal ansatz 
	$\psi(t)  =   \psi_F(t) \otimes \psi_B(t)$--preserving the initial data  structure--in the Schr\"odinger dynamics \eqref{schrodinger intro} 
	leads to  the Hartree-Hartree equation  \eqref{hartree hartree intro}
	as a self-consistent approximation.

	\item 
	While the informal justification of the emergence of  \eqref{hartree hartree intro}
	may not present a   challenge,
	obtaining a scaling regime in which
	the above system 
	describes the leading order dynamics  
	is non-trivial. 
	More specifically, using Second Quantization methods, 
	we provide rigorous control
	of error terms.  They
	become small only in a particular non-empty parameter window, including
	both macroscopic and microscopic scaling regimes. See Theorem \ref{thm1} for more details.

	\item 
		In the physics literature, 
		the Hartree-Hartree system 
		appears for instance in \cite[Eq. 1]{Kar04}, 
		although   written in terms of the orbitals of $\omega(t) $ (i.e. as a system of equations). 
		Additionally, the authors
		consider  contact interactions, e.g. formal potentials  of the form 
		$V(x) = \delta(x)$ in 
		\eqref{hartree hartree intro}. 
		Note also these  authors  include additional
		Fermi-Fermi and Bose-Bose interactions which we here neglect, 
		as well as external trapping potentials. 
\end{itemize}

 In our   second result, stated in Theorem \ref{thm2}, 
we study a new scaling regime (contained in the aforementioned parameter window) 
for the combined system.
In this   scaling regime, 
the fermion component $\omega$
is described semi-classically, 
but the bosonic component $\vp$
remains quantum-mechanical. 
More precisely, the regime that we focus on  is given by
\begin{equation}
	\label{scaling regime intro}
	\lambda = \frac{1}{N} 
	\ , 
	\qquad \hbar =  \frac{1}{M^{\frac{1}{d}}}  \ , 
	\qquad  m_B = \hbar \ , 
	\qquad m_F =1 \, 
	\quad \t{and} \quad 
	N = M^{	 1 +		\frac{1}{d}	} \ . 
\end{equation} 
Let us note here that under this scaling regime, 
the Hamiltonian that drives the boson field  $\vp$
is proportional to $\hbar$. 
In other words,  it has the form
$	\hbar ( -  (1/2) \Delta +   V*\rho_F^\hbar ) 	$.
Thus, it follows that factors of  $\hbar$ cancel out in the second equation of \eqref{hartree hartree intro}, 
enabling  us to analyze the semi-classical limit. 
To this end, we consider $f^\hbar = W^\hbar[\omega^\hbar]$, 
the Wigner transform of the fermionic component (see \eqref{wigner transform} for the definition of the Wigner transform), 
and prove that in the $\hbar \downarrow 0$ limit, 
there is convergence 
$ (f^\hbar , \vp^\hbar) \rightarrow (f, \vp)$, 
where the latter variables 
satisfy a coupled system of equations. 
This system has the following form, 
and we shall refer to it  as the \textit{Vlasov-Hartree equation}
\begin{align}
	\label{vlasov hartree intro}
	\begin{cases} 
		& 		(\partial_t  + p \cdot  \nabla_x + F_B(t,x)  \cdot  \nabla_p   ) f  = 0 	\\
		&		  i \partial_t \vp 
		=   
		- \frac{1 }{2} \Delta   \vp +   ( V*\rho_F   )   \vp 	
	\end{cases}  \ . 
\end{align}
Here,  $F_B(t,x) \equiv 
- \int\nabla V(x-y ) | \vp(t,y)|^2 \d y 
$
is a mean-field force that the bosons exert over the fermions, and 
$\rho_F 	(t,x) \equiv  \int_{	\R^d	} f (t,x,p ) \d p 	 $
is the fermionic position density. 
Our proof of convergence is quantitative, 
and implements   for the problem at hand  recently developed techniques of Quantum Optimal Transportation (QOT) \cite{GolseMouhotPaul2016,GolsePaul2017,GolsePaul2021,GolsePaul2022,Lafleche2023}.

\vspace{2mm}

In addition to 
recognizing a mean-field scaling regime 
that allows us to  rigorously derive the Hartree-Hartree system \eqref{hartree hartree intro}, 
one of the main contributions of this article 
is the identification of a novel mean-field  
semi-classical scaling regime 
in which the limiting dynamics of  \eqref{hartree hartree intro} is non-trivial. 
To the authors best knowledge, this regime
had not been identified previously in the literature.

	In order to conclude this introductory section, 
	lets us briefly give a short experimental background on ultra-cold atomic gases. First, 
	the experimental realization of Bose-Einstein condensates  goes back to the groundbreaking works 
	\cite{Anderson95, Davis95}  which led to the Nobel prize in 2001 
	and promoted much activity in the field.
	Soon after, experiments for Fermi gases were realized, 
	and the first observation  of a degenerate Fermi gas 
	is due to \cite{DeMarco99}.
	For a nice discussion between theory and experiments, 
	we refer the reader to the review article \cite{Gio08}.
	As for  Bose-Fermi mixtures, 
	a  nice review article with
	recent experiments can be found for instance
	in
	\cite[Table 1]{Onofrio16}. 
	We would like to point out in particular the article \cite{Park12}
	which studies a Bose-Fermi mixture in which  
	$ N / M  \sim 10     $ and $ m_B / m_F  \sim  0.5  $; 
	that is, an ultracold Bose-Fermi mixture
	of numerous, lighter bosons interacting with heavier fermions. 
	Finally, 
	we also 
	note that   small mass limits
	have been studied theoretically in  physics in the search for evidence of ultralight bosons in cosmology, see e.g. \cite{Moczetal2018}.
	Here, the authors consider the Schr\"odinger-Poisson system, and compare it to the Vlasov-Poisson equation in the classical limit with $\hbar/m_B \rightarrow 0$.

\subsection{Organization of this paper}
In Section \ref{section main results}
we formulate our main results in Theorem \ref{thm1}
and
\ref{thm2}. 
In Section \ref{section second quantization}
we give preliminaries on the Second Quantization formalism, that we will extensively use.
In Section \ref{section fluctuation} we study the dynamics of the fluctuations
around a combined Bose-Einstein condensate and degenerate Fermi gas, 
which then we use to prove Theorem \ref{thm1} in Section \ref{section proof thm 1}.
Next, in Section \ref{section QOT} 
we adapt  the formalism of   Quantum Optimal Transportation
and utilize it 
to prove Theorem \ref{thm2}.
Finally, we include  Appendix \ref{section analysis PDEs} where we  state some basic well-posedness results
regarding the PDEs introduced in this paper, 
and  Appendix \ref{appendix inf generator}
where we give details of the calculation of the infinitesimal generator of the fluctuation dynamics. 

\subsection{Acknowledgments}
E.C is very thankful to  François Golse for an enlightening conversation regarding QOT, 
and to Niels Benedikter, Marcello Porta and  Chiara Saffirio for  helpful discussions
regarding the mean-field dynamics of Fermi systems. 
The authors are deeply grateful to Thomas Chen and Laurent Lafleche 
for their valuable comments that helped improve the first version of this manuscript. 
E.C. gratefully acknowledges support from the Provost  Graduate Excellence Fellowship at The University of Texas at Austin and from the NSF grant DMS-2009549, 
and the NSF grant DMS-2009800 through Thomas Chen. 
J.M. gratefully acknowledges support from the Provost Graduate Excellence Fellowship at The University of Texas at Austin and from the NSF grants No. DMS-1840314.
N.P. gratefully acknowledges support from the NSF under grants No. DMS-1840314, DMS-2009549
and DMS-2052789.

\section{Main Results}
\label{section main results}
In this section, we describe the main results of this article, 
that have already been announced in the introductory section. 
In particular, in subsection \ref{subsection thm 1} we present Theorem \ref{thm1}, 
describing
the quantum mean-field approximation
of the many-body Schr\"odinger dynamics. 
Here, we prove 
an upper bound on the error term 
that comes from the approximation of 
the one-particle reduced density matrices for the corresponding fermionic and
bosonic subsystems, 
and the solution of the Hartree-Hartree equation \eqref{hartree hartree intro}. 
In subsection \ref{subsection thm 2} we present Theorem \ref{thm2}, in which 
we study the scaling regime \eqref{scaling regime intro} for the Bose-Fermi system. 
We prove that in the $\hbar \downarrow 0$ limit, 
there is convergence towards 
the Vlasov-Hartree equation \eqref{vlasov hartree intro}. 
As stated in the Introduction,
one of the main contributions of this article 
is the identification of a semi-classical scaling regime
in which the limiting dynamics of the coupled system is non-trivial--to the authors best knowledge, this regime
had not been identified previously in the literature. 
In subsection \ref{subsection strategy of proofs} we briefly discuss
the strategy of our proofs and the methods that we employ.

\smallskip

\textit{Notations}. 
Before we move on to the 
main results of this section, let us introduce some notation that we will be using in the rest of the article. 
\begin{itemize}[label=- , leftmargin=*]
	\item   $L^p (\R^n )$   denotes the   Lebesgue spaces of $p$-th integrable functions,  
	for $p \in [1,\infty]$. 
	The subset of non-negative functions is denoted by $L^p_+(\R^n )$.
	\item $\mathscr P_m (\R^n)$ is the space of probability measures on $\R^n$
	that have finite $m\in \mathbb N $ moments. 
	\item $\mathscr S(\R^n ) $ denotes the space of Schwartz functions of rapid decay. 
	\item $W^{k,p} (\R^n )$
	for $k\in \mathbb N$ and $p\in[1,\infty]$, 
	denotes the Sobolev space of functions with derivatives of order $k$, that are $p$-th integrable. 
	\item $H^s (\R^n ) = W^{s,2} (\R^n )$ for $s \geq 1 $ 
	and $\dot H ^s (\R^n) $ is the usual homogeneous Sobolev space. 
	
	\item $\mathscr L ^1 (X)$ stands for the Banach space of trace-class operators over $X$, endowed 
	with the  norm  $\|   A \|_{\tr} \equiv \tr |A|$. 
	Similarly, $\mathscr L ^2 (X)$ is the space of Hilbert-Schmidt
	operators with norm $\| A\|_{HS} \equiv   \|   A^*A  \|_{\tr}^{1/2}$. 
	
	\item 
	We say that $C>0$ is a
	\textit{constant} if it is a positive number, 
	independent of the physical parameters
	$N,M , \hbar , \lambda , m_F, m_B$ and $t$. 
	
	\item 
We write
  $A + h.c \equiv A + A^*$
  to denote
   ``adding the hermitean conjugate".
	
	\item   $\<\xi \> =  (1 + \xi^2)^{1/2}$ denotes the standard angle bracket. 
\end{itemize}

\subsection{The mean-field approximation}
\label{subsection thm 1}

As we have previously discussed, the main interest in this article is 
to consider the mean-field  dynamics generated by the Hamiltonian $H$, 
introduced in 
\eqref{hamiltonian HNM}. 
To this end, we introduce the wave function of the system at time $t \in \R$ 
\begin{eqnarray}
	\label{wave function}
	\psi(t) 
	\equiv 
	\exp \Big(
	- i t H / \hbar 
	\Big)  \psi 
\end{eqnarray}  
where $\psi \in \mathscr H $ is the initial data of the system. 
Since our gas corresponds of two subsystems, each composed of identical particles, it will be crucial
to introduce  
the corresponding fermionic and bosonic one-particle reduced density matrices. 
These 
are the time-dependent
trace-class operators
$\gamma_F(t), \,  \gamma_B(t) \in \mathscr{L}^1(L^2(\R^d))$ 
whose kernels are defined as  the partial traces 
\begin{align}
	\label{definition kernels}
	\begin{cases}
		&  \gamma_F (t ; x ,x') 
		\equiv 
		M \, 
		\int_{  \R^{d(M-1)} \times \R^{dN }	}
		\psi (  t ;  x ,  \bx_{M-1};  \by_N  )
		\overline{ \, \psi \,  }   ( t; x' ,   \bx_{M-1};  \by_N ) 
		\d   \bx_{M-1} \d    \by_N  	\\ 
		& 	\gamma_B (t ; y ,y') 
		\equiv 
		N  \, 
		\int_{  \R^{dM } \times \R^{d(N-1)  }	}
		\psi   (   t;    {\bx}_M ;   y  ,   {\by}_{N-1}  )
		\overline{ \,  \psi \, }  (   t ;   {\bx}_M ;   y'  ,   {\by}_{N-1}  )   
		\d  {\bx}_M  \d  {\by}_{N-1}    \\ 
	\end{cases}
\end{align}
for $t \in \R$ and $x,x',y, y' \in \R^d$.
Here,  
we denote by $\bx_{M-1} = (x_1, \ldots, x_{M-1})$, 
$\by_N = (y_1, \ldots, y_N)$
and similarly
$\bx_M$
and $\by_{N-1}$, the variables that are being traced out. 
In particular, we note here that 
the normalizations are chosen so that  
for all times $t\in \R $
there holds 
\begin{equation}
	\tr\gamma_F(t) = M  
	\qquad 
	\t{and} 
	\qquad 
	\tr \gamma_B(t) =N  \ .    
\end{equation}

\vspace{2mm}
 
	We describe now 
	the conditions that we shall impose
	in the initial data $\psi \in \mathscr H $
	associated to the solution of the Schr\"odinger dynamics \eqref{wave function}. 
	Physically, the situation we consider  
	concerns the description of an initially prepared cold gas of fermions and bosons.
	We assume that  the 
	fermion component is    described as a degenerate Fermi gas--parametrized by a   Slater determinant--whereas the boson gas  
	undergoes  Bose-Einstein condensation, described by a single-particle wave function. 
	This is   made rigorous in Assumption \ref{assumption 1}, 
	and is motivated by previous results in single-species systems.
	In addition, we require additional assumptions on the scales
	in which the Fermi gas varies--see Remark \ref{remark semiclassical bounds} for more details.  
	
\vspace{2mm}

Let us now discuss the effective dynamics of this system. 
If the interactions between particles are weak enough, we expect the zero temperature structure described above to approximately persist for times $t>0$. 
 More precisely,  in our chosen scaling the force-per-particle remains  $O(1)$, 
	and the   initially prescribed structure    will be shown to   approximately persist in time. 
Thus,  a  mean-field approximation for the reduced density matrices
$\gamma_F$ and $\gamma_B$
is given in terms of  a pair of  interacting   variables
\begin{equation}
	(\omega, \vp )
	: \R 
	\rightarrow 
	\mathscr{L}^1 
	( L^2 ( \R^d)   ) 
	\times L^2 ( \R^d  )  
\end{equation} 
that solve a self-consistent equation.
A formal calculation using a time-dependent Slater determinant/fully factorized
ansatz combined with replacing the full interaction $V(x-y)$ 
with an average over the position densities then yields
\begin{align}
	\label{hartree hartree}
	\begin{cases}
		& i \hbar \partial_t \omega   
		= 
		[   - (\hbar^2 / 2m_F)\Delta 
		+ 
		\lambda N   \, 
		(V  * \rho_B )  , 
		\omega 
		]			\\
		&   i \hbar \partial_t  \vp 
		= 
		- (\hbar^2 / 2m_B)\Delta  \vp 
		+ 
		\lambda  M    \, 
		(V  *  \rho_F)   
		\vp  	\\
		& (\omega ,\vp)(0)  = (\omega_0 , \vp_0) \in 
		\mathscr{L}^1 
		( L^2 ( \R^d)   ) 
		\times L^2 ( \R^d  )
	\end{cases}   \   ,
\end{align}
up to leading order.
Here $\rho_F (t , x) = \frac{1}{M } \omega( t; x , x )$
and $\rho_B(t,x) = |\vp(t,x)|^2$
correspond to the 
fermionic and bosonic position densities, respectively. 
We shall refer to \eqref{hartree hartree}
as the  {Hartree-Hartree equation}.

We are now ready to rigorously state our assumptions which we indicated above.

\begin{assumption}
	[Schr\"odinger initial data]
	\label{assumption 1}
	We assume that  the initial data $\psi\in\mathscr H$
	satisfies the following conditions. 
	\begin{enumerate}[leftmargin=*]

		\item (Zero temperature)
		$\psi$ is  a factorized state of the form  
		\begin{equation}
			\notag
			\psi = \psi_F \otimes \psi_B \ . 
		\end{equation}
		Additionally, each factor satisfies the following assumptions. 
		\begin{itemize}[leftmargin=1cm]
			\item[(1.1)]
			There exists a rank-$M$ orthogonal projection $\omega_0 = \sum_{i=1}^M \ket{\phi_i} \bra{\phi_i} $
			on  $L^2(\R^d)$
			such that  
			\begin{equation}
				\notag
				\psi_F
				(x_1 , \ldots, x_M )
				= 
				\frac{1}{\sqrt{ M! }}
				\det_{1 \leq i,j\leq  M } 
				\big[
				\phi_i(x_j) 
				\big]  \ . 
			\end{equation}
			
			\item[(1.2)] There exists a unit vector  in the one-particle space $\vp_0 \in L^2(\R^d)$, 
			such that 
			\begin{equation}
				\notag 
				\psi_B (y_1, \ldots, y_N )
				= 
				\vp_0 (y_1) \cdots \vp_0 (y_N ) \  . 
			\end{equation}
		\end{itemize}
		
		\item (Semi-classical bounds) 
			We assume that there exists $C>0$
			such that
			\begin{align} 
				\label{semi1}
				\| [ x ,\omega_0]   \|_{\tr}  
				& \ \leq \ 
				C 
				m_F^{-1/4 }    
				M \sqrt{\hbar M^{-1/d} }   \\
				\label{semi2}
				\| [ i\hbar \nabla  ,\omega_0]   \|_{\tr}  
				& \ \leq \ 
				C m_F^{1/4 }    
				M \sqrt{\hbar M^{-1/d}}
			\end{align}  
			for all values of $\hbar,$ $M$ and $m_F$. 	 
	\end{enumerate}  
\end{assumption}

\begin{remark}
	[Reduced density matrices]
	\label{remark 1pdm}
	Let us observe that under the above assumptions, 
	one can calculate that the following relations hold at $t = 0$
	\begin{equation}
		\gamma_F(0) =
		\omega_0 
		\qquad
		\t{and}
		\qquad
		\gamma_B(0)
		= 
		N 
		\ket{\vp_0}
		\bra{\vp_0} \ . 
	\end{equation} 
	In other words, the initial data 
	is such that
	the one-particle reduced density matrices are given by
	$\omega_0$, 
	and 
	$N
	\ket{\vp_0}
	\bra{\vp_0}
	$, respectively. 
\end{remark}

	\begin{remark}[Semi-classical bounds]
		\label{remark semiclassical bounds}
		Two comments are in order regarding the semiclassical bounds that are present in Assumption \ref{assumption 1}. 
		\begin{enumerate}[leftmargin=*]
			\item[(i)]  
			Let $(\omega,\vp)$ be the solution of the Hartree-Hartree system
			\eqref{hartree hartree}
			with initial data $(\omega_0,\vp_0 ) \in \mathscr L ^1 (L^2(\R^d)) \times L^2(\R^d)  $
			verifying Assumption \ref{assumption 1}.
			Additionally, assume that the 
			parameters are constrained so that
			\begin{align}
				\label{scaling 2}
				m_F   \geq 1  , \qquad 
				\lambda N \leq m_F 
				\quad 
				\t{and}
				\quad 
				M^{ -\frac{1}{d}}
				\leq \hbar \ . 
			\end{align}
			Then, 
			we may adapt the proof of
			\cite[Proposition 3.4]{BenedikterPortaSchlein2014}
			and show that
			the semi-classical bounds \eqref{semi1} and \eqref{semi2}
			are propagated in time, 
			provided one updates the 
			constant $C(t)$. 
			In particular, there exists
			a constant $C_0>0$ such that
			for all $t\in \R$ 			\begin{eqnarray}
				\label{semiclassical structure}
				\|     [ e^{ i \xi \cdot x } , \omega (t) ]	   \|_{\tr }
				\leq C_0
				\exp(C_0 | t| ) 
				M 
				\hbar \<\xi \> \  \ , \qquad \forall \xi \in \R^d  . 
			\end{eqnarray}
			In order to avoid repetition, we do not repeat the proof here. 
			Let us comment that only the time dependent bound
			\eqref{semiclassical structure}
			enters the estimates in our proof. 
			Let us also note that the bound \eqref{semiclassical structure} holds automatically in the case $\hbar=1 $ corresponding to microscopic scaling; here, it would  not be  necessary to assume \eqref{semi1} and \eqref{semi2}. 
			In order to give an integrated proof that works in all cases, we choose to assume \eqref{semi1}
			and \eqref{semi2}. 
			
			\vspace{1mm}
			
			\item[(ii)]
			The semi-classical bounds \eqref{semi1}-\eqref{semi2} appeared first in 
			\cite{BenedikterPortaSchlein2014}
			in the derivation of the Hartree-Fock equation
			from interacting Fermi systems, 
			with scaling
			$m_F =1 $, $\lambda =1/M$ and
			$\hbar = 1/M^{1/d}$. 
			From the physical point of view,
			these commutator estimates 			
			state 
			that $\omega$
			is varying in a 
			\textit{macroscopic} scale. 	
			We refer the reader to the original reference for a more in-depth physical discussion. 
			In general, proving that an   orthogonal projection $\omega$
			verifies \eqref{semi1}-\eqref{semi2} is a non-trivial task.
			The first   known  examples
			were given for non-interacting systems, in which 
			$\omega =\textbf{1}_{( -\infty ,0 ]} (   H ) $
			where  
			$ H = -\hbar^2 \Delta + V(x)$
			is   a Schr\"odinger operator 
			with smooth potential $V(x)$. 
			See for instance 
			\cite[Theorem 3.2]{Benedikter2022} and  \cite[Theorem 1.2]{Soren2020}. 
			More recently, in \cite{Cardenas2025}  L. Lafleche and the first author of this paper
			verified the validity of the commutator estimates for potentials in the class $V\in C^{1, \frac{1}{2} }(\R^3)$. 
			As a consequence, one is able to verify \eqref{semi1}-\eqref{semi2}
			for an \textit{interacting} particle system in the Hartree approximation; 
			$\omega$ corresponds to the minimizer of a non-linear functional, 
			and one replaces $V(x)$ with $K * \rho_{\omega } (x)$  where  
			$K(x) = \pm |x|^{-a}$ for $0 < a < 1$. 
			Finally, let us comment 
			that  at positive temperature the situation is much better, see e.g.
			\cite{Chong2022}. 
		\end{enumerate}

	\end{remark}

\vspace{2mm}

The natural  topology in which convergence is expected to hold corresponds to that of trace-class operators. 
Our main theorem is the following result.

\begin{theorem}
	[The mean-field approximation]
	\label{thm1}
	Assume that the interaction potential satisfies
	$
	\int_{ \R^d 	}	
	\<\xi\>^2	|\hat V (\xi)	| \d \xi < \infty .
	$
	Let us consider the following: 
	\begin{itemize}[leftmargin=0.9cm, label=$\square$]
		\item 
		Let $\psi(t)= \exp(-itH/\hbar)\psi $ be the wave function of the system, 
		with initial data verifying 
		Assumption \ref{assumption 1}. 
		Let $\gamma_F(t)$
		and $\gamma_B(t)$
		be the one-particle reduced density matrices, as defined in \eqref{definition kernels}.
		
		\item 
		Let $(\omega(t), \vp(t))$
		be the solution of the Hartree-Hartree equation  \eqref{hartree hartree}.
	\end{itemize}
 Additionally, assume that
		the scaling regime is chosen so that 
		for all physical parameters
		$\lambda$, $\hbar$, $N$,  $M$, $m_F$ and $m_B$:  
		\begin{enumerate}
			\item[(i)] 
			For all $\ell \geq 1$ there exists 
			$ k_\ell  \geq  1$ such that 
			\begin{equation}
				\label{k condition}
				\frac{ \lambda \sqrt N }{\hbar}
				M^\ell 
				\leq 
				( \hbar M )^{k_\ell }. 
			\end{equation}
			
			\item[(ii)] 
			$m_F   \geq 1 , $  
			$ 	\lambda N \leq m_F $ 
			and 
			$M^{ -\frac{1}{d}}
			\leq \hbar $. 
		\end{enumerate}
 
	\noindent 	Then, there exists 
	a constant $C   > 0  $ 
	such that 
	for all $ t\in \R$
	there holds 
	\begin{align}
		\label{thm1 fermion estimate}
		\frac{1}{M }
		\| 
		\gamma_F  (t)  - \omega (t)   
		\|_{\tr}
		& \leq 
		\frac{C}{\sqrt  M }
		{  \exp   \bigg[  C   \lambda \sqrt{  \frac{NM}{\hbar }	 } 
			\bigg( 1 + \sqrt{\frac{\hbar M}{N}}    \bigg)  \exp |t|      \bigg]		}   \ ,  \\ 
		\label{thm 1 boson estimate}
		\frac{1}{N}
		\|
		\, \gamma_B (t)  - 
		N 
		\ket{\vp (t)}
		\bra{\vp(t) } 
		\|_{\tr}
		& \leq 
		\frac{C}{\sqrt  N  }
		{  \exp   \bigg[  C    \lambda \sqrt{   \frac{NM}{\hbar }	  } 
			\bigg( 1 + \sqrt{\frac{\hbar M}{N}}    \bigg)  \exp |t|      \bigg]		}  \ . 
	\end{align} 
\end{theorem}

\vspace{2mm}

\begin{remark}
	The above Theorem provides an explicit convergence rate 
	from the many-body Schr\"odinger dynamics
	to the solution of the Hartree-Hartree system.
	Note that we have chosen not to fix the 
	parameter regime in the theory--this is in contrast to 
	most works in the literature.
	The reason is that the scaling regime in which Theorem \ref{thm1}
	provides a reasonable approximation was \textit{not} known 
	by the authors in the onset of this investigation.
	Our interest then was not to prove an optimal convergence rate, 
	but to actually \textit{find} a meaningful scaling regime. 
\end{remark}

Regarding the previous remark, 	Theorem \ref{thm1} contains a meaningful approximation 
as long as  the argument in the time dependent function is $\O(1)$
with respect to the physical parameters. 
Let us   describe two scaling regimes that we regard as interesting. 

\vspace{2mm}

\textbf{Microscopic regime}. 
If one is working in microscopic units, we may set $\hbar  =1 $. 
One can then investigate the mean-field regime  in which the number of bosons and
fermions is the same. Namely 
\begin{equation}
	\lambda  =  \frac{1}{N }  \ , \qquad \hbar =1  
	\quad \t{and} \quad 
	N = M 
\end{equation}
and we also set $m_F = m_B =2 $ for completeness. 
Clearly, the condition \eqref{k condition}
is verified with $k_\ell= \ell $.
In this case, one should regard Theorem \ref{thm1}
as capturing the emergence of the mean-field equations 
\begin{align}
	\begin{cases}		
		& i    \partial_t \omega   
		= 
		[   -   \Delta 
		+ 
		(V   * \rho_B )  , 
		\omega 
		]			\\
		& i   \partial_t  \vp 
		= 
		-  \Delta  \vp 
		+   
		(V  *  \rho_F)    
		\vp   \  ,
	\end{cases}	
\end{align}
as the leading order term driving the dynamics of the Hamiltonian $H$, for our choice of initial data.  
Note that in this case, the semiclassical condition 
imposed in \eqref{semiclassical structure}
is verified immediately, independently of the structure of the initial data. 
However, since $ \mathrm{Tr}  \omega (t) = M $, the above equation  does not yield a non-trivial limit when  $M \rightarrow \infty$.

\smallskip 

\textbf{Macroscopic regime}.  In macroscopic units, the value of $\hbar$ becomes small.
As is well-known, 
for a system of  confined fermions, 
the energy  scale of each particle is 
$\hbar^2 M^{\frac{2}{d}}/m_F$. 
One is then interested in the regime for which  
this scale is  $\O(1)$ -- this is the so-called semi-classical limit that has been studied extensively in the literature for systems of interacting fermions.
On the other hand, for bosons the energy per particle has 
the scale 
$\hbar^2/m_B $. 
We can then tune the parameters so that the \textit{total energy} of the system is balanced.
For instance, we may look at
\begin{equation}
	\label{scaling regime}
	\lambda = \frac{1}{N} 
	\ , 
	\qquad \hbar =  \frac{1}{M^{\frac{1}{d}}}  \ , 
	\qquad  m_B = \hbar \ , 
	\qquad m_F =1 \, 
	\quad \t{and} \quad 
	N = M^{	 1 +		\frac{1}{d}	} \ . 
\end{equation} 
{It is possible to   check that   condition \eqref{k condition}
	is verified
	with 
	$ k_\ell 
	=
	\frac{1 + d (2\ell-1 )}{2(d-1) }.
	$
	Similarly,  
	one may readily verify that the condition \eqref{scaling 2} is satisfied.
	This leads to a natural candidate on the initial data for the fermionic component $\omega_0$
	that verifies Assumption \ref{assumption 1}; 
	see Remark \ref{remark semiclassical bounds} for more details.  
}

\subsection{The semi-classical limit}
\label{subsection thm 2}
In this subsection, we adopt the   
scaling regime given by  
\eqref{scaling regime}. 
Let us now motivate the upcoming semiclassical analysis of the 
coupled Hartree system.
In what follows, we shall   incorporate
the $\hbar$ dependence  on the solution $( \omega^\hbar, \vp^\hbar)$
of the coupled Hartree-Hartree equation \eqref{hartree hartree}.
We start by noting that one of the main consequences of the scaling regime \eqref{scaling regime}  is that 
$\lambda M = \hbar $.
Hence, the  Hamiltonian for the boson field  $\vp_t^\hbar$ is 
proportional to $\hbar$, i.e. it has the form
$	\hbar ( -  (1/2) \Delta +   V*\rho_F^\hbar ) 	$--it follows that factors of  $\hbar$ cancel out in the equation. 
Thus,  
the solution of the   coupled Hartree equation can now be analyzed semi-classically, in the limit
$\hbar \downarrow 0$.
Indeed,  for $t\in\R $
we   consider the Wigner transform of the fermionic density matrix 
\begin{equation}
	\label{wigner transform}
	f^\hbar  (t) 
	\equiv 
	W^\hbar[\omega^{\hbar} (t) ]
	\quad
	\t{where}
	\quad 
	W^\hbar[\omega ] ( x, p )
	\equiv 
	\frac{ 1}{(2 \pi)^{  d }   }
	\int_{\R^d }
	\omega 
	\Big( x +   \frac{y}{2} , x -  \frac{y}{2}
	\Big) 
	e^{		 - i  \frac{y\cdot p }{ \hbar }	}		 \d y   \ , 
\end{equation}
where $(x,p  ) \in      \R^d  \times \R^d .$
Heuristically, the pair $( f^\hbar , \vp^\hbar)$   converges to a solution 
$(f, \vp)  $  of 
the system 
\begin{align}
	\label{vlasov hartree}
	\begin{cases} 
		& 		(\partial_t  + p \cdot  \nabla_x + F_B(t,x)  \cdot  \nabla_p   ) f  = 0 	\\
		&		  i \partial_t \vp 
		=   
		- \frac{1 }{2} \Delta   \vp +   ( V*\rho_F   )   \vp 	\\ 
		&  (f,\vp) ( 0 )  = (f_ 0 , \vp_0)  
		\in L^1_+(\R^{2d})\times L^2(\R^d) 
	\end{cases} 
\end{align}
where $F_B(t,x) \equiv 
- \int\nabla V(x-y ) | \vp(t,y)|^2 \d y 
$
and
$\rho_F 	(t,x) \equiv  \int_{	\R^d	} f (t,x,p ) \d p  	. $
We shall refer to \eqref{vlasov hartree} as the {Vlasov-Hartree equation.}

\vspace{3mm}

\textbf{Distances}. 
As we have mentioned previously, 
we are interested in the case in which the initial system is  
at zero temperature. 
In this situation,  
Thomas-Fermi theory suggests that  
one cannot expect  the initial data of 
the classical fermion subsystem to     higher regularity than a characteristic function (see for reference
\cite{Fournais2018}).
In what follows, we introduce distances
which we will use throughout this article. 
In particular, they will be necessary in our analysis of convergence to the Vlasov-Hartree system 
in this context of ``low regularity''.

\begin{itemize}[leftmargin=*, label=$\square$]
	\item \textit{Wasserstein distance}. 
	Given $n \in \mathbb N$, we denote the  $n$-th Wasserstein distance between probability measures 
	$\mu,\nu \in \mathscr{P}_n (\R^{2d})$ by 
	\begin{equation}
		W_n (\mu , \nu )
		\equiv 
		\Big( 
		\inf_{ \pi } 
		\int_{\R^{2d} \times \R^{2d}}  | z - z ' |^n  \pi( \d z \otimes  \d z  ' )
		\Big)^{\frac{1}{n}}
	\end{equation}
	where the inifimum is taken over all couplings of $\mu$ and $\nu$, i.e.
	probability measures $\pi \in \mathscr{P}(\R^{2d} \times \R^{2d})$ 
	with  first marginal $\mu$, and second marginal $\nu$.

	\item \textit{Fourier-based norms}.
	Given $s \in \R$, 
	and $ g : \R^{2d} \rightarrow \C $
	we introduce the following 
	Fourier-based norm 
	\begin{equation}
		\label{fourier based norm definition}
		| g |_s
		\equiv 
		\sup_{ \zeta \in \R^{2d}} 
		(1 +  |\zeta|)^{-s} \ |  \hat g ( \zeta  )  |   \  . 
	\end{equation}
	In applications, we take $s \geq 0 $. Hence,  we also regard $ |\cdot|_s$ 
	as a negative Sobolev norm. 
\end{itemize}

\vspace{2mm}

Our assumption for the initial data now reads. 
\begin{assumption}[Hartree initial data]
	\label{assumption 2}
	The pair $(\omega_0^\hbar ,\vp_0^\hbar) \in \mathscr{L}^1 (L^2(\R^d))\times L^2(\R^d)$ satisfies the following conditions. 
	\begin{enumerate}[leftmargin=*] 
		\item 
		$\omega_0^\hbar \in \mathscr{L}^1 ( L^2(\R^d)  )$
		satisfies $0  \leq \omega_0^\hbar  = (\omega_0^\hbar)^* \leq 1$,  $\t{Tr} \omega_0^\hbar = M $
		and
		$\tr \omega_0^\hbar( -\hbar^2 \Delta + x^2  ) < \infty  $. 
		Further, 	we assume that there exists 
		a real-valued 
		$      f_0 \in L^1   (\R^{2d})$ 
		such that :
		\begin{enumerate}[leftmargin = 1.3cm]
			\vspace{1mm}
			\item[(1.1)] $ 0 \leq f_0(x,p) \leq 1$
			and
			$\int_{	\R^{2d}	} f_0 (x,p ) \d x \d p  = 1$. 
			\vspace{1mm}
			\item[(1.2)] 
			There are finite second moments: 
			$f_0 \in \mathscr P _2 (\R^{2d})$. 
			\vspace{1mm}
			\item[(1.3)]     
			$
			\lim_{\hbar \downarrow 0} 
			| f _0 - f_0^\hbar |_{1}
			= 0 
			$, 
			where  $f_0^\hbar = W^\hbar [\omega_0^\hbar]$ . 
		\end{enumerate}
		\vspace{2mm}
		\item 
		There exists $\vp_0 \in L^2  (\R^d)$
		with $\| \vp_0	\|_{L^2}=1 $ 
		such that 
		$  \lim_{\hbar \downarrow 0} \|   \vp_0 - \vp_0^\hbar		\|_{L^2}=0.$ 
	\end{enumerate}
\end{assumption}

\begin{remark}[Low regularity]
	Let us note that in Assumption \ref{assumption 2}
	there are \textit{no} regularity requirements
	on the limits of the sequence of initial data. 
	Of course, this comes with a price. 
	First, we shall need  the interaction potential to be  at least $V \in C^{1,1} (\R^d;\R)$.
	Second, the metric that we use to measure the distances between fermion densities is rather weak.
	Namely, it involves testing over functions $h (x,p)$
	for which both the integrals  
	$ \int_{	\R^{2d}	}  \< \zeta\>^2 |\hat h (\zeta)|^2 \d \zeta  $
	and
	$\int_{	\R^{2d}	} |\zeta | \,  |\hat h (\zeta)| \d \zeta $
	are finite. 
	Third--compared to similar results in the literature--we need
	two moments in phase space, rather than only one. 
\end{remark}

	\begin{remark}[Fermion mode of convergence]
		In Assumption \ref{assumption 2} we require  that  $f_0^\hbar \rightarrow f_0$
		with respect to the negative Sobolev norm $|\cdot|_1$. 
		Some comments are in order. 
		
		\vspace{1mm}
		
		\textit{(i)}
		This assumption  on the initial data
		can be verified 
		for examples  of interacting Fermi gases, 
		that arise 
		as minimizers of variational problems in the presence of an external trap $U_{ \t{ext} }(x)$. 
		More precisely, let  $\omega_0^\hbar $ 
		be  the one-particle reduced density matrix of an approximate ground state $\psi_F^\hbar $ of the   minimization problem 
		\begin{equation}
			E (M) \, = \,  
			\t{inf}_{\bigwedge_{i=1}^M L^2(\R^d )}  
			\ 
			\sigma 
			\bigg[  
			\sum_{i=1}^M  - \hbar^2  \Delta_{x_i}    + U_{\t{ext}}   (x_i)  
			+ 
			\frac{1}{N} \sum_{ i< j} V(x_i - x_j)
			\bigg]  \  , \qquad \hbar = M^{ -1 /d } \ . 
		\end{equation}   
Regarded as a trial state,  
		 we  may assume that  $\omega_0^\hbar$ is an approximate minimizer of the associated
		Hartree-Fock problem 
		\begin{align}
			E_{HF}( M ) 
			= 
			\inf
			\Big\{
			\tr  ( \omega [ - \hbar^2 \Delta + (\rho^\omega  * V) + U_{\t{ext}} (x) ] )  
			\ :  \ 
			\omega = \omega^2 = \omega^*  , \ \tr \omega = M 
			\Big\}  \ , 
		\end{align}  
		where $\rho^\omega(x) = N^{-1} \omega(x;x)$. 
		Thus, 
		$\omega_0^\hbar$
		can be assumed to be 
		an orthogonal projection, \textit{i.e.}
		$\omega_0^\hbar = (\omega_0^\hbar)^2$, 
		which  is equivalent to $\psi_F^\hbar$ being 
		a Slater determinant. 
		It has been proven in  \cite[Theorem 1.2]{Fournais2018}
		that,   as $\hbar \downarrow0$ (and, up to extraction of a subsequence), 
		the Wigner transform $f_0^\hbar = W^\hbar[\omega_0^\hbar]$
		converges   in a weak sense to  the function  $f_0 (x,p) = \1 (     p^2 \leq C_{TF}\rho(x)^{2/d}  )$. 
		Here,   
		$\rho$ is the minimizer of  the associated Thomas-Fermi problem
		\begin{align*}
 	\mathcal E (\rho)  
			 = \frac{d C_{\rm TF}}{d+2} \int_{\R^d } \rho(x)^{1 + \frac{2}{d}} \d x 
			+  \frac{1}{2}\int_{\R^{2d} } \rho(x) V(x-y)\rho(y) \d x \d y 
			+ \int_{\R^d }U_{\t{ext}}  (x)    \rho (x) \d x 
		\end{align*} 
		with constraints $ \rho (x) \geq 0 $, $\int \rho(x) dx =1 $ and $\rho \in L^1 \cap L^{1 + \frac{2}{d}}$
		and $C_{\rm TF} = 4\pi^2 |B_{\R^d}(0,1)|^{-2/d}$ is the Thomas-Fermi constant. 
		In particular, 
		the $\hbar\downarrow0 $ convergence   
		can be shown to hold  
		with respect to the negative Sobolev norm  \eqref{fourier based norm definition} as well. 
		In other words,   there holds $\lim_{\hbar\downarrow 0 }   |        f_0^\hbar - f_0   |_1  =0 $. 
		See e.g. \cite{Cardenas2023}. 
		
		\vspace{1mm}
		
		Additionally, it has been proven \cite[Theorem 3]{Lafleche2023(2)}
		that the Husimi transform $\widetilde{f}^\hbar $ of the orthogonal projection 
		\begin{equation}
			\omega^\hbar \equiv \1_{ |p|^2 \leq U(x)}    \qquad \t{where} \qquad  p = - i \hbar \nabla \ 
		\end{equation}
		converges locally in the Sobolev space $W^{s,p}$ for all $s < 1/p $,
		to the classical distribution $f (x,p ) = \1 (|p|^2 \leq U(x)$, 
		provided
		$U \in ( L^\infty \cap L^{\frac{d}{2}}) (\R^d) $ is a nice enough external potential. 
		This  local convergence result can be extended to $\R^d$ by means of Agmon-type estimates, and then to the convergence of the Wigner transform $f^\hbar$ with 
		respect to the negative Sobolev norm $|\cdot|_1$. The authors are very grateful to L. Lafleche for   his comments in this regard.

		\vspace{2mm}
		\textit{(ii)}
		The above discussion  should be compared with the
		$L^1$-norm convergence   considered in 
		\cite[Theorem 2.5]{BenedikterPortaSaffirioSchlein2016} 
		for the  initial data, in the context of interacting Fermi gases. 
		While their conclusion is strictly stronger--that is, stronger mode of convergence--to the authors best knowledge 
		the only examples in $\R^d$ for which the $L^1$ convergence  has been verified correspond to coherent states.
		Unfortunately,
		these are not examples  of  zero temperature states 
		(\textit{i.e.} orthogonal projections). 
		We believe that there is value in our approach 
		since--as the examples considered in \textit{(i)} arise form orthogonal projections--we are able to put together Theorem \ref{thm1}
		and our next result.
		In particular, with this approach   we obtain  a 
		\textit{quantitative} convergence 
		from the Schr\"odinger to the Vlasov-Hartree dynamics, 
		in the situation of low regularity or--put differently--the zero temperature situation.

	\end{remark}

Our main result concerning the semi-classical limit of the coupled Hartree equations is the following theorem.

\begin{theorem}
	[The semi-classical limit]
	\label{thm2}
	Assume that the interaction potential satisfies
	$$
	\int_{ \R^d 	}	
	\<\xi\>^2 |\hat V (\xi)	| \d \xi < \infty . 
	$$
	Let us consider the following: 
	\begin{itemize}[leftmargin=0.9cm, label=$\square$]
		\item 	Let $(\omega^\hbar , \vp^\hbar)$
		be the solution of the Hartree-Hartree  eq. 
		\eqref{hartree hartree},
		with initial data $(	 \omega_0^\hbar , \vp_0^\hbar	)$ satisfying Assumption \ref{assumption 2}.  
		Denote by 		  $f^\hbar (t  ) = W^\hbar[\omega^\hbar (t)]$ 
		its Wigner transform.
		\vspace{1mm}
		\item 
		Let $(f, \vp)$ be the solution to the Vlasov-Hartree 
		system     \eqref{vlasov hartree},
		with initial data $(f_0, \vp_0). $
	\end{itemize}
	Then,  there exists $C>0$ such that   
	for all times
	$ t \in \R $
	and test functions $h : \R^{2d} \rightarrow \C$, 
	the following inequalities hold true 
	\begin{align}
		\nonumber 
		|	 \<   h , ( f_t   -  f^\hbar_t  )  \>		|	 
		& \leq 
		C_2(t) 
		\|	  \<\zeta\>  \hat h  	\|_{L^1}
		\Big(
		|	   f^\hbar_0  -   f_0  			|_1 
		+ 
		\hbar 
		\Big)
		+ 
		C_1(t) 
		\|	   |\zeta|  \hat h  	\|_{L^2}
		\Big(
		\|	  \vp_0 -   \vp_0^\hbar 	\|_{L^2} 
		+ 		\hbar^{1/2} 
		\Big)   \, \\ 
		\| \vp_t    -  \vp^\hbar_t		\|_{L^2 }
		& 	\leq 
		C_2(t) 
		\Big(
		|	   f^\hbar_0  -   f_0  			|_1 
		+ \hbar 
		\Big) 
		+
		C_1 (t) 
		\Big(
		\|	  \vp_0 -   \vp_0^\hbar 	\|_{L^2} 
		+
		\hbar^{1/2} 
		\Big) \  . 
	\end{align}
	Here, we are denoting 
	$C_1(t) =
	C\exp(C t^2 )
	$
	and
	$C_2(t) = C\exp ( C \exp  C  | t |  ) $. 
\end{theorem}

\begin{remark}[Convergence rates]
	The above result gives an explicit convergence rate from the 
	Hartree-Hartree to the Vlasov-Hartree dynamics.
	Of course, this is  not the optimal convergence rate in $\hbar$, which
	we believe should be $\mathcal O (\hbar)$
	on the right hand side.  
	In this work, we have not tried to optimize this rate.
	Indeed, our main goal was to \textit{identify}
	the leading order equations that drive the effective dynamics 
	of the Bose-Fermi mixture, 
	which Theorem \ref{thm2} appropriately does. 
	In a similar spirit, we have not tried to optimize the growth-in-time
	of the constants  involved in our estimates. 
\end{remark}

\begin{remark}[A variational norm]
	Here we have formulated our theorem in terms of test functions.
	Alternatively,  it can be formulated in terms of the norm 
	\begin{equation}
		\|		f \| \equiv \sup 
		\big\{ 
		\<  h , f \> 
		: 
		h\in \mathscr{S} (\R^{2d}) 
		, 
		\|   \< \zeta \> \hat h		\|_{L^1} \leq 1 
		\	    \t{ and } \ 
		\|    | \zeta| \,  \hat h		\|_{L^2} \leq 1 
		\big\}  , 
	\end{equation}
	which is 
	strictly weaker than the norms 
	$
	\|		\cdot \|_{\dot H ^{-1}}$
	and 
	$ | \cdot |_{1 }$. 
\end{remark}

\subsection{Strategy of the proofs}
\label{subsection strategy of proofs}
Let us  outline the proofs of our main results, Theorem \ref{thm1} and \ref{thm2}. 

\vspace{2mm}

The proof of Theorem \ref{thm1} 
consists of the study of an appropriate fluctuation dynamics.
For  gases of interacting bosons 
the approach was  first carried out in \cite{RodnianskiSchlein2009}, 
whereas for gases of interacting fermions  
the approach  was employed  in \cite{BenedikterPortaSchlein2014}. 
The difficulty  of tackling the Bose-Fermi 
mixture lies in how to properly combine these two approaches.
In the present paper, we adapt the approach of studying fluctuation dynamics for the problem at hand. Namely, we introduce in Section \ref{section second quantization}
the formalism of Second Quantization on Fock space $\F$. 
In this formalism, 
coherent states describing Bose-Einstein condensates 
are parametrized by a Weyl operator $\W[\sqrt N \vp(t)]$, 
whereas degenerate Fermi gases are implemented 
by a particle-hole transformation $\calR[\omega(t)]$; see Sections \ref{subsubsection PH}  and \ref{subsubsection CH} for more details. 
In Section \ref{section fluctuation} 
we then study the dynamics of \textit{fluctuations}
around the tensor product of these states.
Roughly speaking,  the problem  is then reduced 
to estimating  
the ``number of excitations" outside of $(\vp(t), \omega(t))$. 
We implement this point of view by introducing a new  unitary transformation on $\F$, 
denoted by $\U(t,s)$ and defined in \eqref{unitary U}. 
Its understanding is  fundamental in our analysis, and leads to the number estimates contained in Theorem \ref{prop number estimates}. 
The proof of these estimates
is based on the analysis of  its infinitesimal generator, which has the form (see Lemma \ref{lemma infinitesimal generator} for details)
\begin{equation}
	\L(t) 
	=
	\d \Gamma_F [h_F(t) ] \otimes \1 
	+
	\1 \otimes \d \Gamma_B [h_B(t)]
	+ \lambda \sqrt N \L_{2,1}(t) + \lambda \L_{2,2}(t) 
\end{equation}
Here, the difficulty lies in controlling the terms $\L_{2,1} (t)$ and $\L_{2,2}(t)$,	which do not commute with
particle number operators and can potentially increment the number of fluctuations.

The proof of Theorem \ref{thm2}  is essentially divided in two steps.
First, we rely  on techniques developed in 
\cite{BenedikterPortaSaffirioSchlein2016}
to understand the stability of the Hartree-Hartree equation \eqref{hartree hartree}
with respect  to the metric in $\mathscr{L}^1( L^2( \R^d)) $
that is induced by the norm $|\cdot |_s$, 
defined in \eqref{fourier based norm definition}. 
Second, 
using some recently developed tools 
from Quantum Optimal Transportation, 
we are able to 
show that 
the convergence from the Hartree-Hartree  \eqref{hartree hartree}
to the Vlasov-Hartree \eqref{vlasov hartree} dynamics 
can be controlled in   
the negative Sobolev space $\dot H ^{-1}$. 
These tools include the introduction in 
\cite{GolseMouhotPaul2016}
of the quantum analogue of the classical Wasserstein distance between two probability measures, 
which have been later developed and applied to the analysis of single species many-particle systems in a series of papers--see for instance
\cite{GolsePaul2017,GolsePaul2021,GolsePaul2022,Lafleche2023}. 
One of the main advantages 
of   these techniques is    the fact that they require \textit{no} regularity
on the initial data under consideration.
This is compatible with Assumption \ref{assumption 2}, in which we assume our initial data corresponds
to a zero temperature state--for fermionic systems, 
the $\hbar\downarrow 0 $ limit of the  Wigner function of 
these states is expected to be of Thomas-Fermi type, 
which fails  for instance to be in $W^{1,1}$.


\section{Second Quantization I: Preliminaries}
\label{section second quantization}
It is convenient to study the Hamiltonian $\eqref{hamiltonian HNM}$
in the second quantization formalism. 
Here, 
we allow the number of particles to fluctuate 
and thus consider the 
Hilbert space composed of the corresponding Fock spaces.
Namely, we let 
\begin{equation}\label{fock space definition}
	\F \equiv \F_F \otimes \F_B\ , 
	\quad \t{where} \quad 
	\mathscr{F}_F \equiv \C \oplus \bigoplus_{n=1}^\infty 
	L^2_a( \R^{dn} )
	\quad \t{and} \quad 
	\mathscr{F}_B \equiv \C \oplus \bigoplus_{n=1}^\infty 
	L^2_s(\R^{dn})  
\end{equation}
are the fermionic and bosonic Fock spaces, respectively. 
Here,  $L^2_a $ 
and  $ L^2_s  $ 
correspond to $L^2$-functions  that are antisymmetric and symmetric, respectively, with respect to permutations of  the particle   variables. 
On $\F_F$ we introduce \textit{fermionic} creation- and annihilation operators $a_x$ and $a_x^*$ 
as the operator-valued distributions 
that satisfy the Canonical Anticommutation Relations (CAR)
\begin{equation}
	\{  a_x , a_{x'}^*	\} = \delta(x- x '  )   
	\qquad
	\t{and}
	\qquad
	\{ a_x ,a_{x' } \}  
	=
	\{ a_x^* ,a_{x' }^*  \}  =0 
	, 
	\quad x,x ' \in \R^d  \,  , 
\end{equation}
where $ \{ \cdot, \cdot \} $ is the anticommutator bracket.
Similarly, on $\F_B$ we introduce \textit{bosonic} creation- and annihilation operators
$b_y$ and $b_y^*$
as the operator-valued distributions 
that satisfy the Canonical Commutation Relations (CCR)
\begin{equation}
	[b_y,  b_{y ' }^*]= \delta( y - y'  )   
	\qquad
	\t{and}
	\qquad
	[  b_y  ,b_{y' }   ] 
	=
	[  b_y^* ,b_{y'}^*  ]   =0
	, 
	\quad y,y ' \in \R^d  \,  , 
\end{equation}
where now $[ \cdot  , \cdot ] $ is the commutator bracket. 
We shall denote by 
$\Omega_F = (1, \0)$ 
and 
$\Omega_B = (1,\0 )$ 
the vacuum vector in each space, 
and by $\Omega \equiv \Omega_F \otimes \Omega_B$ 
the vacuum state of the combined system.  

\vspace{2mm}

In this setting, the  many-particle Hamiltonian 
introduced in the previous section
can be written in terms of
creation- and annihilation-
operators 
in the following form 
\begin{equation}
	\label{hamiltonian}
	\H = 
	\frac{\hbar^2 }{2m_F}
	\int_{  \R^d }
	a_x^* (  - \Delta_x )
	a_x \d x  
	+ 
	\frac{\hbar^2 }{2m_B}
	\int_{  \R^d }
	b_y^* (  - \Delta_y )
	b_y \d y  
	+ \lambda 
	\int_{ \R^{2d} } 
	V(x-y) \ 
	a_x^*a_x   b_y^* b_y 
	\ \d x \d y  \ ,
\end{equation}
where we do not display explicitly the tensor product symbols. 
As for the dynamics, 
we introduce 
the time evolution
of the second quantized system 
as 
\begin{equation}
	\Psi(t) 
	\equiv 
	\exp\big(
	- i t \H / \hbar
	\big)
	\Psi(0)  \ , \qquad \forall t \in \R \ .
\end{equation}
Since the Hamiltonian $\H$
is quadratic and diagonal in creation-
and annihilation operators, 
it commutes 
with the \textit{fermionic}
and \textit{bosonic}
number operators
\begin{equation}
	\N_F 
	\equiv 
	\int_{	\R^d	} a_x^* a_x 
	\d x 
	\qquad
	\t{and}
	\qquad 
	\N_B 
	\equiv 
	\int_{	\R^d	} 
	b_y^* b_y \d y \ .  
\end{equation}
Consequently, if 
$\Psi(0)^{(n,m)} 
=
\delta_{n,N} \delta_{m,M} \psi(0)$, 
then for all $ t \in \R $ it holds that 
\begin{equation}
	\Psi(t)^{(n,m)}
	= 
	\delta_{n,N} \delta_{m,M} \psi(t) \   ,
\end{equation} 
where $\psi (t)$
is the state
corresponding to the 
$(N+M)$-particle system, defined in \eqref{wave function}. 
Most importantly,  one may verify that
the following relations hold true
for the one-particle reduced density matrices
\begin{equation}
	\label{density matrices second quant}
	\gamma_F(t ; x_1 ,x_2)
	= 
	\< 
	\Psi (t) , 
	a_{x_2}^* a_{x_1}\otimes
	\Psi(t) 
	\>_\F    
	\quad
	\t{and}
	\quad
	\gamma_B (t ; y_1 ,y_2)
	= 
	\< 
	\Psi (t) , 
	\1 \otimes 
	b_{y_2}^* b_{y_1}
	\Psi(t) 
	\>_\F  
	\   . 
\end{equation}
The equations given in \eqref{density matrices second quant}
are  the starting point in the proof of Theorem \ref{thm1}. 

\vspace{2mm}

In the rest of this section, 
we introduce   preliminaries  
that we will need to prove Theorem \ref{thm1}.
Namely, 
In Subsections \ref{subsection fermions} and \ref{subsection fermions}
we give a more detailed account   of the second quantization formalism for both fermions and bosons. 
Our goal here is not   to be thorough, but to collect basic results and fix the notation that will be used throughout the article. 
The reader is referred to the book \cite{BenedikterPortaSchlein2016}
and the lecture notes 
\cite{Solovej2007}
for more details.

\subsection{Fermions}
\label{subsection fermions}
Throughout this subsection, we will write the Fermionic Fock space 
$\F_F$  as follows 
\begin{equation}
	\mathscr{F}_F
	= 
	\bigoplus_{n=0}^\infty 
	\mathscr{F}_F^{(n)}
	\quad
	\t{where}
	\quad 
	\F_F^{(0)} \equiv \C 
	\quad 	 
	\t{and}
	\quad 
	\mathscr{F}_F^{(n)} \equiv 	 L^2_a( \R^{dn} ) , \ \forall n \geq 1  , 
\end{equation}
where  $	L^2_a( \R^{dn} )  $ 
corresponds 
to the space of $L^2$ functions that are    antisymmetric 
with respect to the permutation of the particles position variables.
The space $\F_F$  becomes a Hilbert space when endowed  with the inner product 
\begin{equation}
	\< \Psi_1 ,  \Psi_2 \>_{\F_F}
	\equiv 
	\sum_{n=0}^\infty  
	\langle  
	\, 
	\Psi_1^{(n) } , \Psi_2^{(n)}  
	\,
	\rangle_{\F_F^{(n)}}  \  , 
	\qquad \forall \Psi_1 ,\Psi_2 \in \F_F \ . 
\end{equation} 
On the Fock space $\F_F$ one introduces the smeared-out
creation- and annihilation operators  
as follows. 
Given $f \in L^2(\R^d) $,  
we let   $a^*(f)$ and $a(f)$
be defined   for  $\Psi  \in \F_F$ as 
\begin{align}
	\big(		 a^*(f)  \Psi \big)^{(n)}
	(x_1, \ldots, x_n) 
	& \equiv 
	\frac{1}{\sqrt n }
	\sum_{i=1}^n 
	(-1)^i 
	f( x_i) \Psi^{(n-1)}
	(x_1 , \ldots, x_{i-1} , x_{i+1} , \ldots, x_n )	\ , \\ 
	\big(	 a  (f)	 \Psi \big)^{ (n)}
	(x_1, \ldots, x_n) 
	& \equiv 
	\sqrt{n+1 }
	\int_{\R^d }  \overline{f(x)} \Psi^{(n+1)} (x , x_1 , \ldots, x_n) \d x  \ . 
\end{align}
In particular, they satisfy
the following version of the CAR 
\begin{equation}
	\{ a(f)  , a^*(g) \}    = \< f,g\>_{L^2} 
	\quad
	\t{and}
	\quad 
	\{	  a^\# (f) , a^\# (g)		 \} = 0 
\end{equation}
where we recall $\{ \cdot, \cdot  \}$ stands for the anticommutator.
In particular,  it is easy to see that 
the CAR   turns them into bounded operators, with norms
\begin{equation}
	\|	  a^*(f)	\|_{B (\F_F)} = 
	\|	  a(f)	\|_{B (\F_F)}  = 
	\|	 f\|_{L^2 } \ . 
\end{equation}
Let us finally mention that the connection with the operator-valued distributions $a_x^*$ and $a_x$ is by means of the formulae  
\begin{equation}
	a^*(f) = \int_{\R^d }  f(x)  a_x^* \d x 
	\qquad
	\t{and}
	\qquad 
	a(f) = \int_{	\R^d	} \overline{f(x)} a_x \d x   \  , 
	\qquad f \in L^2(\R^dw)
\end{equation}

\subsubsection{Fermionic operators on $\F_F$}
Given a closed linear operator 
$\O: \calD (\O) \subset L^2(\R^d ) \rightarrow L^2(\R^d )$, 
we consider its second quantization  
$\d \Gamma_F [\O]$ 
as the diagonal operator 
on $\F_F$, defined as 
\begin{equation}
	( 	\d\Gamma_F [\O]  )^{ (n)}
	\equiv 
	\sum_{ i = 1 }^n 
	\mathbb{1}^{ i -1 } \otimes \O \otimes \mathbb{1}^{ n - i }  , 
	\ n \geq 1 
	\qquad
	\t{and}
	\qquad 
	( 	 	\d\Gamma[\O]  )^{ (0)} \equiv 0  \   , 
\end{equation}
initially   on tensor products of elements of $\calD(\O)$, and then closed. 
In most cases of interest $\O$ is bounded (or even trace-class)
and the domain of $ \d \Gamma_F [\O]$ is contained in $\calD(\N_F)$--the only exception will be the  Laplacian $\O = - \Delta$, in which case 
$\hbar^2/2m_F \d \Gamma_F[-\Delta]$ is the 
associated kinetic energy of the system.

\vspace{1mm}

The reader should be aware that, at least formally, 
if   $\O$ has an operator kernel
$\O(x,x')$, then  one may write in terms of creation- and annihilation operators 
\begin{equation}
	\d\Gamma_F  [\O ] = \int_{	\R^{2d}	} \O (x,x')      a_{x}^*a_{x'}     \d x \d x' \ . 
\end{equation}
In this context, one of the most important observables in second quantization corresponds to the 
the {fermionic number operator}.
It is defined as   the second quantization of the 
the identity operator  on $L^2(\R^d)$, which
has the distributional  kernel
$\1 (x,x') = \delta(x-x')$.
Namely,    
\begin{equation}
	\label{fermion identity N}
	\N_F =  \bigoplus_{n=0}^\infty  n = \d \Gamma_F [\1 ] = \int_{	\R^d	} a_x^*a_x \d x  \ . 
\end{equation}


Let us now  collect some basic results concerning estimates
for the second quantization of operators in fermionic Fock space
in the following lemma. 
For a proof, we refer the reader to 
\cite[Lemma 3.1]{BenedikterPortaSchlein2014}. 

\begin{lemma}[Estimates for fermionic operators]
	\label{lemma fermion estimates}
	Let $\O: L^2(\R^d) \rightarrow L^2(\R^d)$ 	 be a bounded operator. 
	Then, the following holds true.

	\begin{enumerate}[leftmargin=1cm]
		\item 
		For all $\Psi , \Phi \in \calD(\N_F)$  
		\begin{align}
			\|	 \d \Gamma_F (    \O	  	) \Psi	 \|_{\F_F}			&	
			\leq	\|	   \O  	\|_{B(  L^2 )} 	
			\|	  \N_F  \Psi 	\|_{\F_F}		 \ , \\ 
			| 		\< \Psi , \d  \Gamma_F  (	 \O 	)\Phi \>_{\F_F}	|   
			& 
			\leq	\|	   \O  	\|_{B(  L^2 )} 	 
			\| 			\N_F^{\frac{1}{2}}	 		\Psi 	\|_{\F_F}
			\| 				\N_F^{\frac{1}{2}}				 \Phi 	\|_{\F_F} \ . 
		\end{align}
		\item 
		If $\O$ is Hilbert-Schmidt   with kernel $\O(x,y)$, then for all $\Psi \in \calD (\N_F^{\frac{1}{2}})$
		\begin{align}
			\|		 \d \Gamma_F (\O)	 \Psi 	\|_{\F_F}		& \leq 		\|	 \O	\|_{HS}		\|		 \N_F^{\frac{1}{2}}	\Psi 	\|_{\F_F}		\\
			\|		\int_{	 \R^d \times \R^d 	}	 \O(x_1 , x_2 )  a_{x_1} a_{x_2} \d x_1 \d x_2  		\|_{\F_F}				 & \leq 		 \|	 \O	\|_{HS}		\|		 \N_F^{\frac{1}{2}}	\Psi 	\|_{\F_F}		\\
			\|		\int_{	 \R^d \times \R^d 	}	 \O(x_1 , x_2 )  a^*_{x_1} a^*_{x_2} \d x_1 \d x_2  		\|_{\F_F}				 & \leq 		 \|	 \O	\|_{HS}		\|	( 	 \N_F+\1)  ^{\frac{1}{2}}	\Psi 	\|_{\F_F}	 \ . 
		\end{align}
		\item 
		If $\O$ is trace-class with kernel $\O(x,y)$, then for all $\Psi \in \F_F$
		\begin{align}
			\|		 \d \Gamma_F  (\O)	 \Psi 	\|_{\F_F}		& \leq 		\|	 \O	\|_{Tr}		\|		 	\Psi 	\|_{\F_F}			\\
			\|		\int_{	 \R^d \times \R^d 	}	 \O(x_1 , x_2 )  a_{x_1} a_{x_2} \d x_1 \d x_2  		\|_{\F_F}			 & \leq 			\|	 \O	\|_{Tr}		\|		 	\Psi 	\|_{\F_F}			\\
			\|		\int_{	 \R^d \times \R^d 	}	 \O(x_1 , x_2 )  a^*_{x_1} a^*_{x_2} \d x_1 \d x_2  		\|_{\F_F}			 & \leq 			\|	 \O	\|_{Tr}		\|		 	\Psi 	\|_{\F_F}		 \ . 
		\end{align}
	\end{enumerate}
\end{lemma}

\subsubsection{Particle-hole transformation}
\label{subsubsection PH}

In this subsection we introduce a class of 
Bogoliubov transformations 
on Fock space 
that we will use in the proof of Theorem \ref{thm1}; 
they will be useful in
quantifying 
the number of fluctuations
outside of a degenerate Fermi gas.

\vspace{2mm}
More precisely, 
let us consider 
$\omega$
to be a   rank-$M$ orthogonal projection   
on $L^2(\R^d)$.
Thus, there exists an orthonormal basis 
$\{  \phi_i \}_{i=1}^\infty \subset L^2(\R^d)$ such that 
\begin{equation}
	\textstyle 
	\omega = \sum_{i=1}^M \ket{\phi_i}		\bra{\phi_i}   \ .
\end{equation} 
We   introduce a       map on Fock space
$\calR [\omega] : \F_F \rightarrow \F_F$, 
which we shall refer to as a 
\textit{particle-hole transformation} associated to  $\omega$.
We define it 
according to its action on creation-
and annihilation
operators as follows 
\begin{equation}
	\label{R omega}
	\calR ^*[\omega] a^*(\phi_i) 	\calR [\omega] 
	\equiv 
	\begin{cases}
		a^* (\phi_i) \, ,  & i \leq M \\
		a(\phi_i)  \, , & i> M 
	\end{cases} \ , 
\end{equation}
and its action on the vacuum 
$	\calR [\omega] \Omega_F
\equiv 
a^*(\phi_1) \cdots a^*(\phi_M) \Omega_F .$
Note that since 
the span 
of vectors of the form
$   a^*(\phi_{i_1})   \cdots a^*(\phi_{i_n}    )\Omega_F$
is dense in $\F_F$,  
the above prescription
completely determines $\calR[\omega]$. 

\vspace{2mm}

Let us now collect additional properties of the map $\calR[\omega]$
In order to state them, we need to introduce some important notation. 
Indeed, we consider the operators 
on the one-particle space
$u,v \in B(L^2(\R^d))$
defined as 
\begin{equation}
	\label{definition u and v}
	\textstyle 
	u \equiv \1 - \omega 
	\qquad
	\t{and}
	\qquad 
	v= \sum_{ i = 1 }^M \ket{\overline{\phi_i}} \bra{\phi_i} \ . 
\end{equation}
The following properties are recorded in the following Lemma.
We refer the reader to \cite{BenedikterPortaSchlein2014} for more details.

\begin{lemma}
	[Properties of $\calR$]
	\label{lemma fermion properties}
	Let $\omega$, $u$, $v$ and $\calR[\omega]$ be as above. 
	Then, the following statements hold true. 
	
	\begin{enumerate}
		\item 
		$\calR[\omega]$ is a unitary transformation on $\F_F$, and $\calR^*[\omega] = \calR[\omega]$.

		\item 
		We denote  $u_y (x) \equiv u(x,y)$   and $v_y(x) \equiv v(x,y)$. 
		Then, for all $x \in \R^d$  
		\begin{align}
			\calR ^*[\omega]   
			\ a^*_x  \ 
			\calR [\omega] 
			& 	= 
			a^* ( u_{x}) 
			+ 
			a( \overline{  v_{ x } } )	  \ ,   \\  
			\calR ^*[\omega] 
			\ a_x  \ 
			\calR [\omega] 
			& 	= 
			a ( u_{ x }) 
			+ 
			a^*( \overline{  v_{ x} } )	 \ . 
		\end{align}
		
		\item 
		For  all $x,y  \in \R^d$  there holds 
		\begin{equation}
			\< 	\calR [\omega]  \Omega_F    ,  \, a_{y}^*  \, a_x \,  	\calR [\omega]   \Omega_F    \>_{\F_F}
			=  \omega (x,y)   \  .
		\end{equation}
		In  words, the one-particle reduced density matrix of 	$\calR [\omega]  \Omega_F $ corresponds to
		$\omega$. 
		
		\item 
		$u^* = u^2 = u \ $
		and
		$ \ v^* = \overline v$.

		\item 
		$u^* u + v^*v = \1  \ $ 
		and 
		$ \ u \overline v = v u = 0 . $
		
	\end{enumerate}
	
\end{lemma}

\begin{remark}
	The unitary map
	$\calR[\omega]$ is an example of the implementation of a Bogoliubov 
	transformation--that is, a map on Fock space that preserves the 
	Canonical Anticommutation Relations. 
	More precisely, consider the maps 
	\begin{equation} 
		\label{nu}
		\nu 
		=
		\begin{pmatrix}
			u & \bar v \\
			v & \bar u
		\end{pmatrix} \ 
		\qquad
		\t{and}
		\qquad 
		A(f,g)  \equiv  a(f) + a^* (\bar g) 
	\end{equation}
	for all $f,g \in L^2(\R^d).$
	In this context, $\nu$ is a \textit{Bogoliubov transformation}.
	Namely,  it holds true that  
	for all $f,f_1 ,f_2 , g \in L^2(\R^d): $
	\begin{align}
		\{  A(\nu (f_1, g_1) , \nu (f_2, g_2) )   \} & = \{ A(f_1, g_1) ,    A(f_2 , g_2)   \}  \ , \\
		A^*(\nu (f,g)    )  & =  A( \nu (\bar g , \bar f))   \, . 
	\end{align}
	Furthermore,   
	$\calR[\omega]$ implements $\nu$
	on the Fock space $\F_F$, 
	in the sense that  
	for all $f,g\in L^2(\R^d)$
	\begin{equation}
		\label{conjugation 1}
		\calR ^* [\omega ]A(f,g) 	\calR  [\omega]  
		=
		A( \nu (f,g) )  
	\end{equation}
\end{remark}

Let us note that while the notion of Bogoliubov transformations is quite general, 
in the physical situation at hand
it is sufficient to consider   particle-hole transformations, 
which has an explicit representation. 
This is because the initial state we consider  is a pure state
$\psi_F(0)$
corresponding to a Slater determinant of $M$ particles.
Consequently, 
its one-particle reduced density matrix is 
a rank-$M$ orthogonal projection. 
The situation is quite different in the positive temperature case, 
when states are mixed, and no longer orthogonal projections.

\subsection{Bosons}
\label{subsection bosons}
Similarly, throughout this subsection
we emply the following notation
to denote the bosonic Fock space $\F_B$
\begin{equation}
	\mathscr{F}_B
	\equiv
	\bigoplus_{n=0}^\infty 
	\mathscr{F}_B^{(n)}
	\quad
	\t{where}
	\quad 
	\F_B^{(0)} \equiv \C 
	\quad 	 
	\t{and}
	\quad 
	\mathscr{F}_B^{(n)} \equiv 	 L^2_s( \R^{dn} ) , \ \forall n \geq 1  
\end{equation}
where $	L^2_s( \R^{dn} )  $ 
corresponds to the subspace of  symmetric functions.
$\F_B$   is a Hilbert space when endowed with the inner product 
\begin{equation}
	\< \Phi_1,  \Phi_2 \>_{\F_B}
	\equiv 
	\sum_{n=0}^\infty  
	\<\Phi_1^{(n) } , \Phi_2^{(n)}    \>_{\F_B^{(n)}}  \  , 
	\qquad \forall \Phi_1 ,\Phi_2 \in \F_B \ . 
\end{equation}
On the bosonic Fock space $\F_B$ one introduces the smeared-out
creation- and annihilation operators  
as follows. 
Given $f \in L^2(\R^d) $,  
we let   $b^*(f)$ and $b(f)$
be defined   for  $\Phi  \in \F_F$ as  
\begin{align}
	\big(		 b^*(f)  \Phi \big)^{(n)}
	(y_1, \ldots, y_n) 
	& \equiv 
	\frac{1}{\sqrt n }
	\sum_{i=1}^n  
	f( y_i) \Phi^{(n-1)}
	(y_1 , \ldots, y_{i-1} , y_{i+1} , \ldots, y_n )	\ , \\ 
	\big(	 b (f)	 \Phi \big)^{ (n)}
	(y_1, \ldots, y_n) 
	& \equiv 
	\sqrt{n+1 }
	\int_{\R^d }  \overline{f(y)} \Phi^{(n+1)} (y , y_1 , \ldots, y_n) \d y  \ . 
\end{align}
In contrast to the fermions, they satisfy the Canonical Commutation Relations (CCR)
\begin{equation}
	[  b(f)  , b^*(g)  ]   = \< f,g\>_{L^2} 
	\quad
	\t{and}
	\quad 
	[   b^\# (f) , b^\# (g)	] = 0  \ . 
\end{equation}
In particular, they are unbounded operators on $\F_B$, but relatively bounded with respect to the bosonic number operator--see Lemma \ref{lemma boson estimates}. 
They are connected to the  operator-valued distributions $b_y^*$ and $b_y$  by means of the formulae   
\begin{equation}
	b^*(f) = \int_{\R^d }  f(y)  b_y^* \d y 
	\qquad
	\t{and}
	\qquad 
	b(f) = \int_{	\R^d	} \overline{f(y)} 
	b_y \d y 
	\  , \qquad f \in L^2(   \R^d )    \ . 
\end{equation}

\subsubsection{Operator estimates} 
We proceed analogously as we did for fermions.
Namely,  given an operator $\O$ in the one-particle space $L^2(\R^d)$, 
we  define its second quantization  $\d \Gamma_B [\O]$ 
acting on $\F_B$    as 
\begin{equation}
	( 	\d\Gamma_B[\O]  )^{ (n)}
	\equiv 
	\sum_{ i = 1 }^n 
	\mathbb{1}^{ i -1 } \otimes \O \otimes \mathbb{1}^{ n - i }  , 
	\ n \geq 1 
	\qquad
	\t{and}
	\qquad 
	( 	 	\d\Gamma[\O]  )^{ (0)} \equiv 0  \ . 
\end{equation}
Similarly, if $\O$ has an operator kernel
$\O(y,y')$,    one may write in terms of creation- and annihilation operators 
\begin{equation}
	\d\Gamma_B [\O ] = \int_{	\R^{2d}	} \O (y,y')      b_{y}^*b_{y'}     \d y \d y' \ . 
\end{equation}

Analogously as we did in the case of fermions,
we now   introduce
the corresponding   relations 
for the  { bosonic number operator}  
\begin{equation}
	\label{boson identity N}
	\N_B  =   \bigoplus_{n=0}^\infty  n 
	=
	\d \Gamma_B [\1 ] = \int_{	\R^d	} b_y^*b_y \d y  \ . 
\end{equation}

Let us now  collect some basic results concerning estimates
for the second quantization of operators in bosonic Fock space. 
For reference, see
\cite[Lemma 3.1]{ChenLeeSchlein2011}.

\begin{lemma}[Estimates for bosonic operators]
	\label{lemma boson estimates}
	Let $\O : L^2(\R^d)\rightarrow L^2(\R^d)$ be a bounded operator, 
	and let $ f  \in L^2(\R^d)$.
	Then, the following holds true

	\vspace{1mm}
	
	\noindent (1) For all $\Phi \in \calD (\N_B^{\frac{1}{2}})$
	\begin{align}
		\|	  b( f )  \Phi		\|_{\F_b}			&  \leq \|	  f 	\|_{L^2}		\|   \N_B^{\frac{1}{2}}  \Phi		\|_{	\F_b	}			\\ 
		\|	  b^*( f )  \Phi		\|_{\F_b}			&  \leq \|	  f 	\|_{L^2}		\|    (  \N_B+1)^{\frac{1}{2}}  \Phi		\|_{	\F_b	}		 \ . 
	\end{align}

	\noindent (2) For all $\Phi \in \calD(\N_B )$
	\begin{align}
		\|	 \d \Gamma_B (\O )	 \Phi	\|_{\F_b} \leq 	  \| \O		\|_{B(L^2)}		\| \N_B	 \Phi		\|_{	 \F_b	} \ . 
	\end{align}
\end{lemma}

\subsubsection{Coherent states}
\label{subsubsection CH}

Analogously as we did for fermions, we introduce a unitary transformation that we shall make use of in the rest of the article. 
Namely, for $f \in L^2(\R^d )$ we introduce the \textit{Weyl operator} as 
\begin{equation}
	\W[f] \equiv \exp \Big(     b (f) - b^*(f)    \Big) : \F_B \rightarrow \F_B \ . 
\end{equation}
Note that since the argument in the exponential is anti self-adjoint, $\W[f]$ is automatically a unitary map. 
Its action on the vacuum vector creates a \textit{coherent state}. 
That is,    thanks to the Baker-Camper-Hausdorff formula, one has 
on $\F_B $
\begin{equation}
	( 	\W[f] \Omega_B)^{(n) }   
	=
	e^{   - \frac{\|	f\|^2  }{2}     }
	\frac{ f^{\otimes n }}{\sqrt{n!}}\  , \quad \forall n \geq  0 \ . 
\end{equation}
In particular, the probability of finding the system with 
$n$ particles is given by 
$ e^{ - \| f \|} \|	 f\|^n / n!  $
which follows a Poisson distribution with parameter $\lambda  = \|	f\|_{L^2}$.

\vspace{1mm}

We collect in the following lemma 
some properties of  the  Weyl operator 
and coherent states.
For more details, we refer the reader to 
\cite[Lemma 2.2]{RodnianskiSchlein2009}

\begin{lemma}[Properties of $\W$]
	\label{lemma properties bosons}
	Let $f\in L^2(\R^d)$, and $\W[f]$ be as above. 
	Then, the following statements hold true 
	\begin{enumerate}
		\item 
		$\W[f]$ is a unitary transformation on $\F_B$, 
		and
		$\W^*[f] = \W[-f]$. 
		
		\item 
		For all $y \in \R^d$
		\begin{align}
			\label{conjugation f}
			\W^*[f] b_y  \W [f]
			& 	=
			b_y +    f(y) 	\ , \\ 
			\W^*[f] b^*_y  \W [f]
			& 				=
			b^*_y  +   \overline{ f(y) } \   . 
		\end{align}

		\item 
		For all $y,y' \in \R^d$
		\begin{equation}
			\< \W[f] \Omega_B , b_{y'}^* b_y \W[f] \Omega_B \>_{\F_B}
			= 
			\overline{ f(y')} f(y) \ .  
		\end{equation}
		In   words, the one-particle reduced density matrix of $\W[f] \Omega_B $
		corresponds to the projector $ \ket{f} \bra{f}.$

\end{enumerate}
\end{lemma}


\section{Second Quantization II: The Fluctuation Dynamics}
\label{section fluctuation}
The main goal of this section is to set up the proof of Theorem \ref{thm1}. 
Namely, we will introduce and study the fluctuation dynamics around a state consisting of a   degenerate Fermi gas, and a Bose-Einstein condensate, evolving in time according to the mean-field equations \eqref{hartree}. 
We prove that the number of fluctuations around this state is small, relative to the numbers $N$ and $M$. 
In the next section, we show that these estimates imply Theorem \ref{thm1}. 

\vspace{2mm}

This approach is nowadays considered standard, 
and has been successfully employed in the derivation of several mean-field equations from many-particle systems. 
The first work to use these techniques in the derivation of the Hartree equation
for bosons is \cite{RodnianskiSchlein2009}, 
while for fermions is \cite{BenedikterPortaSchlein2014}.
Our proofs are heavily inspired by their ideas,
and actually borrow a few estimates. 
Finally, let us also   refer the reader to the book \cite{BenedikterPortaSchlein2016}
for a cohesive treatment of the subject. 

\vspace{2mm}

\textit{A note on domains}. 
In what follows, we will be extensively 
manipulating the number operators $\N_F$ and $\N_B$.
These are positive, unbounded self-adjoint
operators   on the Hilbert space $\mathscr F$, 
with domains 
$\calD(\N_F) = \{ \Psi \in \F : \sum_{n,m} n^2 \| \Psi^{(n,m)} \|^2_{ \F_F^{(n)} \otimes \F_B^{(m)}    }  < \infty  \}$
and
$\calD(\N_B) = \{ \Psi \in \F : \sum_{n,m} m^2 \| \Psi^{(n,m)} \|^2_{ \F_F^{(n)} \otimes \F_B^{(m)}    }  < \infty  \}$.
In order to simplify the exposition, 
we shall 
avoid making reference to the unbounded nature of these operators.
Let us note that there is no risk in doing so. 
Indeed, in applications 
all the states $\Psi \in \mathscr F$
that we manipulate
belong to the intersection $\cap_{k=1 }^\infty \calD (\N_F^k) \cap \calD(\N_B^k)$, 
and  all the dynamics   we consider 
leave the above intersection invariant.

\subsection{Number Estimates}
Throughout this section, we denote by   $(\omega , \vp)$   the  pair of variables that solves 
the Hartree-Hartree equations
\eqref{hartree}.

\vspace{2mm}

Let us   now introduce a fundamental  family of  unitary transformations in our analysis. 
Namely, using the notation from Section \ref{section second quantization}, 
we define the following time-dependent 
particle-hole transformation, and Weyl operator 
\begin{align}
	\label{definition W and R}		
	\mathcal{R}_t 
	\equiv 
	\mathcal{R} [\omega (t) ] : \F_F \rightarrow \F_F 
	\quad 
	\t{and}
	\quad 
	\mathcal{W}_t \equiv 
	\mathcal{W} [\sqrt N \vp(t) ] : \F_B \rightarrow \F_B \ . 
\end{align} 
Both  of these transformations     map
the respective vacuum vectors 
into states in Fock space whose one-particle reduced density matrices
correspond to  $\omega (t)$ for fermions, whereas for bosons
 $  N \ket{\vp(t)} \bra{\vp (t)}$. 
In other words, for all $t \in \R$ and $(x,y)\in \R^{2d}$
\begin{align}
	\< 
	\mathcal{R}_t  \Omega_F , a_y^* a_x 
	\mathcal{R}_t \Omega_F
	\>_{\F_F} 
	= 
	\omega(t ; x,y) 
	\quad
	\t{and}
	\quad 
	\< 
	\mathcal{W}_t  \Omega_b , b_y^* b_x 
	\mathcal{W}_t \Omega_b
	\>_{\F_B}  = 
	N 
	\overline{  \vp(t, y ) }  \vp(t,x) . 
\end{align}

\vspace{2mm}

We proceed to define the two-parameter
family of unitary maps 
$\U(t,s) : \F \rightarrow \F$, which we refer to as  the \textit{fluctuation dynamics}, 
as follows  
\begin{equation}
	\label{unitary U}
	\U(t,s)
	\equiv 
	\big(
	\mathcal{R}_t \otimes \mathcal{W}_t
	\big)^*
	\exp \bigg[ - i \frac{(t-s)}{\hbar} \H   \bigg]
	\big( 
	\mathcal{R}_s \otimes \mathcal{W}_s
	\big)  \ , \qquad t,s \in\R \ . 
\end{equation} 
The fluctuation dynamics measures
how far the many-body  Schr\"odinger dynamics
is from the mean-field variables $(\omega (t),\vp(t))$. 
In order to make this statement precise, 
we 	recall that we have defined on $\F_F$ and $\F_B$ the
fermionic and bosonic number operators, respectively, 
\begin{align}
	\mathcal{N}_F    = 
	\int_{\R^d } a_x^* a_x  \d x 
	\qquad
	\t{and}
	\qquad 
	\mathcal{N}_b 
	= 
	\int_{\R^d }
	b_y^* b_y \d y \  . 
\end{align}
Unless confusion arises, we denote with the same symbols their natural extension to
$\F$. Finally, we introduce   on $\F$  the \textit{total number operator}
\begin{equation}
	\N \equiv \N_F   +  \N_B \ . 
\end{equation}

\vspace{2mm}

The main result of this section is the following theorem.
It contains estimates for the growth-in-time
for the expectations of $\N_F$ and $\N_B$
with respect to the fluctuation dynamics $\U(t,s)$. 

\begin{theorem}[Number estimates]
	\label{prop number estimates}
	Let $(\omega,\vp)$ satisfy \eqref{hartree}, and assume that the assumptions in Theorem \ref{thm1} hold true.  
	Let $\U(t,s)$ be the fluctuation dynamics. 
	Then, the following statements hold true 
	\vspace{1mm}
	
	\begin{enumerate}[leftmargin=.7cm]
		\item 
		For all $\ell,  k \in \mathbb N$
		there is a constant $C>0$ such that    
		for all $\Psi \in \F $ and $t,s\in \R $
		\begin{align}
			\langle      \Psi,  &  \U^*(t,s)    \,   \calN_F^\ell   	 \, \U(t,s)  \Psi	\rangle_\F  \\ 
			\nonumber 
			& 	
			\leq 
			K(t-s ) 
			\Big[
			\|	 \Psi	 \|_\F 
			\| (\N^\ell + \1 )	 \Psi	\|_\F 
			+  
			\Theta_{k,\ell}
			\|  (\1 + M^{-1} \N_F)^\ell    \Psi		\|_\F 
			\|	 	  ( \N +\1)^{k + 3/2 	}    \Psi 		\|_\F   
			\Big]   \  . 
		\end{align}

		\vspace{1mm}
		
		\item 
		For all $\ell,  k \in \mathbb N$
		there is a constant $C>0$ such that    
		for all $\Psi \in \F $ and $t,s\in \R $
		\begin{align}
			\langle   \Psi,  & \U^*(t,s)     \,   \N_B^\ell   \, 	\U(t,s)  \Psi	\rangle_\F 	\\ 
			\nonumber 
			& 
			\leq 
			K(t-s ) 
			\Big[
			\|	 \Psi	 \|_\F
			\| (\N^\ell+ \1 )	 \Psi	\|_\F 
			+  
			\Theta_{k,\ell}
			\|  (\1 + N^{-1} \N_B)^\ell   \Psi		\|_\F 	\|	 	  ( \N+\1)^{k + 3/2  }    \Psi 		\|_\F   
			\Big]  \ . 
		\end{align}
		
	\end{enumerate}
	Here, we denote 
	$K(t)  \equiv  C 	{  \exp     [  C   \lambda  \sqrt{NM/\hbar}\, 
		( 1 + \sqrt{  \hbar M / N}    )	\exp |t|       ]		} $ 
	together with 
	\begin{equation}
		\Theta_{k,\ell} 
		\equiv 
		\frac{\lambda \sqrt N }{\hbar}
		\frac{M^\ell}{(\hbar M)^k} \geq 0 \ . 
	\end{equation}
\end{theorem}

\begin{remark}[Evolution of  the vacuum vector]
	\label{remark control omega}
	In applications, 
	we will consider $\Psi = \Omega$.
	In particular,  there holds  $\| ( \N_F + \1)^\ell\Omega	\|_\F =1 $
	for all $\ell \in \mathbb N$. 
	Thus, given $\ell \in \mathbb N$, we may take $ k = k_\ell$  as in the statement of Theorem \ref{thm1}.
	Hence,  there holds 
	$\Theta_{k,\ell }\leq 1$
	uniformly in the physical parameters $\lambda$, $\hbar$, $N$ and $M$. 
	In this situation one obtains the following estimate, 
	provided one re-updates the constant $C> 0 $
	\begin{equation}
		\|	 (  \N+\1)^\ell  \U(t,s)  \Omega  	\|_{\F} 
		\leq 
		K(t - s)  \ , \qquad \forall t,s\in \R \ . 
	\end{equation}
\end{remark}

\begin{remark}[Boundedness of operators]
	\label{remark boundedness of N}
	An important consequence of      Theorem \ref{prop number estimates}
	is that for all $\ell$, 
	there exists $K_\ell \geq 1$ such that
	for all $\Psi \in \F$
	\begin{equation}
		\|	 (  \N+1)^\ell  \U(t,s)	 \Psi	\|_\F
		\leq
		K(t-s )
		\|   (\N + \1 )^{K_\ell }	 \Psi	\|_{\F} \ , \qquad \forall t,s\in \R  \ . 
	\end{equation}
	Consequently,  
	$ 
	(\N+\1)^\ell \U(t,s) (\N+\1)^{-K_\ell}
	$
	is a bounded linear operator in $\F$, 
	and the same holds for its adjoint. 
	We record this in the following statement
	\begin{equation}
		\|		(\N+\1)^\ell \U(t,s) (\N+\1)^{-K_\ell}	\|_{B( \F)}
		+
		\|		(\N+\1)^{-K_\ell } \U(t,s) (\N+\1)^{\ell }	\|_{B( \F)} 
		\leq K(t-s ) \  . 
	\end{equation}
	Here of course, we have used that 
	$\U^*(t,s) = \U(s,t)$
	and the symmetry
	$K(t-s) = K(s-t). $
	
\end{remark}

We dedicate the rest of this section to the proof of the above Theorem.

\subsection{The infinitesimal generator}
In order to establish the number estimates contained in Theorem \ref{prop number estimates} 
we need to study the time evolution of the fluctuation dynamics. 
To this end, we introduce 
infinitesimal generator $\mathcal L(t) $ of  $\U(t,s) $
as  the time-dependent, 
self-adjoint operator on $\F $
determined by the equation 
\begin{equation} \label{L operator}
	\begin{cases}
		& i \hbar \partial_t \U (t,s )
		= \mathcal L (t) 
		\U (t,s ) \ , 	\\
		& \U (t, t ) = \1 \ . 
	\end{cases} 
\end{equation}
The computation of $\L(t)$ is 
tedious, but can be carried out explicitly. 
Let us record the result of the  calculation in Lemma \ref{lemma infinitesimal generator} below, and postpone the proof to an Appendix. 

\vspace{2mm}

\textit{Notation.}
Before we state the explicit form of the infinitesimal generator, we   introduce 
useful notations. 
Recall that $(\omega(t), \vp(t))$
is a solution of the Hartree-Hartree equation
\eqref{hartree}. 
\begin{enumerate}[leftmargin=0.7cm, label=$\square$]
	\item 
	We denote by
	$h_F(t)$
	and $h_B(t)$
	the following  time-dependent 
	one-particle  Hamiltonians on $H^2(\R^d)$ 
	\begin{align}
		\label{hamiltonians}
		h_F(t)
		& 	\equiv 
		- \frac{\hbar^2 }{2m_F}
		\Delta 
		+ 
		\lambda  N  
		\big(V * \rho_B (t)  \big)    \ ,   \\ 
		h_B(t)
		& 	\equiv 
		- \frac{\hbar^2 }{2m_B}
		\Delta 
		+ 
		\lambda  M   \big(V * \rho_F (t)	\big)    \ . 
	\end{align}
	Here, 
	$\rho_F (t,x) = M^{-1 } \omega(t ; x,x)$
	and 
	$\rho_B(t,x) = |\vp(t,x)|^2$
	are the corresponding fermionic and bosonic position densities. 
	
	\item 
	Upon decomposing $\omega(t) = \sum_{ i = 1 }^M \ket{\phi_i (t)} \bra{\phi_i(t)}$ we denote 
	with the same notation as in Section \ref{section second quantization}, 
	for all $ t \in \R$
	\begin{equation}
		\label{u and v}
		\textstyle
		u (t ) \equiv  \1 - \omega(t) 
		\qquad
		\t{ and }
		\qquad 
		v(t) \equiv \sum_{ i = 1 }^M 
		\ket{  \overline{ \phi_i (t) } } \bra{ \phi_i(t)  }   
	\end{equation} 
	
	\item 
	For fixed 
	$x \in \R^d $ and $ t \in \R$, 
	we denote by $u_{x,t}$ and $v_{x,t}$ 
	the distributions 
	given by  
	\begin{equation}
		u_{t,x} ( y ) = u( t;  y, x)
		\qquad
		\t{and}
		\qquad 
		v_{t,x} (y) = v(t ; y,x) 
	\end{equation}
	for all $x' \in \R^d . $
\end{enumerate}

\begin{lemma}
	\label{lemma infinitesimal generator}
	Let $\U(t,s)$ be the unitary transformation defined in \eqref{unitary U}, and let $\L(t)$ be its inifinitesimal generator. 
	Then, $\L(t)$ admits the following representation (modulo scalars)
	\begin{equation}
		\L(t) 
		=
		\d \Gamma_F [h_F(t) ] \otimes \1 
		+
		\1 \otimes \d \Gamma_B [h_B(t)]
		+ \lambda \sqrt N \L_{2,1}(t) + \lambda \L_{2,2}(t) 
		\qquad
		\forall t \in \R \ . 
	\end{equation}
	Here, $h_F(t)$ and $h_B(t)$ are the one-particle Hamiltonians defined in Eq. \eqref{hamiltonians}.
	The  time-dependent 
	operators  $\L_{2,1} (t) $ and $\L_{2,2} (t) $ 
	are self-adjoint operators on $\F$,
	and 
	are given by the expressions (here, we suppress time labels for convenience)
	\begin{align}
		\nonumber 
		\L_{2,1}   		& = 
		\int_{ \R^d \times \R^d  } v(x - y ) 
		a^*(u_x) a(u_x) \otimes \big(		
		\vp (y) b_y^* +  \overline \vp (y) b_y
		\big)
		\ \d x \d y 	\\
		\nonumber 
		& - 
		\int_{ \R^d \times \R^d  } v(x - y ) a^*( \overline{v _x}) a(\overline{v_x})
		\otimes 
		\big(	
		\vp (y) b_y^* +  \overline \vp (y) b_y
		\big) \ \d x \d y 	\\
		\nonumber 
		& + 
		\int_{ \R^d \times \R^d  } v(x - y ) a^*(  {u _x}) a^*(\overline{v_x}) 	 \otimes 
		\big(	
		\vp (y) b_y^* +  \overline \vp (y) b_y
		\big) \ \d x \d y 	
	 \ 	+ \  h.c \ . 
	\end{align}
	and 
	\begin{align}
		\nonumber 
		\L_{2,2}  
		& = 
		\int_{ \R^d \times \R^d  } v(x - y ) a^*(u_x) a(u_x) 
		\otimes b_y^*b_y
		\ \d x \d y 	\\
		\nonumber 
		& - 
		\int_{ \R^d \times \R^d  } v(x - y ) a^*( \overline{v _x}) a(\overline{v_x}) 
		\otimes b_y^*b_y
		\ \d x \d y 	\\
		\nonumber 
		& +
		\int_{ \R^d \times \R^d  } v(x - y ) a^*(  {u _x}) a^*(\overline{v_x})  \otimes b_y^*b_y
		\ \d x \d y   \ + \  h.c \ . 
	\end{align} 
\end{lemma}
Unfortunately--as it happens often with   similar mean-field theories--the above generator 
is not exactly to prove estimates as the ones contained in Proposition \ref{prop number estimates}. 
We shall introduce a modified generator, 
together with a new \textit{truncated} fluctuation dynamics, 
and then prove that these two are close together.

\subsection{Truncated dynamics}
We define $\widetilde \U (t,s)$
as the strongly continuous, two-parameter family of  unitary operators that solves
\begin{equation}
	\begin{cases}
		& i \hbar \partial_t \widetilde \U (t,s )
		=   \widetilde \L (t) 
		\widetilde 	\U (t,s ) \ , 	\\
		& \widetilde \U (s, s  ) = \1 \  , 
	\end{cases} 
\end{equation}
where the \textit{truncated infinitesimal generator}  $\widetilde  \L(t)$ 
is defined as  follows 
\begin{equation}
	\label{truncated generator}
	\widetilde 	\L(t) 
	=
	\d \Gamma_F [h_F(t) ] \otimes \1 
	+
	\1 \otimes \d \Gamma_B [h_B(t)]
	+ \lambda \sqrt N \widetilde \L_{2,1}(t) + \lambda \L_{2,2}(t) 
	\qquad
	\forall t \in \R \    . 
\end{equation} 
The first, second, and fourth terms in \eqref{truncated generator}
are identical to those found in Lemma \ref{lemma infinitesimal generator}
for $\L(t)$.
However,  for the third term   we have
introduced a cut-off in one of the off-diagonal terms originally found in $\L_{2,1}(t)$: 
\begin{align}
	\nonumber 
	\widetilde	\L_{2,1}   
	& = 
	\1 \otimes \chi(\N_B \leq \hbar M ) 		\int_{ \R^d \times \R^d  } v(x - y ) 
	a^*(u_x) a(u_x) \otimes \big(		
	\vp (y) b_y^* +  \overline \vp (y) b_y
	\big)
	\ \d x \d y 	\\
	\nonumber 
	& - 
	\1 \otimes  		\chi(\N_B \leq \hbar M )
	\int_{ \R^d \times \R^d  } v(x - y ) a^*( \overline{v _x}) a(\overline{v_x})
	\otimes 
	\big(	
	\vp (y) b_y^* +  \overline \vp (y) b_y
	\big) \ \d x \d y 	\\
	\nonumber 
	& + 
	\int_{ \R^d \times \R^d  } v(x - y ) a^*(  {u _x}) a^*(\overline{v_x}) 	 \otimes 
	\big(	
	\vp (y) b_y^* +  \overline \vp (y) b_y
	\big) \ \d x \d y 	\\
	& + 	 
	\int_{ \R^d \times \R^d  } v(x - y ) a(  \overline{v _x}    ) a( u_x  ) 	 \otimes 
	\big(	
	\vp (y) b_y^* +  \overline \vp (y) b_y
	\big) \ \d x \d y  \ . 
\end{align}
The cut-off $\N_B \leq \hbar M $  is specifically tailored for the problem at hand, 
and is introduced so that one can ``close the inequalities''
when   running  a Gr\"onwall-type argument. 

\vspace{1mm}

Let us recall that we have introduced the number operators
$\calN_F$ and $\calN_B$. 
In this subsection, 
we shall give estimates for the growth-in-time
of expectations of these observables with
respect to the truncated fluctuation dynamics $\widetilde \U(t,s).$
The main result of this subsection is the following proposition, containing   
relevant number estimates. 
\begin{proposition}
	[Number estimates for the truncated dynamics]
	\label{prop number estimates truncated}
	Assume that the solution of the mean-field equations $(\omega,\vp)$
	satisfies the following apriori bound
	\begin{eqnarray}
		\label{trace}		
		\|       [e^{i \xi \cdot x}   ,\omega(t)  ]       \|_{\tr} 
		\leq C \exp(C t)  M \hbar \<\xi \> \  \ , \qquad \forall \xi \in \R^d  \ . 
	\end{eqnarray}
	Then, for every $k \in \mathbb N $  there exists a constant 
	$C>0$ such that for all   $t,s \in \R$
	and 
	$\Psi \in\F$  
	there holds 
	\begin{align}
		\langle      \Psi,  \widetilde \U^*(t,s)  \N^k  	  \widetilde 	\U(t,s)  \Psi	\rangle_\F  
		\nonumber 
		& 		\leq 
		C 		 
		{  \exp     \bigg[  C   \lambda  \sqrt \frac{NM}{\hbar } 
			\bigg( 1 + \sqrt{\frac{\hbar M}{  N}}    \bigg)	\exp |t-s|      \bigg]		} 
		\<  \Psi ,  (\calN^k  + \1 ) \Psi	\>_\F \  . 
	\end{align}
\end{proposition}

The proof of the above estimates is based on the Gr\"onwall inequality, and the commutator estimates contained in the following lemma. 

\begin{lemma}[Commutator estimates]
	\label{lemma commutator estimates}
	Under the same assumptions of Proposition \ref{prop number estimates truncated}, 
	for all $\theta \in \R $ there is $C>0$ such that for all $t,s\in\R$   for all $\Psi , \Phi \in \F$
	there holds 
	\begin{align} 
		|	 \langle	\Psi , [  \widetilde \L (t) , \calN_F	] \Phi  	\rangle_\F		|  
		& \leq
		C e^{Ct }
		\lambda
		\Big(  \sqrt{NM \hbar }   + M \hbar \Big) 
		\bigg( 
		\|	 (\N + \1 )^{	 \frac{1 + \theta}{2}	}	 \Psi	\|_{\F}^2
		+ 
		\|	 (\N + \1 )^{	 \frac{1 -  \theta}{2}	}	 \Phi	\|_{\F}^2
		\bigg)  \ , 
		\label{prop estimate fermion}
		\\ 
		|	 \langle	\Psi , [  \widetilde \L (t) , \calN_B	] \Psi 	\rangle_\F		|  
		& \leq
		C e^{Ct }
		\lambda
		\Big(  \sqrt{NM \hbar }   + M \hbar \Big) 
		\bigg( 
		\|	 (\N + \1 )^{	 \frac{1 + \theta}{2}	}	 \Psi	\|_{\F}^2
		+ 
		\|	 (\N + \1 )^{	 \frac{1 -  \theta}{2}	}	 \Phi	\|_{\F}^2
		\bigg)  \ . 
		\label{prop estimate boson}
	\end{align}
\end{lemma}

First, we turn to the proof of the above commutator estimates.
Subsequently, we show how they imply Proposition \ref{prop number estimates truncated}.

\begin{proof}[Proof of Lemma \ref{lemma commutator estimates}]
	Let us  introduce some notation and facts that we use throughout the proof. First, we recall that in \eqref{u and v} we associated  $u_t = u(t)$	 and
	$v_t = v(t)$    to the solution $\omega(t)$ of the mean-field equation. 
	In particular, $ u(t) v^*(t)  = 0$. 
	Hence, the trace estimate \eqref{trace} implies 
	that for all $\xi \in \R^d$
	\begin{equation}
		\|   u (t) e_\xi v^*(t)      \|_\tr  
		= 
		\|    [ u (t)  , e_\xi  ] v^* (t)       \|_\tr  
		= 
		\|    [\omega(t) , e_\xi ]  v^*(t) \|_\tr 
		\leq C e^{Ct} M \hbar \< \xi \> 
	\end{equation}    
	where we used  $\|    v^*(t) \|_{B (L^2)} \leq 1$. 
	In particular, thanks to $\|  A\|_{  HS } \leq  \|  A  \|_{\tr}^{1/2}$ for a trace-class operator $A$, we also obtain 
	for all $\xi \in \R^d$ 
	\begin{align}
		\|   u (t) e_\xi v^*(t)      \|_{HS}
		\leq C e^{Ct}  \sqrt{ M \hbar  \< \xi \> }    \ .   
	\end{align}
	For convenience, we also denote 
	$\| V\| \equiv \int_{\R^d } |\hat V (\xi)| \< \xi \> \d \xi  $. 
	
	\vspace{1mm}
	
	Let us now turn to the proof of Lemma \ref{lemma commutator estimates}. First, we prove the estimates for $\N_F$. 
	Secondly, we prove the estimates for $\N_B$. 
	Here, we only prove the case for $\theta = 0 $. 
	For general $\theta $ it suffices to insert an identity 
	$\1 = (\N + 4  )^{1+\theta} (\N + 4 )^{1-\theta}$
	and use the pull through formulas 
	$  \N a_x = a_x (\N + \1 )$,
	and $\N b_y = b_y (\N+ \1 )$
	on each term of $\widetilde \L (t). $

	\vspace{1mm}

	\noindent 	\textit{Proof of Eq. \eqref{prop estimate fermion}}.
	Using the relations  $[\N_F, a^*(g)] = + a^*(g)$ and $[\N_F , a(g)] = - a(g)$
	one is able to calculate the commutator 
	\begin{align}
		\nonumber 
		[  \widetilde \L(t) ,\N_F]
		= 
		&  + 
		\lambda \sqrt N 
		[  \L_{2,1}(t), \N_F]
		+ 
		\lambda 
		[\L_{2,2} (t), \N_F ]
	\end{align}
	where 
	\begin{align}
		\nonumber
		[  \L_{2,1}(t), \N_F]
		= 
		&  + 
		2
		\int_{	 \R^d \times \R^d 	} V(x-y )
		a^*(u_x) a^*(	 \overline v _x) \otimes \big(	  \vp(y) b_y^* + \overline{\vp(y)}  b_y					\big)	\d x \d y 		\\  
		\label{prop eq 1}
		&  -  	 2
		\int_{	 \R^d \times \R^d 	} V(x-y )
		a(	 \overline v _x)		a (u_x)  \otimes \big(	  \vp(y) b_y^* + \overline{\vp(y)}  b_y					\big)	\d x \d y 		\  , 	\\
		\nonumber 
		[\L_{2,2 }(t) , \N_F ]	 
		= &   + 
		2
		\int_{	 \R^d \times \R^d 	} V(x-y )
		a^*(u_x) a^*(	 \overline v _x) \otimes  b_y^* b_y 	\d x \d y 		\\ 
		\label{prop eq 2}
		&  -  	 2
		\int_{	 \R^d \times \R^d 	} V(x-y )
		a(	 \overline v _x)		a (u_x)  \otimes	 b_y^* b_y 	\d x \d y 	 \  .
	\end{align}
	Let us first estimate the terms in \eqref{prop eq 1}. 
	To this end, we use a Fourier decomposition			
	$V(x)  =   	\int_{\R^d }	 \hat V (\xi)	 e_\xi(x) 	 \d \xi 	$,
	where $e_{\xi} (x) \equiv (2\pi)^{-d/2} \exp[ i x \cdot \xi ]$ , 
	to find that 
	\begin{align}
		[  \L_{2,1} & (t)   ,   \N_F] \\ 
		\nonumber 
		& 	=  
		2			\int_{	\R^d	} 
		\hat V ( 		\xi)  
		\Big(    \int_{	 \R^d \times \R^d 	}   [   u e_\xi v^*   ] (x_1 , x_2)	 a_{x_1}^* a_{x_2}^* \d x_1 \d x_2 	 \Big)   
		\otimes 
		\big(
		b^*[ e_{-\xi} \vp 		]    
		+ 
		b[ e_{ \xi} \vp 		]    
		\big) \d \xi 		\\ 
		\nonumber 
		&  -  	
		2			\int_{	\R^d	} 
		\hat V ( 		\xi)  
		\Big(    \int_{	 \R^d \times \R^d 	} [   v e_\xi u   ] (x_1 , x_2)	 a_{x_1} a_{x_2} \d x_1 \d x_2 	 \Big)   
		\otimes 
		\big(
		b^*[ e_{-\xi} \vp 		]    
		+ 
		b[ e_{ \xi} \vp 		]     
		\big) \d \xi 	  \  .  
	\end{align}
	Thus, we use the estimates contained in Lemma \ref{lemma fermion estimates} and \ref{lemma boson estimates}
	and the Cauchy-Schwarz inequality for $\<  \cdot , \cdot \>_\F $
	to find that 
	there exists a constant $C>0$ 
	such that 
	for all $\Psi ,\Phi \in \F$
	there holds 
	\begin{align}
		\nonumber 
		|			 \langle	 \Psi , [  \L_{2,1} (t),  & \N_F]  \Phi  	\rangle_\F		|    \\ 
		\nonumber 
		&  \leq 
		C 
		\|	 V	\|
		\sup_{\xi \in\R^d}   \<\xi\>^{-1  }
		\|	 u e_\xi v^*	\|_{HS}
		\|   (\N_F+\1 )^{\frac{1}{2}} \otimes \1 	 \Psi 	\|_\F 
		\|	 \vp 	\|_{L^2} \| \1 \otimes (\N_B+\1)^{\frac{1}{2}}	 \Phi 	\|_\F  \  \\ 
		& \leq 
		C 
		\|	 V	\|
		\sqrt{M\hbar} \exp(Ct)
		\|     (\N_F+\1 )^{\frac{1}{2}} \otimes \1 	 \Psi 	\|_\F   \| \1 \otimes   (\N_B+\1 )^{\frac{1}{2}}	 \Phi 	\|_\F  \  	 \ ,   
		\label{L 21 NF}
	\end{align}
	where in the last line we have used  
	$	 \|	 u e_\xi v^*	\|_{HS}    
	\leq    C \sqrt{M\hbar   \< \xi \> }     e^{Ct } $
	and $\|	  \vp	\|_{L^2} =1 $. 
	Similarly, the second term in \eqref{prop eq 2} 
	may be expanded in its Fourier  coefficients--we find that 
	\begin{align} 
		[\L_{2,2}(t)  ,   \N_F] 
		\nonumber 
		& 	=  
		2 
		\int_{	\R^d	} 
		\hat V ( 		\xi)  
		\Big(    \int_{	 \R^d \times \R^d 	}   [   u e_\xi v^*   ] (x_1 , x_2)	 a_{x_1}^* a_{x_2}^* \d x_1 \d x_2 	 \Big)   
		\otimes 
		\d\Gamma[e_{-	 \xi 	}]
		\d \xi 		\\  
		&  -  	
		2 
		\int_{	\R^d	} 
		\hat V ( 		\xi)  
		\Big(    \int_{	 \R^d \times \R^d 	} [   v e_\xi u   ] (x_1 , x_2)	 a_{x_1} a_{x_2} \d x_1 \d x_2 	 \Big)   
		\otimes 
		\d\Gamma[e_{-	 \xi 	}]
		\d \xi 	  \  .  
	\end{align}
	Thus, we use the estimates contained in Lemma \ref{lemma fermion estimates} and \ref{lemma boson estimates}
	and the Cauchy-Schwarz inequality for $\<  \cdot , \cdot \>_\F $
	to find that 
	there exists a constant $C>0$ 
	such that 
	for all $\Psi \in \calD (\N_F) \cap \calD(\N_B)$
	there holds 
	\begin{align}
		\nonumber 
		|		  	 \langle	  \Psi ,&   [  \L_{2,2} (t),  \N_F]  \Phi  	\rangle_\F		|      \\
		\notag 
		& = 
		|			 \langle      (\N_B + \1 )^{\frac{1}{2}}  	  \Psi , [  \L_{2,2} (t),  \N_F]  	  (\N_B + \1 )^{ - \frac{1}{2}}  	\Phi	\rangle_\F		|     \\
		\nonumber 
		&  \leq 
		C 
		\|	 V	\|
		\sup_{\xi \in\R^d}   \<\xi\>^{-1  } 
		\|	 u e_\xi v^*	\|_{Tr} 
		\|   	  (\N_B + \1 )^{\frac{1}{2}}  \Psi 	\|_\F 
		\|	     \vp 	\|_{L^2} 
		\| \1 \otimes \N_B    (\N_B + \1 )^{- \frac{1}{2}}   \Phi 	\|_\F  \  \\ 
		& \leq 
		C 
		\|	 V	\|
		M  \hbar  \exp(Ct )
		\|   	(\N_B + \1 )^{	\frac{1}{2}	} \Psi 	\|_\F 
		\|  \1 \otimes ( \N_B  +\1)^{	 \frac{1}{2}	} \Phi 	\|_\F   \  , 
		\label{L 22 NF}
	\end{align}
	where in the last line we have used  $	 \|	 u e_\xi v^*	\|_{Tr}   \leq C M  \hbar \<\xi \>e^{Ct } $
	and $\|	  \vp	\|_{L^2} =1 $. 
	
	\vspace{1mm}
	
	We gather the estimates contained in Eqs. \eqref{L 21 NF} and \eqref{L 22 NF}, 
	and use     the basic estimates $\N_F \otimes \1 \leq \N $, 
	$ \1 \otimes \N_B \leq \N $
	together with Young's inequality $ab \leq a^2/2 + b^2/2$ for $a,b\geq0. $
	This finishes the proof.

	\vspace{1mm}
	
	\noindent 	\textit{Proof of Eq. \eqref{prop estimate boson}}.
	Using the relations  $[\N_B, b^*(g)] = + b^*(g)$ and $[\N_B , b(g)] = +  b(g)$
	one is able to calculate the commutator 
	\begin{align}
		\nonumber 
		[  \widetilde\L(t) ,\N_B]
		= 
		&  + 
		\lambda \sqrt N 
		[\widetilde \L_{2,1}(t), \N_B] 
	\end{align}
	where 
	\begin{align}
		\nonumber 
		[ \widetilde \L_{2,1}(t), \N_B]
		&	= 
		+ 
		2
		\chi(\N_B \leq \hbar M )
		\int_{	 \R^d \times \R^d 	} V(x-y )
		a^*(u_x)  a(u_x )\otimes \big(	  \vp(y) b_y^* + \overline{\vp(y)}  b_y					\big)	\d x \d y 		\\  
		\nonumber 
		&   +  	
		2
		\chi(\N_B \leq \hbar M )
		\int_{	 \R^d \times \R^d 	} V(x-y )
		a^*(	 \overline v _x)		a(	 \overline v _x)	 \otimes \big(	  \vp(y) b_y^* + \overline{\vp(y)}  b_y					\big)	\d x \d y 	 	\\
		\nonumber 
		&+ 
		2
		\int_{	 \R^d \times \R^d 	} V(x-y )
		a^*(u_x) a^*(	 \overline v _x) \otimes \big(	  \vp(y) b_y^* + \overline{\vp(y)}  b_y					\big)	\d x \d y 		\\  
		&   +  	 2
		\int_{	 \R^d \times \R^d 	} V(x-y )
		a(	 \overline v _x)		a (u_x)  \otimes \big(	  \vp(y) b_y^* + \overline{\vp(y)}  b_y					\big)	\d x \d y 		\ . 
	\end{align}
	Similarly as before, a Fourier decomposition
	for the interaction potential yields 
	\begin{align}
		[\widetilde \L_{2,1}(t), \N_B]
		&	= 
		\label{prop eq 3}
		2
		\int_{	\R^d	} 
		\hat V ( 		\xi)  
		\d \Gamma [u e_\xi u ]
		\otimes		\chi(\N_B \leq \hbar M )
		\big(
		b^*[ e_{-\xi} \vp 		]    
		+ 
		b[ e_{ \xi} \vp 		]    
		\big) 
		\d \xi 	\\  
		\label{prop eq 4}
		&   +  	 	2
		\int_{	\R^d	} 
		\hat V ( 		\xi)  
		\d \Gamma [ v^* e_\xi v ]
		\otimes 
		\chi(\N_B \leq \hbar M )	\big(
		b^*[ e_{-\xi} \vp 		]    
		+ 
		b[ e_{ \xi} \vp 		]    
		\big) 
		\d \xi 	\\
		\label{prop eq 5}
		&+ 
		2
		\int_{	\R^d	} 
		\hat V ( 		\xi)  
		\Big(    \int_{	 \R^d \times \R^d 	} 
		[   u e_\xi v^*   ] (x_1 , x_2)	 a_{x_1}^* a_{x_2}^* \d x_1 \d x_2 	 \Big)   
		\otimes 
		\big(
		b^*[ e_{-\xi} \vp 		]    
		+ 
		b[ e_{ \xi} \vp 		]     
		\big) 	\d \xi 	\\  
		\label{prop eq 6}
		&   +  	 2
		\int_{	\R^d	} 
		\hat V ( 		\xi)  
		\Big(    \int_{	 \R^d \times \R^d 	} [   v e_\xi u   ] (x_1 , x_2)	 a_{x_1} a_{x_2} \d x_1 \d x_2 	 \Big)     \otimes  	\big(
		b^*[ e_{-\xi} \vp 		]    
		+ 
		b[ e_{ \xi} \vp 		]     		 \big) 	\d \xi 	\ . 
	\end{align}
	Up to an overall minus sign, the terms contained in 
	Eqs. \eqref{prop eq 5} and \eqref{prop eq 6} 
	have already been estimated above.
	Thus, it suffices to estimate the two terms contained in 
	Eqs. \eqref{prop eq 3} and \eqref{prop eq 4}.  
	Let us start from the one in Eq. \eqref{prop eq 3} and, for simplicity, let us only show how to bound the 
	contribution arising from $b[e_\xi \vp]$--the other one is analogous.
	Let us fix $\xi \in \R^d$, and let  $\Psi , \Phi \in \F $. 
	Denoting $\chi_M \equiv 		\chi(\N_B \leq \hbar M )$, 
	we find using   
	and   \textit{pull-through formula}
	that 
	$\N_B b_x = b_x (\N_B + \1)$
	\begin{align}
		|	 \<   \Psi  ,  
		\d \Gamma[   u e_\xi u  ]  \otimes \chi_M 
		b[e_\xi \vp ]	
		\Phi  	 \>_\F 		| 
		& 		= 
		\nonumber 
		|	 \langle  \chi_M   \N_B^{\frac{1}{4}	}   \Psi  ,   \d \Gamma[   u e_\xi u  ]  \otimes \chi_M 
		b[e_\xi \vp ]   (\N_B+\1)^{-\frac{1}{4}	} 	 \Phi  	 \rangle_\F 		| \\
		\nonumber 
		& 
		\leq 
		\|		\N_F^{\frac{1}{2}} 		 \chi_M    \N_B^{\frac{1}{4}	}   \Psi  	 	\|_\F 
		\|			\N_F^{\frac{1}{2}} 		  \chi_M   b[e_\xi \vp ]   (\N_B+\1)^{-\frac{1}{4}	} 	 \Phi  	\|_\F    \\ 
		& \leq 
		\sqrt{ M \hbar }
		\|	 \N_F^{\frac{1}{2}}	 \otimes \1 		\Psi	\|_{\F}
		\|	 \N_F^{\frac{1}{2}}		 \otimes \1 	\Phi	\|_{\F}
	\end{align}
	In the last line, we have made use of 
	Lemma \ref{lemma fermion estimates},  \ref{lemma boson estimates}, 		$		\|  u e_\xi u 	 \|_{B(L^2)} \leq 1 $, 
	and the cut-off bounds 
	$ \|  \chi_M b[e_\xi \vp] (\N_B+1)^{-1/4} \| \leq  C (M \hbar)^{1/4}$
	and 
	$\|   \chi_M \N_B^{1/4} \| \leq (M \hbar )^{1/4}$.
	Similarly, we find  
	\begin{equation}
		|	 \<   \Psi  ,   \d \Gamma[   v^* e_\xi v	  ]  \otimes \chi_M  b[e_\xi \vp ]	 \Psi 	 \>_\F 		| 
		\leq 
		\sqrt{ M \hbar }
		\|	 \N_F^{\frac{1}{2}}	 \otimes \1 		\Psi	\|_{\F}
		\|	 \N_F^{\frac{1}{2}}		 \otimes \1 	\Phi	\|_{\F}
	\end{equation}
	The proof if finished, once we gather our estimates and use 
	the elementary  inequalities 
	$\N_F \otimes \1 \leq \N $, 
	$ \1 \otimes \N_B \leq \N $
	together with Young's inequality $ab \leq a^2/2 + b^2/2$ for $a,b\geq0. $
\end{proof}

\begin{proof}[Proof of Proposition \ref{prop number estimates truncated}]
	Let us fix $ k \in \mathbb N  $ and   recall that
	$\i \hbar \partial_t \widetilde \U(t,s ) = \widetilde\L(t)\widetilde \U(t,s)$. 
	Hence, we may compute the time derivative for $t,s \in\R $
	\begin{align}
		\nonumber 
		\frac{d}{dt} 
		\widetilde \U^* (t,s) \N^k 	\widetilde \U(t,s)  
		&  =
		\frac{1}{ \i \hbar }
		\widetilde \U^* (t,s) [  \widetilde \L (t ) , \N^k   ] 		\widetilde \U (t,s) 	 \\ 
		&   = 
		\frac{1}{ \i \hbar }
		\sum_{i =1 }^k 
		\widetilde \U^* (t,s)  \N^{ i -1 }  [  \widetilde \L (t ) , \N   ]  \N^{ k - i }		\widetilde \U (t,s)  \ . 
	\end{align}
	Taking $\Psi \in \F$, we can then estimate that 
	\begin{align}
		\frac{d}{dt}	 \<	  \Psi , 	\widetilde \U^* (t,s)  \N^k \widetilde \U (t,s) \Psi 	 \>_\F 
		\leq  
		\frac{1}{\hbar }
		\sum_{i =1 }^k 
		| \langle
		\N^{ i -1 }   \widetilde \U (t,s) \Psi,     [  \widetilde \L (t ) , \N   ]  \N^{ k - i }		\widetilde \U (t,s)  \Psi \rangle_\F  | \ . 
	\end{align}
	Fix $ i  =1, \ldots , k $ and let $\theta = 1 + k -2 i \in \R $. 
	In view of Lemma
	\ref{lemma commutator estimates} there exists $C > 0 $ such that 
	\begin{align}
		\nonumber 
		| \langle
		& 	\N^{ i -1 }     \widetilde \U (t,s) \Psi	 		,     [  \widetilde \L (t ) , \N   ] 
		\N^{ k - i }		\widetilde \U (t,s)  \Psi
		\rangle_\F  |  	\\
		\nonumber 
		& 	\leq 
		C e^{Ct }
		\lambda
		(  \sqrt{NM \hbar }   + M \hbar ) 
		\Big(    
		\|	 (\N + \1 )^{	 \frac{1 + \theta}{2}	} 	\N^{ i -1 }   \widetilde \U (t,s) \Psi 	\|_{\F}^2
		+ 
		\|	 (\N + \1 )^{	 \frac{1 - \theta}{2}	}		 \N^{ k - i }		\widetilde \U (t,s)     \Psi	\|_{\F}^2
		\Big)   , \\   
		\nonumber 
		& \leq 
		C e^{Ct }
		\lambda
		(  \sqrt{NM \hbar }   + M \hbar )  
		\|	 (\N + \1 )^{\frac{k}{2}	} 	  \widetilde \U (t,s) \Psi 	\|_{\F}^2 	  \\
		&   = 
		C e^{Ct }
		\lambda
		(  \sqrt{NM \hbar }   + M \hbar )  
		\langle 		\Psi , \widetilde \U^*(t,s)  ( \N^k+ \1)  \widetilde \U(t,s) \Psi 	 \rangle_\F 	  \ . 
	\end{align}
	The proof is then finished after we apply the Gr\"onwall inequality and
	use the initial condition $\widetilde\U(s,s)=\1 . $
\end{proof}

\subsection{Proof of Theorem \ref{prop number estimates}}
Let us now give a proof of Theorem \ref{prop number estimates}.
Essentially, we shall compare the number estimates generated by   $ \U(t,s)$
and $\widetilde \U(t,s)$. 
The latter have already been established in Proposition \ref{prop number estimates}.
Additionally, we shall need the following \textit{weak bounds}
on the growth of particle number with respect to the original fluctuation dynamics. 

\begin{lemma}[Weak number estimates]
	\label{lemma weak estimates}
	Let $\U(t,s)$ be the fluctuation dynamics, defined in  \eqref{unitary U}.
	Then, the following statements hold true. 
	\begin{enumerate}[leftmargin=*]
		\item 
		For all $\ell \in \mathbb N$ there is a constant $C>0$ such that for all $t,s\in \R$ and $\Psi \in \F$
		there holds 
		\begin{equation}
			\|    ( \N_F+M)^\ell \U(t,s) \Psi \|_\F 
			\leq 
			C 
			\|    ( \N_F+M)^\ell   \Psi \|_\F 
		\end{equation}

		\item 
		For all $\ell \in \mathbb N$ there is a constant $C>0$ such that for all $t,s\in \R$ and $\Psi \in \F$
		there holds 
		\begin{equation}
			\|    ( \N_B+N )^\ell \U(t,s) \Psi \|_\F 
			\leq 
			C 
			\|    ( \N_B+N)^\ell   \Psi \|_\F 
		\end{equation}
	\end{enumerate}
\end{lemma}

\begin{proof}
	We only give a proof for the fermion number operator since the proof for bosons is similar;  
	we refer the reader to  \cite[Lemma 3.6]{RodnianskiSchlein2009} 
	for the situation in which only bosons are present (interactions do not change the proof). 
	
	\vspace{1mm}
	Let us then  consider the particle-hole transformation $\calR_t$ in the definition of $\U(t,s)$ and notice that
	\begin{equation}
		\calR_t (\N_F+M)^\ell \calR_t^* = [ \calR_t (\N_F+M) \calR_t^* ]^\ell  = 
		(
		\d \Gamma[v(t)] 
		-
		\d\Gamma[u (t)] 
		+ 2 M 
		)^\ell  \ . 
	\end{equation}
	The terms 
	$\d\Gamma[u (t)] $ and $ \d \Gamma[v(t)] $
	can be estimated using Lemma \ref{lemma fermion estimates}. 
	Namely,  it follows that there exists $c_0$ such that
	for all $\Phi\in \F$ 
	\begin{equation}
		\|
		(
		\d \Gamma[v(t)] 
		-
		\d\Gamma[u (t)] 
		+ 2 M 
		) \Phi
		\|_\F 
		\leq 
		c_0 \|
		(\N_F + M )
		\Phi\|_\F \ . 
	\end{equation}  
	Consequently, since $ [ \d\Gamma_F[J], \N_F]=0$
	for $J=u(t)$ and $J=v(t)$, 
	an $\ell$-fold application of the previous estimate yields 
	\begin{align}
		\label{l fold application}
		\|      
		( \N_F     +M)^\ell
		\calR_t^* 
		\Phi  
		\|_\F  
		& \leq 
		c_0^\ell 
		\|       (\N_F + M )^\ell \Phi     \|_\F \  , 
		\qquad \forall t \in \R\  . 
	\end{align}
	Here  we have used the fact that $\calR_t$ is a unitary transformation on Fock space. 
	We conclude using the fact that $(\N_F + M)$
	commutes with the Schr\"odinger dynamics $\exp(- i(t-s) \H / \hbar)$, 
	and the bosonic  Weyl operator $\W_t$.
	That is 
	\begin{align}
		\nonumber
		\|       (\N_F + M )^\ell   \U (t,s) \Psi        \|_\F 
		&  =
		\|             (\N_F + M )^\ell    \calR_t^*        \W_t^*e^{-i(t-s)\H /\hbar} \calR_s \W_s \Psi            \|_\F        \\
		\nonumber
		& \leq  
		c_0^\ell   
		\|        (\N_F + M)^\ell  \calR_s \Psi    \|_\F \\
		& \leq 
		c_0^{2\ell }
		\|        (\N_F + M)^\ell   \Psi    \|_\F \  . 
	\end{align} 
	where in the last line we have used that $\calR_s = \calR_{-s}^*$
	together with the    estimate \eqref{l fold application}. 
	This finishes the proof. 
\end{proof}

\begin{proof}[Proof of Theorem \ref{prop number estimates}]
	Since the proof is essentially the  same one for fermions and bosons, 
	we   only present the proof of the former in full details,  and point out the differences with respect to the latter. 
	To this end,  	let us now fix  $\Psi \in \F$   together with $t,s \in \R $
	and $k,\ell \in \mathbb N. $ 
	We start by computing that 
	\begin{align}
		\nonumber 
		\| \N_F^{\frac{\ell}{2}	} \U(t,s) \Psi		\|_\F^2 
		& 	 =  
		\langle 	
		\Psi , \U (t,s )	
		\N_F^\ell  
		\widetilde \U(t,s) 
		\Psi 		\rangle_\F 
		+
		\langle 	 \Psi , \U (t,s ) 
		\N_F^\ell 
		\Big( 				\U(t,s) 	- 			\widetilde \U(t,s) 				\Big) 
		\Psi 		\rangle_\F 	\\
		& \leq 
		\|	 \Psi 	\|_\F	 
		\|	 	 \N_F^\ell   \widetilde \U(t,s) \Psi 		\|_\F 
		+ 
		\|	 \N_F^\ell 		\U (t,s) 		\Psi 	\|_\F 
		\|	 	 (	  \widetilde \U(t,s)			-  \U(t,s))		 \Psi 		\|_\F 
		\label{eq 1} \ . 
	\end{align}
	We now estimate  the two terms in \eqref{eq 1}.  
	
	\vspace{1mm}
	For the first one, we use the fact that the truncated dynamics satisfies  the estimate 
	contained in Proposition \ref{prop number estimates truncated}. 
	Namely,  for a constant $C>0 $
	\begin{eqnarray}
		\|	 	 \N_F^\ell   \widetilde \U(t,s) \Psi 		\|_\F 
		\leq 
		C 	 
		K(t-s)
		\|	 	  ( \N+\1)^\ell     \Psi 		\|_\F  \ . 
	\end{eqnarray}
	
	\vspace{1mm}
	
	For the second term  in \eqref{eq 1}
	we use the weak number estimates from Lemma \ref{lemma weak estimates}.
	Namely, for a constant $C = C(\ell)$
	\begin{equation}
		\|	\N_F^\ell \U(t,s)	 \Psi \|_\F
		\leq 
		C 
		M^\ell 
		\|	    (\1 + \N_F / M)^\ell \Psi	\|  \  , \qquad \forall t, s \in \R \   .
	\end{equation}		
	Next,  we study the difference between
	the original and the truncated dynamics.
	That is, 
	we use Duhamel's formula
	to find that 
	\begin{align}
		\widetilde \U(t,s)  -  	\U(t,s )  
		= 
		-i  \delta 
		\int_s^t 
		\U(t, r  )
		\Big(
		\widetilde \L_{2,1} (r )  - \L_{2,1}(r ) 
		\Big)
		\widetilde \U(r ,  s ) \d r   ; 
	\end{align}
	where we have used the fact that 
	$ \widetilde \L (t)   - \L(t) =  \lambda \sqrt N  ( \widetilde \L_{2,1} (t)  -  \L_{2,1} (t))$, 
	and we have collected $\delta = \lambda \sqrt N  / \hbar $ as a pre-factor. 
	Next, for $\Phi = \U(t,s)\Psi  \in \F$ one may estimate using Lemma  \ref{lemma fermion estimates} 
	and \ref{lemma boson estimates} that 
	for all  $ k \geq 0 $
	(here, we omit time labels for notational convenience)
	\begin{align}
		\nonumber 
		\|	  (	 \widetilde \L_{2,1}	  	-   \L_{2,1}  	)	 \Phi	\|_{\F} 
		& \leq 
		\int_{\R^d}
		|	 \hat V(\xi)	|	
		\Big\|	
		\d \Gamma [u e_\xi u ]	 \otimes  \chi( \N_B \geq \hbar M 	)\Big(	 b[e_{\xi} \vp]  + b^*[ e_{-\xi} \vp ]		\Big)	  \Phi 	
		\Big\|_\F  
		\d \xi
		\\
		\nonumber 
		& + 
		\int_{\R^d}
		|	 \hat V(\xi)	|	
		\Big\| 
		\d \Gamma [ v^* e_\xi v ]	 \otimes \chi( \N_B \geq \hbar M 	)\Big(	 b[e_{\xi} \vp]  + b^*[ e_{-\xi} \vp ]		\Big) \Phi 		\Big\|_\F  \d \xi			\\
		\nonumber 
		& \leq 2 
		\|	 \hat V	\|_{L^1 }
		\|	   \N_F \otimes 	 \chi( \N_B \geq \hbar M 	)(\N_B +\1)^{\frac{1}{2}} \Phi	\|_\F 		\\ 
		\nonumber 
		& \leq 
		\frac { 2
			\|	 \hat V	\|_{L^1 }}{		(\hbar M)^k 		}
		\|	   \N_F \otimes 	 \chi( \N_B \geq \hbar M 	)(\N_B +\1)^{k+ \frac{1}{2}} \Phi	\|_\F 	  \\
		& \leq 
		\frac { C
			\|	 \hat V	\|_{L^1 }}{		(\hbar M)^k 		}
		\|	  (\N + \1 )^{k + \frac{3}{2}} \Phi	\|_\F 	 \   . 
	\end{align}
	We put our last three estimates together to find that 
	thanks to Proposition \ref{prop number estimates truncated} 
	\begin{align}
		\|	 	\Big(  \widetilde \U (t,s)      - \U(t,s) \Big)  		\Psi 		\|_\F 	 
		\leq 
		K(t-s)
		\frac{\lambda \sqrt N }{\hbar}
		\frac{1}{( \hbar M)^k}
		\|	 (1 + \N)^{k+3/2} \Psi	\|_{\F}  \ . 
	\end{align}
	Putting our estimates together, we find that 
	there exists $C>0 $ such that 
	\begin{align*} 	  
		\|  \N_F^{\frac{\ell}{2}	}   \U(t,s) &  \Psi		\|_\F^2  \\
		& 	  \leq 
		K(t-s ) 
		\Big[
		\|	 \Psi	 \|_\F  \| (\N + \1 )	 \Psi	\|_\F 
		+  
		\Theta_{k,\ell}
		\|  \big(\1 + \frac{\N_F}{ M } \big)^\ell   \Psi		\|_\F 	
		\|	 	  ( \N+\1)^{k + \frac{3}{2}	}    \Psi 		\|_\F   
		\Big]
	\end{align*}
	where $\Theta_{k,\ell}$ is as in the statement of the theorem. 
	
	\vspace{1mm}
	
	As for the bosons, the only modification comes from the weak number estimates
	obtained from Lemma \ref{lemma weak estimates}, 
	In the bosonic case, 
	one finds that for a constant $C = C(\ell)$
	\begin{equation}
		\|	\N_B^\ell \U(t,s)	 \Psi \|_\F
		\leq 
		C 
		N^\ell 
		\|	    (\1 + \N_B /  N )^\ell \Psi	\|  \  , \qquad \forall t, s \in \R \  , 
	\end{equation}		
	This finishes the proof. 	
\end{proof}


\section{Second Quantization III: Proof of Theorem \ref{thm1}}
\label{section proof thm 1}

Let us now  turn to the proof our first main result. 
We shall reduce the problem of  proving the estimates contained in Theorem \ref{thm1}, to that of proving number estimates for the fluctuation dynamics, as defined in Section \ref{section fluctuation}. 
Then, we apply Theorem \ref{prop number estimates}. 
The proofs we present are heavily inspired by the works 
\cite{RodnianskiSchlein2009,BenedikterPortaSchlein2014}
and also 
\cite{ChenLee2011,ChenLeeSchlein2011}, 
and are adapted to the case at hand.
Let us remark here 
that our proofs are shorter, because we do not aim at obtaining the optimal convergence rates
$1 / M $ for fermions, and $1/N$ for bosons.
Here instead, we content ourselves with $1/\sqrt M$ and $1/\sqrt N$, respectively. 

\vspace{2mm}

Let us first introduce some notation. 
Letting $\U(t,s)$ be the fluctuation dynamics, 
we consider  the following 
\textit{fluctuation vectors}
in Fock space $\F$
\begin{align}
	\Omega_1 (t) 
	\equiv 
	\U(t,0) \Omega   
	\quad
	\t{and}
	\quad 
	\Omega_2(t)
	\equiv 
	\U(t,0)
	\Omega_F \otimes \W^*[\sqrt N \vp_0] \frac{b^*(\vp_0)}{\sqrt  N } \Omega_B \ . 
\end{align} 
They will be extremely useful in the proof of Theorem \ref{thm1}.
Heuristically,   
$\Omega_1(t)$
determines the fluctuations
of the system when the initial data 
has a bosonic  coherent state.
On the other hand, $\Omega_2(t)$
describes the fluctuations 
when the initial data has a bosonic factorized state. 
Let us collect  in the following lemma 
some results concerning these vectors; 
for reference, see \cite{ChenLeeSchlein2011}
or
\cite{DietzeLee2022}.

\begin{lemma}[Properties of $\Omega_1$ and $\Omega_2$]
	Let us denote from now on 
	\begin{equation}
		d_N \equiv \frac{\sqrt{N!}}{N^{N/2} e^{-N/2}}
	\end{equation}
	Then, the following statements hold true
	\begin{enumerate}
		\item  $P_N \W[\sqrt N \vp_0]\Omega_B 
		= 
		\frac{1}{d_N} 
		\frac{b^*(\vp_0)}{\sqrt{N!}}\Omega_B$
		where 
		$P_N \equiv \1 (\N_B =N)$. 
		
		\item 
		$
		\<   \Omega_2(t) , \Omega_1(t) \>_\F
		= 
		\<   \Omega_2(0 ) , \Omega_1(0) \>_\F
		= 
		\frac{1}{d_N} \ . 
		$
		
		\item 
		There exists a constant $C>0$ such that 	
		$ 
		\|	 (\N_F+\1)^{-1/2}	 \Omega_2(0)	\|_\F 
		\leq 
		\frac{C}{d_N}  \ . $
		
	\end{enumerate}
\end{lemma}

\begin{proof}[Proof of Theorem \ref{thm1}]
	We recall  recall here 
	that one-particle reduced densities  
	have been re-cast in creation-
	and annihilation operators in  
	\eqref{definition kernels}. 
	For the proof, we write
	$a_x (t) = e^{it/\hbar \H} a_x e^{-it/\hbar \H}$. 
	
	\vspace{1mm}
	\noindent
	\textit{Proof of \eqref{thm1 fermion estimate}}		
	First, we establish the result for fermions. 
We obtain 
	\begin{align}
		\nonumber 
		\gamma_F(t;x_1 ,x_2)
		&  = 
		\< \Psi(t) , a_{x_2}^*a_{x_1}\Psi (t) \>_\F 		\\ 
		\nonumber 
		& = 
		\<
		\calR_0 \otimes \frac{b^*(\vp_0)^N}{\sqrt{N!}} \Omega
		, 
		a_{x_2}^* (t)  a_{x_1}  (t)  
		\calR_0   \otimes \frac{b^*(\vp_0)^N}{\sqrt{N!}} \Omega
		\>_\F  \\ 
		\nonumber 
		& = 
		d_N 
		\<
		\calR_0   \otimes \frac{b^*(\vp_0)^N}{\sqrt{N!}} \Omega
		, 
		a_{x_2}^* (t)  a_{x_1}  (t)  
		\calR_0   \otimes 
		\W[\sqrt N \vp_0 ]
		\Omega
		\>_\F  \\ 
		& = 
		d_N 
		\<  \Omega_2(t)  ,  
		\calR_t^*
		a_{x_2}^*  a_{x_1}  
		\calR_t 
		\Omega_1 (t) 
		\>_\F   
		\label{thm 1 eq 1 fermion}
	\end{align}
	Next, we look at the conjugation relations
	\begin{align} 
		\label{conjugation fermions}
		\calR_t^*
		a^*_x
		\calR_t 
		= 
		a^* ( u_{t,x}) 
		+ 
		a( \overline{  v_{t,x} } )	
		\quad \t{and}\quad 	 
		\calR_t^* a_x \calR_t 
		= 
		a ( u_{t,x}) 
		+ 
		a^*( \overline{  v_{t,x} } )	 
	\end{align}
	and obtain, using $\<\bar v_x , \bar v_y \>= \omega(x,y)$
	and 
	$\<\Omega_1 , \Omega_2\> = d_N$ 
	\begin{align}
		\nonumber 
  		\gamma_F( t; x_1 , & x_2) 
		-  \omega ( t;  x_1, x_2) \\
		\nonumber
		& 			
		\ = \     
		d_N  \<   \Omega_2 (t ) ,  a^*(u_{x_2,t})  a(u_{x_1,t})  \Omega_1 (t)			\>_\F   
		\ 	- \   
		d_N 
		\<   \Omega_2 (t ) ,  a^*(\overline{v}_{x_1,t})  a(\overline{v}_{x_2,t})  \Omega_1 (t)			\>_\F 	\\  
		& 	 			\  + \   d_N 
		\<   \Omega_2 (t ) ,  a^*(u_{x_2,t})  a^*(v_{x_1,t})  \Omega_1 (t)			\>_\F   \  + \ 
		d_N  
		\<   \Omega_2 (t ) ,  a(\overline{v}_{x_2,t})a( u_{x_1,t}) \Omega_1 (t)\>_\F    \    . 
	\end{align}
	Next, we let $\O $ be a compact operator in $L^2(\R^d)$
	with kernel $\O(x_1 ,x_2) $.  The above identity now implies that 
	\begin{align}
		\nonumber 
		\mathrm{Tr}  \O \big( \gamma_F (t) - \omega (t)	\big)
		&   =  
		d_N 
		\<	 \Omega_2 (t) ,  
		\,  \d \Gamma_F [u_t \O u_t ]    
		\Omega_1 (t) 	 \>_\F  \\ 
		\nonumber 
		& 			-  d_N 
		\<	 \Omega_2 (t) ,  
		\,  \d \Gamma_F [v_t^* \overline{\O^*} v_t  ]   
		\Omega_1 (t) 	 \>_\F 	\\ 
		\nonumber 
		&     
		+
		d_N 
		\<	 \Omega_2 (t) ,  
		\int_{ \R^d \times \R^d } 				
		[v_t \O   u_t] (x_1, x_2) 
		a_{x_1}    a_{x_2}   \, 
		\Omega_1 (t)	\>_\F  \\
		& + 
		d_N 
		\<	 \Omega_2  (t) , 
		\int_{ \R^d \times \R^d } 				
		[v_t \O   u_t]^* (x_1, x_2) 
		a_{x_1}^*    a_{x_2}^*
		\, 
		\Omega_1 (t)	\>_\F 
	\end{align}
	Next, we make use of Remark \ref{remark boundedness of N}. 
	In particular, 
	we consider $K_{1/2} \geq 1 $ 
	to be the integer that controls of $\N^{1/2}$. 
	The Cauchy-Schwarz inequality
	and 
	Lemma \ref{lemma fermion estimates}
	then gives
	\begin{align}
		\nonumber
		\big|
		\mathrm{Tr}  \O \big( \gamma_F (t) - \omega (t)	\big)
		\big|    
		& 	\leq 
		\|		 u \O  u 	\|_{B(\F)} 
		d_N 
		\| (\N + \1 )^{-K_{1/2}}	\Omega_2(t)	\|_{\F}
		\|	  (\N +\1 )^{K_{1/2} + 1  }	  \Omega_1 (t)	\|_{\F } \\ 
		\nonumber
		& 	+  
		\|		 v \O  v 	\|_{B(\F)} 
		d_N 
		\| (\N + \1 )^{-K_{1/2}}	\Omega_2(t)	\|_{\F}
		\|	  (\N +\1 )^{K_{1/2} + 1  }	  \Omega_1 (t)	\|_{\F }	\\ 
		\nonumber
		& 	+  
		\|		 v \O  u 	\|_{HS} 
		d_N 
		\| (\N + \1 )^{-K_{1/2}}	\Omega_2(t)	\|_{\F}
		\|	  (\N +\1 )^{K_{1/2} + 1/2  }	  \Omega_1 (t)	\|_{\F }	\\ 
		& 	+  
		\|		(vOu)^*	\|_{HS} 
		d_N 
		\| (\N + \1 )^{-K_{1/2}}	\Omega_2(t)	\|_{\F}
		\|	  (\N +\1 )^{K_{1/2} + 1/2  }	  \Omega_1 (t)	\|_{\F } \ . 
	\end{align}
	First, we control $\Omega_2(t)$ and then, we control $\Omega_1(t)$. 
	Indeed, Remark \ref{remark boundedness of N} implies that
	\begin{align}
		\nonumber 
		\| (\N + \1 )^{-K_{1/2}}	\Omega_2(t)	\|_{\F}
		& 		  \leq 
		\| (\N + \1 )^{-K_{1/2}}   \U(t,0) (\N + \1 )^{1/2}	\|_{B(\F)}			
		\| (\N+\1)^{-1/2} 	\Omega_2 (0)	\|_{\F} \\ 
		& 	  	 \leq 
		K(t) / d_N  
		\label{thm1 eq 2}
	\end{align}
	where we have used
	$ 
	\|	 (\N_F+1)^{-1/2}	 \Omega_2(0)	\|_\F 
	\leq 
	c/d_N   \ . $
	On the other hand, Remark \ref{remark control omega} immediately
	implies that 
	\begin{equation}
		\label{thm1 eq 3}
		\|	  (\N +\1 )^{K_{1/2} + 1  }	  \Omega_1 (t)	\|_{\F } 
		\leq K (t) 
	\end{equation}
	provided one re-updates $K(t)$. 

	Finally, we use 
	$\|	 u \O u 	\|_{B(\F)} \leq \|	 \O	\|_{B(\F)}$, 
	$\|	 v \O v 	\|_{B(\F)} \leq \|	 \O	\|_{B(\F)}$, 
	$\|	 v \O u 	\|_{HS)} \leq    \sqrt M   \|	 \O	\|_{B(\F)}$, 
	and
	$\|	   ( v \O u)^*	\|_{HS} \leq    \sqrt M   \|	 \O	\|_{B(\F)}$ 
	and the fact that the space of compact operators is dual to the space of trace-class operators,  to conclcude that 
	\begin{equation}
		\|	 \gamma_F (t)  - \omega(t)		\|_{\tr}
		\leq 
		K(t) \sqrt M \ . 
	\end{equation}
	This finishes the proof of the estimate for fermions.

	\vspace{1mm}
	\noindent 
	\textit{Proof of \eqref{thm 1 boson estimate}.}
	As for the bosons,  we start by looking at the identity 
	that is the analogous of 
	\eqref{thm 1 eq 1 fermion}. 
	Namely,  a similar argument shows that 
	\begin{align}
		\gamma_B(t; y_1, y_2)
		= 
		d_N 
		\< \Omega_1(t) ,   \W_t^* b_{y_2}^* b_{y_1}^* \W_t  \Omega_2(t)\>_\F \  . 
	\end{align}
	Next, the conjugation relations
	\begin{align}
		\label{conjugation bosons}
		\W_t^* b_x  \W_t
		=
		b_x + \sqrt N \vp_t(x)
		\qquad
		\t{and}
		\qquad 
		\W_t^* b^*_x  \W_t
		=
		b^*_x +  \overline{ \vp_t(x) }  
	\end{align}
	combined with the identity 
	$\<\Omega_2, \Omega_1\>= 1/ d_N $ now imply that 
	\begin{align}
		\nonumber 
		\gamma_B(t; y_1, y_2)
		- 
		N 
		\vp_t(y_1 )
		\overline{\vp_t(y_2)}
		&  = 
		d_N 
		\< 	\Omega_2(t) ,				b_{y_2}^*  b_{y_1}					\Omega_1(t) 		\> \\
		\nonumber 
		&  + 
		d_N 
		\sqrt N 
		\< 	\Omega_2(t) ,			b^*_{y_2}			\vp_t(y_1) 			\Omega_1(t) 		\> \\
		&  + 
		d_N 
		\sqrt N 
		\< 	\Omega_2(t) ,			b_{y_1}  \overline{ \vp_t(y_2)	} 					\Omega_1(t) 		\>  \ . 
	\end{align}
	Next, we consider a trace-class operator
	$\O$ with kernel $O(x,y)$
	and integrate over $y_1, \ y_2$ to find that 
	\begin{align}
		\nonumber 
		\Big|
		\tr \Big(  
		\O (	\gamma_B(t) - N \ket{\vp_t} \bra{\vp_t}	)
		\Big)
		\Big| 	
		& 	\leq 
		d_N 
		\<	 \Omega_2(t) , \d \Gamma_B [\O ] \Omega_1 (t)		 \>  \\ 
		& + d_N 
		\sqrt N 
		\<	 \Omega_2(t) , 
		\big(    b^*[\O \vp_t] + b[\O \vp_t ]    \big)    
		\Omega_1 (t)		 \>   \ . 
	\end{align}
	Similarly as we did for fermions, 
	we consider $K_{1/2}$
	to be the integer from Remark \ref{remark control omega}. 
	Thus,  the Cauchy-Schwarz inequality, 
	the pull-through formula
	$\N b_y = b_y ( \N +1)$
	and 
	Lemma \ref{lemma boson estimates} now imply
	that 		
	\begin{align}
		\nonumber 
		\Big|
		\tr \Big(  
		\O (	\gamma_B(t) - N  & \ket{\vp_t} \bra{\vp_t}	)
		\Big)
		\Big| 			\\ 
		\nonumber
		& 	\leq 
		d_N 
		\|	(\N+\1)^{-K_{\frac{1}{2}}}  \Omega_1(t)	\|_\F
		\|	(\N+\1)^{K_{	\frac{1}{2}		}}    \d \Gamma_B[\O]      \Omega_1(t)	\|_\F \\
		\nonumber 
		& + 
		d_N \sqrt N 
		\|	(\N+\1)^{-K_{\frac{1}{2}}}  \Omega_1(t)	\|_\F
		\|	(\N+\1)^{K_{	\frac{1}{2}		}}      b^*[\O \vp_t]       \Omega_1(t)	\|_\F \\
		\nonumber
		& + 
		d_N \sqrt N  
		\|	(\N+\1)^{-K_{\frac{1}{2}}}  \Omega_1(t)	\|_\F
		\|	(\N+\1)^{K_{	 \frac{1}{2}		}}    b[\O \vp_t ]       \Omega_1(t)	\|_\F  \\
		\nonumber
		& \leq 
		d_N 
		\|	\O 	\|_{B(L^2)}
		\|	(\N+\1)^{-K_{\frac{1}{2}}}  \Omega_1(t)	\|_\F
		\|	(\N+\1)^{K_{	\frac{1}{2}	} +1 }         \Omega_1(t)	\|_\F \\
		& + 
		2 d_N 
		\sqrt N \|	 \O \vp_t	\|_{L^2} 
		\|	(\N+\1)^{-K_{\frac{1}{2}}}  \Omega_1(t)	\|_\F
		\|	(\N+\1)^{K_{	\frac{1}{2}	}+ \frac{1}{2}   }         \Omega_1(t)	\|_\F  \ . 
	\end{align}
	In order to conclude, 
	we control the right hand side of the above equation
	with the estimates \eqref{thm1 eq 2} and \eqref{thm1 eq 3}. 
	Using that $\|	 \O \vp	\|_{L^2} \leq \|	 \O	 \|_{B(L^2)}	$
	and
	recalling the space of compact operators is dual to the space
	of trace-class operators, we find 
	\begin{equation}
		\|    \gamma_B(t)   -  N    \ket{\vp_t}  \bra{\vp_t}        		\|_{\tr}
		\leq 
		K(t) \sqrt N \ . 
	\end{equation}
	This finishes the proof of the Theorem. 
\end{proof}


\section{Quantum Optimal Transportation: Proof of Theorem \ref{thm2}}
\label{section QOT}
In this section, we address the semi-classical limit of the coupled 
Hartree-Hartree system first introduced in \eqref{hartree hartree}.
Here, we adopt the scaling regime presented in \eqref{scaling regime}.
Thus, we study the 
solution
$(\omega^\hbar,\vp^\hbar)$ 
of the equation 
\begin{align}
	\label{hartree hartree 2}
	\begin{cases}
		& i\hbar \partial_t \omega^\hbar 
		= \Big[ -  \frac{1}{2}\hbar^2 \Delta + V*\rho_B , \omega^\hbar    \Big]  \\
		& i \partial_t \vp^\hbar   =  -  \frac{1}{2}\Delta \vp^\hbar  + (V*\rho_F) \vp^\hbar  \ . 
	\end{cases}
\end{align}
where $\rho_B(t,x) =  |\vp^\hbar (t,x)|^2$
and
$\rho_F (t,x) = M^{-1} \omega^\hbar (t;x,x)$
are the bosonic and fermionic position densities, 
with some prescribed initial data
$(\omega_0^\hbar , \vp_0^\hbar) \in \mathscr L ^1(L^2(\R^d))\times L^2(\R^d)$. 
We always assume that $\tr \omega_0^\hbar = M = \hbar^{-d}$
and $ (\omega_0^\hbar)^* = \omega_0^\hbar \geq 0$; 
these properties are of course preserved by the flow generated 
in \eqref{hartree hartree 2}. 

\vspace{2mm}

In order to analyze the macroscopic limit, 
given  $\hbar>0$ we denote the   \textit{Wigner transform} of a 
reduced density matrix $\omega \in \mathcal{L}^1(L^2(\R^d))$ by 
\begin{equation}
	W^\hbar [\omega] (x,p)
	\equiv 
	\frac{ \hbar^d }{(2 \pi)^{ d}} 
	\int_{\R^d }
	\omega  
	\Big(
	x +    \frac{\hbar }{2 } \xi  , x -   \frac{\hbar }{2 } \xi
	\Big) 
	e^{		 - i  \xi \cdot p	}		 \d \xi    \ , 
	\qquad    ( x,p )\in \R^{2d} \ . 
\end{equation}
We understand  the above map as a unitary transformation 
$W^\hbar : L^2(\R^d_x \times \R^d_{x'}) \rightarrow L^2 ( \R^d_x \times \R^d_p )$. 
Additionally, we record here the associated inverse map.
Namely, for a regular enough phase-space $f = f(x,p)$ distribution, 
we consider the operator on $L^2(\R^d)$
whose kernel is defined as 
\begin{equation}
	\t{Op}_w^\hbar[f] (x,x' )
	\equiv 
	\hbar^{-d}
	\int_{\R^d}
	f\Big(     \frac{x+x' }{2} , p         \Big) 
	e^{ i p \cdot (x-x')/ \hbar}  \d p  \ , 
	\qquad    ( x,y)\in \R^{2d} 
\end{equation}     
to be the \textit{Weyl quantization} map.

Intuitively, two phenomena happen when the $\hbar \downarrow  0 $ limit is taken.
First, the \textit{dynamics} (i.e. the time evolution of the system)      is changing. 
Secondly, the \textit{initial data} of the systems is converging from the quantum-mechanical, to the classical regime. 
We shall study each process separately.

\subsection{Stability of Hartree-Hartree}
Here and in the rest of the article, we shall be using the following notation
for a special class of semi-classical operators on $L^2(\R^d)$. 
Namely
\begin{equation}
	\label{semiclassical operator}
	\O_{\xi , \eta} \equiv \exp 
	\big(
	i   \xi \cdot x  
	+ 
	i 	  \eta \cdot  p  
	\big)
	,  \qquad \xi , \eta \in \R^d
\end{equation}
where $p = - i \hbar \nabla_x$.  
If $\omega \in \mathscr{L}^1(  L^1(\R^d)  )$ is a trace-class operator, 
then it is possible to define the additional (weaker) norm, which we shall use to measure distances 
\begin{equation}
	||| \omega  |||_s 
	\equiv  
	\sup_{\xi, \eta \in \R^d } 
	(1 + |\xi |  + |\eta|)^{-s }
	| 	\t{Tr} (  \O_{\xi , \eta} \omega )| 
	\ , \qquad s  \geq 0 \ . 
\end{equation}
The motivation for the introduction  of such norm comes from the following observation.
If we denote by $ f = W^\hbar [\omega]$ its Wigner transform,  then it is well-known that 
\begin{equation}
	\hat{f}(\xi , \eta ) =  \frac{1}{ M }	\t{Tr} (  \O_{\xi , \eta} \omega ) \ , \qquad \forall \xi, \eta \in \R^d 
\end{equation}
and, consequently, 
\begin{equation}
	\label{formula isomorphism}
	|f|_s  =\frac{1}{M}  ||| \omega |||_s \ , \qquad \forall s \geq  0 \  ; 
\end{equation}
we remind the reader that the norm $ | \cdot |_s $ has been defined in   \eqref{fourier based norm definition}. 
It is important to note that  if $ \tr |\omega| \leq N$, then  $|f|_s \leq 1 $.

\vspace{3mm}

Our first result towards the prove of Theorem \ref{thm2}
is the stability of the Hartree dynamics 
with respect to norm  $  ( |||  \cdot |||_1 , \|	 \cdot \|_{L^2})$. 
Let us note here that this result uses the ideas
presented in \cite{BenedikterPortaSaffirioSchlein2016}, 
although the result itself is nor stated or proved.

\begin{proposition}[Stability of Hartree-Hartree]
	\label{proposition stability hartree}
	Let $(\omega_1^\hbar, \vp_1^\hbar)$ and 
	$(\omega_2^\hbar , \vp_2^\hbar )$
	be two solutions of the coupled Hartree system   
	\eqref{hartree hartree}
	with initial data
	$ (\omega_{1,0}^\hbar , \vp_{1,0}^\hbar)$
	and
	$ (\omega_{2,0}^\hbar , \vp_{2,0}^\hbar) $, 
	satisfying
	$\tr \omega_{i,0}^\hbar  =M$
	and
	$( \omega_{i,0}^\hbar)^* 
	=
	\omega_{i,0}^\hbar 
	\geq0$
	for $i=1,2$.
	Then, there exists $C>$ such that  
	\begin{align}
		\frac{1}{M}	|||   \omega_1^\hbar  (t) -  \omega_2^\hbar (t) |||_1 
		+ 
		\| \vp_1^\hbar (t)  & - \vp_2^\hbar (t)		\|_{L^2}  \\ 
		\nonumber 
		& 	  \leq 
		C 
		\exp ( \exp C|t|)
		\Big(
		\frac{1}{M}	|||   \omega_{1,0}^\hbar  -  \omega_{2,0}^\hbar  |||_1 
		+ 
		\| \vp_{1 , 0 }^\hbar  - \vp_{2,0}^\hbar 	\|_{L^2}
		\Big) \  
	\end{align}
	for all $  t \in \R$
\end{proposition}

\begin{proof}
	In what follows, we shall drop the $\hbar$ superscript in order to ease the notation. 
	Let us  then 
	consider the unitary evolution groups associated to the solutions 
	$\{ (\omega_i , \vp_i)\}_{i=1,2}$. 
	That is, for $i=1,2$,
	in the notation of Appendix \ref{section analysis PDEs}, 
	we consider 
	\begin{equation}
		\omega_i(t) 
		= 
		U_{F,i}^* (t) \omega_{i.0}
		U_{F,i} (t)
		\qquad
		\t{and}
		\qquad 
		\vp_i(t) = U_{B,i}(t)
		\vp_{i,0}
	\end{equation}
	where $t\in\R$. 
	Then, a straightforward computation 
	using the generators of $U_{F,1}$
	and
	$U_{F,2}$ yields 
	\begin{align}
		i \hbar \frac{d}{dt} 
		U_{F,1}^*(t)
		\Big(
		\omega_{1} (t) -  \omega_{2} (t)
		\Big) 	
		U_{F,1}(t) 
		= 
		U_{F,1}^*  (t)
		\Big[
		V* \big(
		\rho_{B,1} (t)  -  \rho_{B,2} (t) 
		\big)
		, \omega_2(t)  
		\Big] 
		U_{F,1}(t)  \ . 
	\end{align}
	Here, $\rho_{B,i} (t,x) = | \vp_i(t,x)|^2$. 
	Thus, take a semi-classical observable $\O_{\xi , \eta}$ and compute the trace
	\begin{align}
		\nonumber 		
		\tr \O_{\xi , \eta}
		\Big(
		\omega_1 (t) -  & \omega_2 (t)
		\Big)    = 
		\tr \O_{\xi , \eta}
		\big( 
		\omega_{1,0} -  \omega_{2,0}
		\big) \\ 
		\nonumber 
		&  \qquad 
		\frac{ - i }{\hbar}
		\int_0^t
		\tr
		\Big( 
		\O_{\xi , \eta} 
		U_{F,1}(t-s) 
		\big[
		V* \big(
		\rho_{B,1} (s)  -  \rho_{B,2} (s) 
		\big)
		, \omega_2(s)  
		\big] 
		U_{F,1}^*(t-s) 
		\Big) 
		\d s  \ . 
	\end{align}
	Clearly, 
	$ 
	| \tr \O_{\xi , \eta}
	( 
	\omega_{1,0}-  \omega_{2,0}
	) | 
	\leq 
	(1 + | \xi|   +  | \eta|  )  \   |||    \omega_1 (0) - \omega_2 (0) |||_1 
	$.  
	For the second term, 
	we use ciclicity of the trace
	$\tr AB = \tr BA$, 
	combined with 
	the following  general
	Fourier series expansion
	\begin{equation}
		[ V * \rho , \omega ] 
		= 
		\int_{\R^d}
		\hat V (k)
		\hat\rho(k) 
		[e_k , \omega] \d k 
		\qquad
		\t{where}
		\qquad 
		e_k(x) 
		\equiv (2\pi)^{-d/2}
		e^{i k \cdot x} \ . 
	\end{equation}
	We obtain
	\begin{align}
		\nonumber
		| 	\tr
		\Big( 
		& 	\O_{\xi , \eta} 
		U_{F,1}
		(t-s) 
		\big[
		V* \big(
		\rho_{B,1} (s)  -  \rho_{B,2} (s) 
		\big)
		, \omega_2(s)  
		\big] 
		U_{F,1}^*(t-s) 
		\Big) 
		|  \\ 
		\nonumber
		& \leq 
		\int_{	\R^d	}  | \hat V (k)| \ 
		|  \hat  \rho_{B,1} (s,k) - \hat \rho_{B,2} (s,k)    | 
		\  |	 \tr  \Big(
		\big[   	e_k , 
		U_{F,1}^*(t-s)	 
		\O_{\xi , \eta}
		U_{F,1} (t-s)  \big]
		\Big)	 \omega_2 (s) 			| \  d k  \\ 
		& \leq \hbar M  (1 + |\xi| + |\eta| ) 
		\exp (C |t-s |)
		\|	 |k| \hat V	\|_{L^1}
		\|	\hat \rho_{B,1} (s) - \hat\rho_{B,2 }(s)	\|_{L^\infty  } \ . 
	\end{align}
	In the second  line, we   used the commutator estimate  from Lemma \ref{lemma commutator 1}, 
	with 
	$U_\rho  =  U_{F,1}$.
	Next, we use 
	$
	\|	\hat \rho_{B,1}  - \hat\rho_{B,2 }	\|_{L^\infty  }
	\leq 
	\|	  \rho_{B,1}  - \rho_{B,2 }	\|_{L^1 }  \leq  2 \|	 \vp_1 - \vp_2\|_{L^2} $ and analyze the boson fields. 
	A similar argument shows that 
	\begin{align}
		\nonumber 
		\|	 \vp_1 (s) - \vp_2(s)	\|_{L^2}
		& 	\leq 
		\|	 \vp_1 ( 0 ) - \vp_2( 0 )	\|_{L^2}
		+ 
		\int_0^s
		\|    V*(  \rho_{F,1} (r) -  \rho_{F,2} (r)  )    \vp_2 (r)		\|_{L^2} d r
		   \\ 
		\nonumber 
		& \leq 
		\|	 \vp_1 ( 0 ) - \vp_2( 0 )	\|_{L^2}
		+ 
		\|    |k| \hat V  \|_{L^1}	
		\int_0^s \d r   \sup_{ k \in \R^d }
		|k|^{-1} 
		\big|
		\hat  \rho_{F,1} (r,k ) -  \hat \rho_{F,2} (r,k)  
		\big| 	   
		\end{align} 
from which we obtain
\begin{align} 
		\|	 \vp_1 (s) - \vp_2(s)	\|_{L^2}		& \leq 
		\|	 \vp_1 ( 0 ) - \vp_2( 0 )	\|_{L^2}
		+ 
		\frac{ 	\|    |k| \hat V  \|_{L^1}	}{M}
		\int_0^s 
		|||  \omega_{1} (r)   - \omega_2 (r)    |||_{1 } \d r 
	\end{align}
	where we used $\hat \rho_F (k) = M^{-1 } \tr \O_{k,0} \omega $
	We can now close the inequalities and apply the Gronwall inequality.
\end{proof}

\subsubsection{Borrowed Commutator Estimates}
In this subsubsection, we state a result that we used in the proof of the above stability estimate. 
They concern   the propagation-in-time 
of a certain class of commutator estimates.
These results were originally proved in 
\cite[Lemma 4.2]{BenedikterPortaSaffirioSchlein2016}
for 
an interacting system of fermions in the  combined semi-classical
and mean-field regime. 
Since the proofs can be easily adapted to the present case, we shall omit them. 

\vspace{2mm}

In order to state    it, we introduce the following unitary dynamics. 
Namely, 
given a time-dependent position density 
$\rho : \R \rightarrow  L^1(\R^d) $ satisfying
$\rho (t,x )\geq0$
and $\int\rho (t,x) dx =1 $
we consider 
\begin{equation}
	\label{fermion flow}
	\begin{cases}
		&  i \hbar \partial_t 
		U_\rho 
		(t,s ) 
		=
		\Big(
		- \frac{\hbar^2}{2} \Delta
		+ 
		V*\rho(t) 
		\Big)
		U_\rho  (t,s) \\ 
		& 
		U_\rho   (t,t) = \1  \ . 
	\end{cases} 
\end{equation}
In the interacting fermion system, one chooses
$\rho(t,x) = N^{-1 } \omega (t ; x,x)$
whereas in our case, 
the density corresponds to that of bosons
$\rho(t,x) = | \vp(t,x)|^2$.
Since the estimates are quite robust with respect to the choice of density, 
we state here the result for arbitrary $\rho$.  One only needs 
$\sup_{t}  |\hat \rho (t )|_{L^\infty}< \infty$.

\begin{lemma} 
	\label{lemma commutator 1}
	Assume that 
	$\int_{\R^d}  \< \xi\>^2 |\hat V (\xi) |\d \xi  < \infty$
	and let $ U_\rho (t,s) $
	be the unitary flow defined in \eqref{fermion flow}. 
	Then, there exists $C>0$ such that 
	for all $(\xi , \eta ) \in \R^d \times \R^d $  
	and $t,s \in \R $
	there holds 
	\begin{equation}
		\sup\Biggl\{
		\frac{1}{ | k | } 
		\Big|  
		\mathrm{Tr} 
		\big[
		e^{ i k \cdot x }, 
		U_\rho^* (t,s) 
		\O_{\xi , \eta}
		U_\rho  (t,s)
		\big]
		\omega 
		\Big|
		\ : \ 
		k \in \R^d \ , 
		\mathrm{Tr}|\omega| \leq 1
		\Biggr\}
		\leq 
		\hbar ( | \xi| + |\eta |) 
		e^{ C  |t -s |   }
	\end{equation}
	where $ \O_{\xi , \eta}$ is the semi-classical observable   defined in \eqref{semiclassical operator}. 
\end{lemma}

	\subsection{Quantum Optimal Transportation}
	A lot of technology has been developed in the last decade in the context of Quantum Optimal Transporation,
	by means of the introduction of relevant pseudo-metrics.
	In this paper, we apply and adapt some of these results to the study of the problem at hand.

	\vspace{2mm}
	In this context, one is given a probability measure $f \in \mathscr{P}(\R^{2d})$
	describing a classical state, and wishes to compare it to a quantum state $\omega$,
	belonging to the space
	\begin{equation}
		\label{P h}
		\mathcal P (L^2(\R^{d}))
		\equiv
		\Bigl\{
		\omega \in \mathscr L ^1  (  L^2(\R^d) )
		\, : \, 
		\omega = \omega^* \geq 0 
		,  \ 
		\tr \omega = M = \hbar^{-d}
		\Bigr\} \ . 
	\end{equation}
	Let us immediately note here that our normalization is different than most results on the literature, 
	where one considers trace-class operators $\calR  $
	with  $\tr \calR = 1 $.
	This is of course only a matter of scaling, 
	and the passage from one to the other is   given by  $\calR = \hbar^d \omega$. 
	In particular, we choose the normalization in \eqref{P h}
	because it is more natural for the problem at hand. 
	
	\vspace{2mm}
	Let us now introduce the concept of a \textit{coupling}, between classical and quantum states		
	\begin{definition}[Coupling]
		Given $f \in \mathscr P (\R^{2d})$ 
		and
		$
		\omega \in 
		\mathcal P (L^2(\R^{d}))
		$, 
		we say that the operator-valued function
		$Q : \R^{2d}_{x,p} \rightarrow \mathscr{B}( L^2(\R^{2d}) )$
		is a \textbf{coupling} for $f$ and $\omega$, if the following is satisfied 
		\begin{enumerate}
			\item  For almost every $(x,p) \in \R^{2d}$ there holds $Q(x,p) = Q(x,p)^* \geq 0 $.
			\item $\int_{\R^{2d} } Q(x,p) \d x \d p = \omega .$
			\item $\tr Q(x,p) = \hbar^{-d } \, f(x,p).$
		\end{enumerate}
		The set of all couplings between $f $ and $\omega$ is denoted by $\mathcal C (f,\omega). $
	\end{definition}
	
	Throughout this section, 
	we denote by $\hat x : \mathscr D (\hat x) \subset L^2(\R^d) \rightarrow L^2(\R^d)$ 
	the standard multiplication operator on $L^2$, 
	and similarly for $\hat p = - i \hbar \nabla_x   $ on $H^1(\R^d).$
	We introduce the following cost function, 
	taking values 
	in the space of unbounded self-adjoint 
	operators on $L^2(\R^d)$
	\begin{equation}
		c_\hbar (x,p) \equiv \frac{1}{2} (x - \hat x)^2 +  \frac{1}{2} ( p - \hat p)^2  \ , \qquad (x,p)\in \R^{2d}  , 
	\end{equation}
	initially defined on $\mathscr S (\R^d)$ and then closed in $L^2$.
	Let us recall  that  we denote by $\mathscr{P}_n(\R^{2d})$
	the space of measures with finite $n$ moments.
	Similarly, we denote 
	by 
	$
	\mathcal P _n (L^2(\R^{d}))
	$
	the space of 
	quantum states $\omega 	 \in 	\mathcal P  (L^2(\R^{d}))$
	such that 
	$ \tr \big[ \sqrt \omega    ( \hat x^2 + \hat p ^2)^n   \sqrt \omega \big]  < \infty \ .  $

	\begin{definition}[Quantum Wasserstein]
		For all $f \in \mathscr P _2 (\R^d)$ 
		and 
		$\omega \in 		\mathcal P _2 (L^2(\R^{d})) $
		we define the \textbf{second quantum Wasserstein distance}
		as the quantity
		\begin{equation}
			E_\hbar (f, \omega)
			\equiv
			\hbar^{d/2}
			\inf_{Q \in \mathcal C (f,\omega)}
			\Big(
			\int_{	\R^{2d}	} \tr Q(x,p) c_\hbar (x,p)
			\d x \d p 
			\Big)^{1/2}
			\in [0,\infty] \ . 
		\end{equation}
	\end{definition}
	
	Up to scaling, 		the functional $E_\hbar$ has been the object of several studies in recent years.  
	In particular, it has been proven that it is a natural object to study when comparing the dynamics 
	of the Hartree and Vlasov equation, for system of interacting fermions.
	We adapt the proof of 
	\cite[Theorem 2.5]{GolsePaul2017}
	for the present case of interest.

	\begin{theorem}
		\label{proposition wasserstein}
		Assume $V\in C^{1,1} (\R^d,\R)$.	
		Let $(\omega_\hbar (t),\vp_\hbar (t))$ solve the Hartree-Hartree equation \eqref{hartree hartree 2} with initial data
		$(\omega_0^\hbar  , \vp_0^\hbar)$.
		Further, let 
		$(f (t) , \vp(t))$ solve the Vlasov-Hartree equation, with initial data
		$(f_0 , \vp_0) $. 
		Then, there exists $C = C (V)>0$ such that for all $ t \in \R$ there holds 
		\begin{equation}
			\label{theorem 3 rhs}
			E_\hbar  \big(  f(t) , \omega_\hbar (t) \big)
			+ 
			\|	 \vp(t)   -  \vp_\hbar (t)	\|_{L^2} 
			\leq 
			C\exp(C t^2 )
			\Big(
			E_\hbar  
			(f_0 , \omega_0^\hbar )
			+ 
			\|	 \vp_0   -  \vp_0^\hbar	\|_{L^2} 
			\Big) \ . 
		\end{equation}
	\end{theorem}
	
	Two questions remain. 
	The first one is: what norms are    bounded above by  $E_\hbar (f,  \omega_\hbar ) $? 
	The second one is: is the right hand side of \eqref{theorem 3 rhs} small when $\hbar \downarrow 0 $? 
	Neither of these questions has a trivial answer. 
	Fortunately, they have already been answered quite recently in the literature.
	In order to state it, we must introduce
	the following Gaussian mollifier at scale $\hbar>0$
	\begin{equation}
		\mathscr G_{\hbar} (z) 
		\equiv \hbar^{ -d/4}
		\mathscr G_1( \hbar^{-1 /2} z)
		\quad
		\t{with}
		\quad 
		\mathscr G_1 (z)
		\equiv 
		\pi^{-d/4}
		\exp(-z^2/2)
	\end{equation}
	for $z = (x,p) \in \R^{2d}$. 
	In our notation and scaling, we record the relevant results in the following lemma.

	\begin{lemma} 
		\label{lemma wasserstein}
		The following statements are true. 
		\begin{enumerate}[leftmargin=*]
			\item 
			
			\cite[Corollary 1.1]{Lafleche2023}	
			Let $d \geq 2 $, 
			$\omega \in \mathcal P_2 (L^2(\R^d))$ 
			and 
			$f \in \mathscr P_2 (\R^{2d})$.
			Assume that $  \omega \leq 1$. 
			Then, there exists $C = C(d)$ such that 
			\begin{equation}
				\|	 f -  W_\hbar [\omega ]		\|_{ \dot H ^{-1}   }
				\leq 
				E_\hbar [f , \omega]
				+ 
				C \sqrt \hbar \ . 
			\end{equation}
			
			\item 
			\cite[Theorem 2.4]{GolsePaul2017} 
			For all $ f \in \mathscr{P}_2(\R^{2d})$
			there holds
			\begin{equation}
				E_\hbar [f,  \t{Op}_w^\hbar (f * \mathscr{G}_\hbar )  ] \leq \sqrt{d\, \hbar}.	
			\end{equation}

		\end{enumerate}
	\end{lemma}
	
	\begin{remark}[Anti-Wick Quantization]
		So far, we have introduced the Weyl quantization map, as
		the inverse of the Wigner transform.
		As is well-known, this quantization map does not preserve positivity.
		On the other hand, the anti-Wick (or, Toeplitz) quantization does, and is defined as follows. 
		Given a phase-space distribution
		$\mu \in \mathscr P(\R^{2d})$, 
		one defines
		on $L^2(\R^{d})$ the bounded, self-adjoint operator 
		\begin{equation}
			\t{Op}^{\hbar}_{\t{aw}}
			(\mu) 
			\equiv 
			\t{Op}^\hbar_w 
			(\mu * \mathscr G _\hbar )
			=
			\hbar^{-d}
			\int_{\R^{2d}}
			\ket{      f_{x,p}^\hbar   }
			\bra{f_{x,p}^\hbar }
			\d \mu(x,p)
		\end{equation}
		where $f_{x,p}^\hbar(y) \equiv  \hbar^{-d/4} g (  \hbar^{-1/2} (x-y)  ) \exp(i p \cdot y / \hbar)$		
		is a coherent state at scale $\hbar>0$, 
		with an $L^2$-normalized Gaussian profile $g $. 
		Note that preservation of positivity follows immediately from the last formula, 
		since it is the convex combination of positive operators (in this case, orthogonal projections). 
		In particular, it follows from the previous Lemma
		that
		$E_\hbar \big[      f ,   \t{Op}^{\hbar}_{\t{aw}}
		(f)   \big] \leq \hbar^{1/2}$. 
	\end{remark}
	Let us now prove Theorem  \ref{proposition wasserstein}.
	Since the   proof is a simple adaptation of that of 
	\cite[Theorem 2.5]{GolsePaul2017} 
	in which we replace the self-interacting term, 
	with the interaction with the boson field $\vp(t)$, 
	we only provide the sketch of the proof and refer to the original reference for the details. 
		\begin{proof}[Proof of Theorem \ref{proposition wasserstein}]
		First,  we   introduce some notation we shall make use of throughout the proof. 
		Namely, letting  
		$( \omega^\hbar (t) , \vp^\hbar (t))$ 
		and
		$(f(t) ,\vp(t))$
		be as in the statement of the Theorem, we let 
		the bosonic densities be 
		\begin{equation}
			\rho_f(t,x) \equiv 
			|\vp (t,x)|^2 
			\quad
			\t{and}
			\quad 
			\rho_\omega (t,x)
			\equiv
			|\vp^\hbar  (t,x)|^2 
			\ , \qquad \forall (t,x) \in \R\times \R^d  \ . 
		\end{equation}
		The notation is made such that they enter the equations for $f$ and $\omega$, respectively.

		\vspace{2mm}

		\textit{Step 1. Propagation of moments}.
		Let $f_0 \in \mathscr{P}_2(\R^{2d})$ be the initial datum of the Vlasov-Hartree system.
		Then,  it holds true
		that $f(t) \in \mathscr{P}_2 (\R^{2d})$. 
		Indeed, let $\Phi_\vp$ be the ODE characteristics map 
		from Theorem \ref{theorem wp 2}. 
		Then, the claim follows from the formula $f(t,z) = f_0 \circ \Phi_{\vp}^{-1}(t,z)$ and changing variables
		$z = (x,p)\in \R^{2d}$ in 
		\begin{equation}
			\int_{	\R^{2d}	} |z|^2  f(t,z) \d z 
			= 
			\int_{	\R^{2d}	}   |    \Phi_\vp (t,z) |^2  f_0(z) \d z  
			\leq 
			\kappa(t)^2 
			\int_{	\R^{2d}	}  |z|^2     f_0(z) \d z   
			< \infty \ . 
		\end{equation}
		Here, we have used that the estimate
		$
		| \Phi_\vp (t,z)| \leq  \kappa(t)  |z|,
		$ 
		see Remark \ref{remark grow}. 
		\vspace{2mm}
		
		\textit{Step 2. Choice of  Coupling}. 
		Given $Q_0^\hbar \in \mathcal C (f_0 , \omega_0)$
		we consider $Q^\hbar(t)$ to be the time-dependent coupling solving the PDE 
		\begin{equation}
			\label{coupling dynamics}
			\begin{cases}
				& 	\partial_t Q^\hbar 
				+ 
				\{ 
				\frac{1}{2}  p^2  + V*\rho_f (t,x)  , 
				Q^\hbar 
				\}
				+ 
				\frac{i}{\hbar}
				\big[
				\frac{1}{2}\hat p^2
				+ 
				V* \rho_\omega (t,\hat x)  , Q_\hbar 
				\big] = 0 \   ,   \\
				& Q^\hbar (0)  =  Q^\hbar_0 \ . 
			\end{cases}
		\end{equation}
		Here, we denote by $ \{ F,G\}  = \nabla_x F \nabla_p G - \nabla_p F \nabla_x G$ the Poisson bracket of two classical observables, 
		and we recall that $\hat p = - i \hbar \nabla_x$ and $\hat x$ are 
		the standard momentum and position observables in $L^2(\R^d)$.
		One may define $Q^\hbar (t)$  by means of   conjugation of a unitary map, and composition of a vector field--see for details. 
		In particular, it follows from  such representation and    
		\cite[Lemma 4.2]{GolsePaul2017}
		that this procedure actually defines a coupling between $f(t)$ and $\omega^\hbar (t)$. 
		Namely,  $Q_\hbar (t) \in \mathcal{C} (f(t) , \omega^\hbar  (t))$ for all $t\in \R $.

		\vspace{2mm}
		
		\textit{Step 3. Dynamics of the coupling.}
		We now estimate the growth of the second order moments of the coupling $Q_\hbar (t)$. 
		Namely, we define 
		\begin{equation}
			\mathcal E _\hbar (t)
			\equiv 
			\hbar^d 
			\int_{	\R^{2d}	} \tr \big[
			c_\hbar (x,p) Q_\hbar (t;x ,p)
			\big] \d  x \d  p \  , \qquad \forall t \in \R \   
		\end{equation}
		and compute its time derivative as follows. 
		In view of \eqref{coupling dynamics}
		we find that 
		\begin{align}
			\nonumber 
			\hbar^{-d}
			\frac{d}{dt}
			\mathcal E _\hbar (t)
			& =
			\int_{	\R^{2d}	}
			\tr
			\Big[
			Q_\hbar (t;x,p)
			\Big\{ 
			\frac{1}{2} p^2  + V*\rho_f (t,x)  , 
			c_\hbar (x,p)
			\Big\}
			\Big]
			\d x \d p  \\ 
			&   +  
			\frac{i}{\hbar}
			\int_{	\R^{2d}	}
			\tr
			\Big[
			Q_\hbar (t;x,p)
			\big[
			\frac{1}{2}\hat p^2
			+ 
			V* \rho_\omega (t,\hat x)  , c_\hbar (x,p)
			\big] 
			\Big]
			\d x \d p 
		\end{align}
		where we have integrated by parts and put both  brackets on the cost function.
		Following 
		\cite{GolsePaul2017}
		we calculate explicitly the above brackets and find that 
		\begin{align}
			\big\{
			\frac{1}{2}p^2 + V*\rho_f (t,x) , c_\hbar (x,p)
			\big\} 
			& 	 = 
			p \cdot (x - \hat x)
			- \nabla V * \rho_f (t,x) (p - \hat p )  \\ 
			\frac{i}{\hbar}	\big[
			\frac{1}{2}\hat p ^2 + V*\rho_\omega (t,\hat x ) , c_\hbar (x,p)
			\big]
			& 	 = 
			- \frac{1}{2} (x - \hat x)  \hat p  -  \frac{1}{2}  \hat p (x - \hat x )  \\ 
			\nonumber 
			&  +   		  \frac{1}{2} \nabla V*\rho_\omega(t,\hat x) \cdot  ( p - \hat p )
			+     \frac{1}{2}		( p - \hat p )	\cdot  \nabla V*\rho_\omega(t,\hat x)    \ .  
		\end{align}
This implies the identity 
		\begin{align}
			\nonumber 
			\hbar^{-d}
			\frac{d}{dt}
			\mathcal E _\hbar (t)
			& =
			\int_{	\R^{2d}	}
			\tr
			\Big[
			Q_\hbar (t;x,p)
			\frac{1}{2}
			\Big(
			(x - \hat x) \cdot (p - \hat p)
			+ 
			(p - \hat p)\cdot  (x - \hat x) 
			\Big)
			\Big]
			\d x \d p 
			\\ 
			\nonumber 
			&   +  
			\int_{	\R^{2d}	}
			\tr
			\Big[
			Q_\hbar (t;x,p)
			\frac{1}{2}
			\Big(
			\nabla V * \rho_\omega (t,\hat x) 
			- 
			\nabla V  * \rho_f (t,x)
			\Big)
			\cdot(p-\hat p )
			\Big]
			\d x \d p \\
			&  + 
			\int_{	\R^{2d}	}
			\tr
			\Big[
			Q_\hbar (t;x,p)
			\frac{1}{2}
			(p-\hat p )
			\cdot 
			\Big(
			\nabla V * \rho_\omega (t,\hat x) 
			- 
			\nabla V  * \rho_f (t,x)
			\Big)
			\Big]
			\d x \d p  \ .
			\end{align}
			 Thus, 
 a straightforward manipulation,
			 combined  with 
			 $A^*B + B^*A \leq  A^*A + B^*B   $
			 yields  
			 \begin{align}
			\nonumber 
\hbar^{-d}
\frac{d}{dt}
\mathcal E _\hbar (t)
			& \leq 
			C 
			\hbar^{-d}
			\mathcal E _\hbar (t) 
			+ 
			\frac{1}{2}
			\int_{	\R^{2d}	}
			\tr  \Big[
			Q_\hbar (t;x,p)
			|\nabla V * \rho_\omega  (t,\hat x)     -    \nabla V *\rho_f(t,x)  |^2 
			\Big] \d x \d p  \ . 
		\end{align}
		The second	term in the last displayed  equation  can be estimated as follows.
		We use the triangle  inequality and the fact that $Q_\hbar (t)$ is a coupling for $f (t)$
		to find that 
		\begin{align}
			\nonumber 
			\int_{	\R^{2d}	}
			\tr  \Big[
			&Q_\hbar (t;x,p)
			|\nabla V * \rho_\omega  (t,\hat x)     -    \nabla V *\rho_f(t,x)  |^2 
			\Big] \d x \d p    \\ 
			\nonumber 
			& 	\leq 
			\int_{	\R^{2d}	}
			\tr  \Big[
			Q_\hbar (t;x,p)
			|\nabla V * \rho_\omega  (t,\hat x)     -    \nabla V *\rho_\omega (t, x)  |^2 
			\Big] \d x \d p   \\ 
			\nonumber 
			& 	+ 
			\int_{	\R^{2d}	}
			\tr  \Big[
			Q_\hbar (t;x,p)
			\Big] 
			\, 
			| \nabla V * \rho_\omega  (t,  x)     -    \nabla V *\rho_f(t,x)  |^2 
			\d x \d p    \\ 
			\nonumber 
			& \leq 
			| \nabla V   |_{Lip}^2   \|  \rho_\omega(t)	\|_{L^1}^2  \hbar^{-d } \mathcal E _\hbar (t) 
			+ 
			\|	 \nabla V * \rho_\omega  (t)     -    \nabla V *\rho_f(t)  	\|_{L^\infty}^2 
			\|	f (t)	\|_{L^1 }   \\ 
			& \leq 
			| \nabla V   |_{Lip}^2   \hbar^{-d }   \mathcal E _\hbar (t) 
			+ 
			\|   \nabla V \|_{L^\infty }^2 
			\|	 \vp(t) -  \vp^\hbar (t)	\|_{L^2}^2
		\end{align}
		We can now collect the previous estimates and   integrate in time
		to find that 
		\begin{equation}
			\label{gronwall 1}
			\hbar^{-d}
			\mathcal{E}_\hbar (t)  
			\leq 
			\hbar^{-d}	\mathcal{E}_\hbar (0) 
			+ 
			C 
			\int_0^t
			\Big(
			\hbar^{-d}
			\mathcal{E}_\hbar (s ) 
			+ 
			\|	 \vp(s) -  \vp^\hbar (s)	\|_{L^2}^2 
			\Big) \d s \ . 
		\end{equation}
		It suffices now to estimate the difference in the $L^2$ norm of the boson fields. 
		
		\vspace{2mm}
		
		\textit{Step 4. Estimates on the boson densities}.
		The equations for the boson fields can be written 
		in mild formulation as follows 
		\begin{align}
			\vp (t)  & = 
			e^{  -  i t \Delta /2  }
			\vp_0 
			- i 
			\int_0^t 
			e^{    - i (t-s) \Delta /2    } 
			V*  \widetilde \rho _f (s )   \vp(s) \d s  \\
			\vp^\hbar (t)  & = 
			e^{  -  i t \Delta /2  }
			\vp_0^\hbar  
			- i 
			\int_0^t 
			e^{    - i (t-s) \Delta /2    } 
			V*  \widetilde \rho _\omega (s)   \vp^\hbar (s) \d s    
		\end{align}
		where we denote for the fermion densities 
		$$\widetilde \rho_f (t,x) = \int_{\R^d} f(t;x,p) \ d p 
\quad 		\t{and}\quad 
		 \widetilde \rho _\omega (t,x) = (1/N) \omega(t;x,x).$$
		We estimate the $L^2$ norms  as follows 
		\begin{align}
			\nonumber
			\|	 &  \vp(t)  -  \vp^\hbar (t)     \|_{L^2} 
			\\ 
			\nonumber
			& 	\leq 
			\|	 \vp_0 -  \vp^\hbar_0    \|_{L^2}
			+ 
			\int_0^t 
			\|		  (V* \widetilde \rho_f (s)   -   V* \widetilde \rho_\omega (s)      )   \vp(s) 	   \| \d s  	 + 
			\int_0^t 
			\|		    V* \widetilde \rho_\omega (s)     (   \vp(s) 	 - \vp^\hbar (s))    \| \d s  \  \\ 
			\nonumber
			& \leq 
			\|	 \vp_0 -  \vp^\hbar_0    \|_{L^2} 
			+ 
			\int_0^t \|		V* \big(     \widetilde \rho_f (s) -  \widetilde \rho_\omega (s)  	 		\big)   	 \|_{L^\infty  }  \d s 
			+  
			\|	 V	\|_{L^\infty } 
			\int_0^t 
			\|	 \vp(s)  -  \vp^\hbar (s)     \|_{L^2} \d s \ . 
		\end{align}
		Finally, we notice that because $Q_\hbar (t)$ is a coupling between $f(t)$ and $\omega(t)$,
		we obtain that  for all $X \in \R^d$   
		\begin{align}
			V* \widetilde \rho_f (s,X) 
			-  
			V* \widetilde \rho_\omega (s,X)  
			&   =     \int_{	\R^{2d}	}
			\tr
			\Big[
			Q_\hbar (t;x,p)
			\big(
			\nabla V(X -x )
			- 
			\nabla V(X - \hat x )
			\big)
			\Big] \d x \d p   
		\end{align}
		and, consequently, using the Cauchy Schwarz inequality
		\footnote{In the following form 
			$
			\int_{	\R^{2d}	} \d z  \tr ( A(z)^* B(z)  ) \d z
			\leq 
			(\int_{	\R^{2d}	}   \tr (A^*(z)A(z))  dz     )^{1/2} 
			(\int_{	\R^{2d}	}   \tr (B^*(z)B(z))  dz     )^{1/2} 
			$} 
		we find 
		\begin{align}
			\nonumber 
			| V* \widetilde \rho_f (s,X) 
			-  
			V* \widetilde \rho_\omega (s,X)  |^2 
			& \leq  
			\Big| 	   \int_{	\R^{2d}	}
			\tr
			\Big[
			Q_\hbar (t;x,p)
			\big|
			\nabla V(X -x )
			- 
			\nabla V(X - \hat x )
			\big|^2 
			\Big] \d x \d p  \Big| 		 \\ 
			\nonumber 
			& \leq 
			\|	 \nabla V	\|_{Lip}^2
			\int_{	\R^{2d}	} 
			\tr
			\Big[
			Q_\hbar (t;x,p)  |x - \hat x |^2 
			\Big] \d x \d p \\ 
			& 	   \leq  2 
			\|	 \nabla V	\|_{Lip}^2 
			\hbar^{-d }
			\mathcal E_\hbar (t) \ . 
		\end{align}

		We conclude that for some $C = C(V)$ there holds
		\begin{equation}
			\label{gronwall 2.0}
			\|	   \vp(t)  -  \vp^\hbar (t)     \|_{L^2}  
			\leq 
			\|	 \vp_0 -  \vp^\hbar_0    \|_{L^2}  
			+ 
			C 
			\int_0^t
			\Big(
			\hbar^{ - \frac{d}{2}}				\mathcal E _\hbar (t)^{\frac{1}{2}} + 
			\|	   \vp(s)  -  \vp^\hbar (s)     \|_{L^2}  
			\Big) \d s \  . 
		\end{equation}
		In order to compare it with the inequality found in 			
		\eqref{gronwall 1}, we take the square of both sides in \eqref{gronwall 2.0}
		and use the Cauchy-Schwarz inequality to find that for some $C = C(V)>0$ there holds
		\begin{equation}
			\label{gronwall 2}
			\|	   \vp(t)  -  \vp^\hbar (t)     \|_{L^2}^2 
			\leq 
			C
			\Bigg( 
			\|	 \vp_0 -  \vp^\hbar_0    \|_{L^2} 
			+ 
			t
			\int_0^t
			\Big(
			\hbar^{-d }			\mathcal E _\hbar (t) 
			+ 
			\|	   \vp(s)  -  \vp^\hbar (s)     \|_{L^2}^2
			\Big) \d s 
			\Bigg) 
			\  . 
		\end{equation}

		\textit{Step 5. Conclusion.}
		In order to conclude our argument, we  put the inequalities \eqref{gronwall 1} and \eqref{gronwall 2} together
		to find that
		\begin{align}
			\nonumber 
			\hbar^{-d }				\mathcal E _\hbar (t)  
			+ 
			\|	   \vp(t)  -  \vp^\hbar (t)     \|_{L^2}^2 
			& 	\leq 
			C 
			\Big(
			\hbar^{-d }				\mathcal E _\hbar ( 0 )  
			+ 
			\|	   \vp_0  -  \vp^\hbar_0   \|_{L^2}^2 
			\Big)  \\ 
			& 				+ 
			C (1+t )
			\int_0^t
			\Big(
			\hbar^{-d }				\mathcal E _\hbar (s)  
			+ 
			\|	   \vp(s)  -  \vp^\hbar (s)     \|_{L^2}^2 
			\Big) \d s \  . 
		\end{align}
		Thus, we may  apply the Gr\"onwall inequality to find that 
		there exists $C>0$ such that 
		\begin{equation}
			\hbar^{-d }				\mathcal E _\hbar (t)  
			+ 
			\|	   \vp(t)  -  \vp^\hbar (t)     \|_{L^2}^2 
			\leq
			C \exp(Ct^2 )
			\Big(
			\hbar^{-d }				\mathcal E _\hbar ( 0 )  
			+ 
			\|	   \vp_0  -  \vp^\hbar_0   \|_{L^2}^2 
			\Big)  \ . 
		\end{equation}
		In order to conclude, 
		let us recall that
		$   				E_\hbar (f(t)  ,  \omega (t))^2
		\leq 
		\hbar^d 			\mathcal E _\hbar (t)  . 
		$ 
		Additionally,
		we recall  that $	 		\mathcal E_{\hbar }(0) $ is defined in terms of the initial datum 
		$Q^\hbar (0) = Q_0^\hbar \in \mathcal C (f_0 , \omega_0^\hbar)$, corresponding to an arbitrary coupling. 
		Therefore,
		we  minimize the right hand side over all couplings, 
		and take the square root,  to finally find that 
		\begin{equation}
			E_\hbar (f(t)  ,  \omega (t)) 
			+ 
			\|	   \vp(t)  -  \vp^\hbar (t)     \|_{L^2}  
			\leq 
			C \exp(C 
			t^2 
			)
			\Big(
			E_\hbar  
			(f_0 , \omega_0^\hbar )
			+ 
			\|	 \vp_0 -  \vp^\hbar_0    \|_{L^2}  
			\Big)  \ . 
		\end{equation}
		This finishes the proof.
	\end{proof}

	\subsection{Proof of Theorem \ref{thm2}}
	In this subsection, we combine the results  previously established in 
	Proposition \ref{proposition stability hartree} and 
	Theorem \ref{proposition wasserstein}.
	
	\begin{proof}[Proof of Theorem \ref{thm2}]
		Let us consider $(\omega_0^\hbar, \vp_0^\hbar)$, 
		$f_0^\hbar \equiv W^\hbar [\omega_0^\hbar]$, 
		and $(f_0, \vp_0)$
		as in the statement of Theorem \ref{thm2}.
		The proof is divided into two steps.
		In the first step, we consider the evolution given by the  Hartree-Hartree dynamics,
		and use Proposition \ref{proposition stability hartree}
		to change the initial data from $\omega_0^\hbar = \t{Op}_w^\hbar [ f_0^\hbar ]$
		to the intermediate initial data given by 
		\begin{equation}
			\label{intermediate}
			\widetilde	\omega_0^\hbar 
			\equiv 
			\t{Op}_w^\hbar [ f_0 * \mathscr{G}_\hbar] \ . 
		\end{equation}
		Here, $\mathscr{G}_\hbar (z )= \hbar^{-d/4}   \mathscr{G}_1 (\hbar^{-1} z ) $ 
		is the Gaussian mollifier at scale $\hbar$. 
		In the second step, we use Proposition \ref{proposition wasserstein} 
		and Lemma  \ref{lemma wasserstein} 
		to go from the Hartree-Hartree dynamics, to the Vlasov-Hartree dynamics.
		Let us note that these steps will involve different metrics when measuring the distance of the fermion density. 
		In order to conclude, we put these distances together  
		by restricting  our collection of test functions. 
		
		\vspace{2mm}
		\textit{Step 1.}
		Let $( \omega^\hbar (t) , \vp^\hbar(t))$ solve the Hartree-Hartree dynamics with initial data $(\omega_0^\hbar, \vp_0^\hbar)$,
		and let $( \widetilde \omega ^\hbar  (t),  \widetilde \vp^\hbar (t) )$
		solve the Hartree-Hartree dynamics with initial data
		$ ( \widetilde \omega_0^\hbar  ,  \vp_0^\hbar)$, where
		the fermion component has been defined in \eqref{intermediate}. 
		We consider its Wigner transform as 
		\begin{equation}
			\label{tilde f}
			\widetilde f^ \hbar (t) \equiv W^\hbar [\widetilde \omega ^\hbar (t)] \ , \qquad \forall t \in \R \ , 
		\end{equation}
		which we shall refer to as the \textit{intermediate dynamics}. 
		A direct application of  the stability estimate contained in  
		Proposition \ref{proposition stability hartree}, 
		together with 
		the isometric property \eqref{formula isomorphism} yields 
		\begin{align}
			\label{proof thm2 eq 1}
			|	   f^\hbar (t)  -  \widetilde f ^\hbar (t) 				|_1
			+ 
			\| \vp^\hbar  (t)   - \widetilde \vp^\hbar (t)		\|_{L^2}   
			\leq  C \exp ( C \exp C |t|)
			|	   f^\hbar_0  -   f_0* \mathscr G_\hbar 			|_1 
		\end{align}
		for a constant $C>0$, and all $t\in \R$. 
		It suffices to estimate the rigth hand side of Eq. \eqref{proof thm2 eq 1}.
		In particular, the triangle inequality gives 
		\begin{align}
			\label{proof thm 2 eq 2}
			|	   f^\hbar_0  -   f_0* \mathscr G_\hbar 			|_1 
			\leq 
			|	   f^\hbar_0  -   f_0  			|_1 
			+ 
			|	   f_0  -   f_0* \mathscr G_\hbar 			|_1    \ . 
		\end{align}
		The first term on the right hand side of Eq. \eqref{proof thm 2 eq 2} is already contained in the estimate of Theorem \ref{thm2}, 
		so it suffices to estimate the second term. 
		Indeed, we find that 
		for $\zeta \in \R^{2d}$
		\begin{align}
			|	\hat f_0 (\zeta) -  \widehat{ f_0 * \mathscr G _\hbar }  (\zeta )		| 
			= 
			|	 1 -  \widehat{\mathscr G_1 } (\hbar \zeta )	|	\, |	 \hat f_0 (\zeta) 	|
			\leq  	 \t{Lip}
			(  \widehat{ \mathscr G _1}  )				\hbar |\zeta|  \ . 
		\end{align}
		In the last line we have used the fact that
		$  \|	 \hat f_0	\|_{L^\infty }    \leq    \|	f_0\|_{L^1} \leq 1$, 
		and $\widehat{\mathscr{G} _1 }(0)=1$. 
		Upon taking supremum over $\zeta \in \R^{2d}$
		one finds that 
		$|	f_0 - f_0*\mathscr G _\hbar	|_1 \leq C \hbar$.
		Putting everything together, we find 
		\begin{equation}
			\label{proof thm 2 eq 3}
			|	   f^\hbar (t)  -  \widetilde f ^\hbar (t) 				|_1
			+ 
			\| \vp^\hbar  (t)   - \widetilde \vp^\hbar (t)		\|_{L^2}   
			\leq 
			C 
			\exp ( C \exp (C |t|))
			\Big(
			|	   f^\hbar_0  -   f_0  			|_1 
			+ 
			\hbar 
			\Big)
		\end{equation}
		for a constant $C>0$ and all $t \in \R$.

		\vspace{2mm}
		\textit{Step 2.}
		Let $(f(t), \vp(t))$ be the solution of the Vlasov-Hartree system with initial data
		$(f_0 , \vp_0)$.
		Similarly,  let $(	\widetilde \omega^\hbar (t),   \widetilde \vp ^\hbar (t)		)$ 
		be the solution of the Hartree-Hartree syste with 
		initial data $(\widetilde \omega_0^\hbar , \vp_0^\hbar)$, as defined in  Eq. \eqref{intermediate}. 
		Then, Proposition \ref{proposition wasserstein} immediately implies that 
		\begin{equation}
			\label{proof thm2 eq 3}
			E_\hbar \big( 
			f(t), \widetilde \omega _\hbar (t) 
			\big)
			+
			\|	  \vp(t) -   \widetilde \vp ^\hbar (t)	\|_{L^2} 
			\leq 
			C\exp(C t^2 )
			\Big(
			E_\hbar  
			(f_0 ,  \widetilde \omega _0^\hbar )
			+ 
			\|	 \vp_0   -  \vp_0^\hbar	\|_{L^2} 
			\Big) \ . 
		\end{equation}
		Consequently, letting $\widetilde f ^\hbar (t) = W^\hbar [\widetilde \omega^\hbar (t)]$
		be the Wigner transform of the intermediate dynamics, 
		we combine Eq. \eqref{proof thm2 eq 3} with  
		the two estimates found in  Lemma \ref{lemma wasserstein}
		to conclude that 
		\begin{equation}
			\label{proof thm 2 eq 4}
			\|	 f(t) - \widetilde f ^\hbar (t)	\|_{\dot H ^{-1}}
			+ 
			\|	  \vp(t) -   \widetilde \vp ^\hbar (t)	\|_{L^2}  
			\leq 
			C\exp(C t^2 )
			\Big(
			\hbar^{1/2} 
			+ 
			\|	  \vp_0 -   \vp_0^\hbar 	\|_{L^2} 
			\Big) \ . 
		\end{equation}
		\vspace{1mm}
		\textit{Conclusion.} 
		First, we estimate the density of the fermions.
		Namely, let $f^\hbar(t)$ and $f(t)$ be as in the statement of Theorem \ref{thm2}, 
		and let $\widetilde f^\hbar (t)$
		be the intermediate dynamics we have previously introduced. 
		Further, let us  consider a test function $h : \R^{2d} \rightarrow \C$
		such that 
		$  \|	  \< \eta\> \hat h 	 \|_{L^1}$
		and 
		$  \|	  \< \eta\> \hat h 	 \|_{L^2}$
		are both finite.
		Then, the triangle inequality  combined 
		with Eqs. \eqref{proof thm 2 eq 3} and \eqref{proof thm 2 eq 4}
		imply that  
		\begin{align}
			\nonumber 
			|	 \<   h , ( f(t)   -  f^\hbar (t))  \>		|	
			& \leq 
			|	 \langle    h , ( f^\hbar (t)   -   \widetilde f ^\hbar (t))  \rangle 		|		
			+
			|	 \<   h , ( f(t)   -  \widetilde f ^\hbar (t))  \>		|	 	\\
			& \leq  
			\|	  \<\zeta\>  \hat h  	\|_{L^1}
			|	 f^\hbar (t)  - 	  \widetilde f ^\hbar (t)  		|_1 
			+ 
			\|	   |\zeta|  \hat h  	\|_{L^2}
			\|	 f(t) - \widetilde f ^\hbar (t)	\|_{\dot H ^{-1}  } \\ 
			\nonumber 
			& \leq 
			C_2(t) 
			\|	  \<\zeta\>  \hat h  	\|_{L^1}
			\Big(
			|	   f^\hbar_0  -   f_0  			|_1 
			+ 
			\hbar 
			\Big)
			+ 
			C_1(t) 
			\|	   |\zeta|  \hat h  	\|_{L^2}
			\Big(
			\hbar^{1/2} 
			+ 
			\|	  \vp_0 -   \vp_0^\hbar 	\|_{L^2} 
			\Big)  
		\end{align} 
		where 
		$C_1(t) \equiv 
		C\exp(C t^2 )
		$
		and 
		$C_2(t) \equiv C \exp ( C \exp C |t|)$. 
		Similarly, for the boson fields we find
		\begin{align}
			\nonumber 
			\| \vp(t)     -  \vp^\hbar (t)		\|_{L^2 }
			& 	\leq 
			\| \vp(t)     -  \widetilde \vp ^\hbar (t)		\|_{L^2 }
			+
			\|   \vp^\hbar (t)	 -   \widetilde \vp^\hbar (t)    	\|_{L^2 }\\
			&  \leq 
			C_2(t) 
			\Big(
			|	   f^\hbar_0  -   f_0  			|_1 
			+ \hbar 
			\Big) 
			+
			C_1 (t) 
			\Big(
			\hbar^{1/2} 
			+ 
			\|	  \vp_0 -   \vp_0^\hbar 	\|_{L^2} 
			\Big) \ , 
		\end{align}
		where $C_1(t)$ and $C_2(t)$ are as above. This finishes the proof of the theorem. 
	\end{proof}

		\appendix			
		\section{Well-posedness of the PDEs}
		\label{section analysis PDEs}
		In this section, we state basic well-\-posedness results for the Hartree-Hartree eqs. and Vlasov-Hartree eqs.
		that we have introduced in Section \ref{section introduction}.
		For notational simplicity, denote by  
		\begin{equation}
			\mathfrak{h}
			\equiv L^2(\R^d)
		\end{equation}
		the one-particle Hilbert space. 
		
		\subsection{The Hartree-Hartree Equation}
		In what follows, we 
		consider the  
		Hartree-Hartree equation
		that couples the fermionic reduced density matrix, and the bosonic field.
		For notational simplicity, 
		we assume  
		the 
		microscopic scaling regime--of course, 
		every result in this section is independent of the
		scaling regime under consideration. 
		That is,  we consider the equation
		\begin{align}
			\begin{cases}
				\label{hartree hartree 1}
				& 	i      \partial_t \omega   
				= 
				[   -     \Delta 
				+ 
				(V   * \rho_B )  , 
				\omega 
				]			\\
				& 	i     \partial_t  \vp 
				= 
				-    \Delta  \vp 
				+   
				(V  *  \rho_F)    
				\vp   \  , \\ 
				& 	(\omega(0) , \vp(0) = (\omega_0 , \vp_0 )
			\end{cases} 
		\end{align}
		for some initial datum 
		$(\omega_0 , \vp_0 )
		\in 
		\mathscr{L}^1( \h  ) 
		\times
		\h 
		$. 
		Here we will be employing the notation
		for the bosonic and fermionic position densities, 
		for $(t,x) \in \R \times \R^d : $
		\begin{equation}
			\rho_B(t,x) \equiv 
			|   \vp (t,x)	|^2
			\qquad
			\t{and}
			\qquad 
			\rho_F (t,x) \equiv 
			M^{-1} 
			\omega(t;x,x) \ . 
		\end{equation}

		\vspace{1mm }
		Here, we only consider bounded potentials.
		The analysis of mean-field equations 
		with such interactions
		is   classical.
		Hence, we state  
		the main results in the next Theorem and omit the proofs.
		For instance, we refer the reader to 
		e.g \cite{Bove1974}
		whose proof can be adapted to the problem at hand. 
		
		\begin{theorem}
			\label{theorem wp 1}
			Let $(\omega_0 , \vp_0 ) \in 
			\mathscr{L}^1( \h  ) \times \h 
			$
			with $\omega_0^* = \omega_0$; 
			and assume the interaction potential  is bounded $V\in L^\infty(\R^d,\R) $. 
			Then, 
			the following statements hold true 
			\begin{enumerate}[leftmargin=0.7cm]
				\item
				{\rm{(Global well-posedness)}}
				There exists a unique global solution
				$(\omega, \vp) \in C ( \R, \mathscr{L}^1( \h ) \times \h ) $
				to the Hartree-Hartree equation \eqref{hartree hartree 1} in mild form: 
				\begin{align}
					\label{hartree hartree mild form}
					\omega(t) & = 
					e^{- it \Delta}
					\omega_0 
					e^{+  it \Delta}
					- 
					i
					\int_0^t 
					e^{- i (t-s ) \Delta}
					\big[
					V* \rho_B(s) , \omega(s) 
					\big]
					e^{ i (t-s ) \Delta}
					\d s  \\ 
					\vp(t) 
					&  = 
					e^{- i t \Delta}
					\vp_0 
					- 
					i
					\int_0^t 
					e^{- i (t-s ) \Delta}
					\big(
					V * \rho_F(s) 
					\big) \vp (s) 
					\d s \ .   
				\end{align}
				Furthermore, there is continuity with respect to the initial data. 
				
				\item 
				{\rm (Unitary evolution)}
				Let $(\omega,\vp)$ be the mild  solution of \eqref{hartree hartree mild form},
				and 
				consider the time-dependent, 
				mean-field Hamiltonians 
				on $H^2(\R^d)$
				\begin{align}
					h_F (t) =  -\Delta + V*\rho_B(t)
					\qquad
					\t{and}
					\qquad 
					h_B(t)  = - \Delta + V*\rho_F (t) \ . 
				\end{align}
				Further, 
				consider the 
				two-parameter unitary evolution groups on 
				$
				\h 
				$
				that the mean-field Hamiltonians generate 
				\begin{align}
					\begin{cases}
						& 	i \partial_t U_F (t,s) = h_F(t) U_F (t,s) \\ 
						& U_F (t,t) = \1 
					\end{cases}
					\quad 
					\t{and}
					\quad 
					\begin{cases}
						& 	i \partial_t U_B (t,s) = h_B (t) U_B (t,s) \\ 
						& U_B (t,t) = \1 
					\end{cases} \ . 
				\end{align}
				Then, for all $t \in \R$
				\begin{equation}
					\omega(t) = U^*_F (t,0) \omega_0 U_F (t,0)
					\qquad
					\t{and}
					\qquad 
					\vp(t) 
					=
					U_B
					(t,0) \vp_0 \ . 
				\end{equation}
				In particular, 
				$\|	 \omega(t)\|_{\tr} = \|	 \omega_0	\|_{\tr}$, 
				$\omega(t)^* = \omega(t)$, and
				$\| \vp(t) 	\|_{L^2 } = \| \vp_0	\|_{L^2}$
				for all $t \in \R$.
				Additionally, if $\omega_0 \geq 0 $, 
				then $\omega (t) \geq 0$
				and, similarly, 
				if $\omega_0^2= \omega_0$, 
				then 
				$\omega(t)^2 = \omega(t)$. 
				
				\vspace{1mm}
				
				\item 
				{\rm (Strong solutions)}
				Consider the spectral decomposition 
				$\omega_0 = \sum_{k=0}^\infty 
				\lambda_k 
				\ket{\phi_k }
				\bra{\phi_k},
				$
				and assume that 
				$ \sum_{k=0}^\infty  | \lambda_k|   \|  \phi_k\|_{H^2}^2    $
				and
				$\|	\vp_0\|_{H^2}$
				are finite. 
				Then,  
				the  solution map  is continuously differentiable
				$ (\omega , \vp) \in  
				C^1(\R; \mathscr{L}^1( \h  )\times \h )  $, 
				and 
				\eqref{hartree hartree 1}
				holds in the strong sense. 
				
			\end{enumerate}
			
		\end{theorem}

		\subsection{The Vlasov-Hartree Equation}
		In this section, we analyze the   Vlasov-Hartree equation that we have introduced in Section \ref{section introduction}. 
		Namely, the coupled PDEs 
		\begin{align}
			\label{vlasov hartree 2}
			\begin{cases} 
				& 		(\partial_t  + p\cdot  \nabla_x + F_\vp  (t) \cdot    \nabla_p   ) f  = 0 	\\
				&		  i \partial_t \vp 
				=   
				- \frac{1 }{2} \Delta   \vp +   V_f(t)  \vp 	\\ 
				&  (f,\vp) ( 0 )  = (f_ 0 , \vp_0)  
				\in L^1_+(\R^{2d})\times L^2(\R^d) 
			\end{cases} 
		\end{align}
		where 
		$(f_0,\vp_0)\in L_+^1(\R^{2d}) \times L^2(\R^d)$
		is some initial  datum. 
		Here and in the sequel, we shall be using the following notation 
		for the boson-driven  force field, 
		and fermion-field potential energy, 
		for $(t,x) \in \R\times \R^d$: 
		\begin{equation}
			F_\vp (t,x) 
			\equiv 
			- \int_{\R^d }
			\nabla V(x-y ) | \vp(t,y)|^2 \d y 
			\qquad
			\t{and}
			\qquad 
			V_f	(t,x)
			\equiv 
			\int_{	\R^{2d}	}  V(x-y ) f (t,y,p ) \d p  \d y 	 
		\end{equation}

		\vspace{2mm}
		The notion of solution we use is the following. 
		Let  $  I = (a,b) \subset \R $ be an open interval 
		containing $0$. 
		We say that a bounded measurable map 
		$(f , \vp)  :  [a,b] \rightarrow L^1(\R^{2d})\times L^2(\R^d)$
		is a \textit{weak-mild} solution on $I $ of the 
		Vlasov-Hartree equation \eqref{vlasov hartree 2}
		with initial data $(f_0, \vp_0)$, 
		if the following three conditions are satisfied: 
		\begin{itemize}[leftmargin=.8cm, label=$\square$]
			\item For all $ h  \in  \mathscr S (\R^{2d}) $ 
			the map $t \mapsto  \<  f(t) , h\>$ is differentiable on $I $
			\item For all $ h \in \mathscr S (\R^{2d})$
			and $ t \in I $
			it holds true  that 
			\begin{align}
				& 		 \frac{d}{dt}
				\<  f(t), h \>  
				= \<	 f(t) , (p \cdot  \nabla_x + 
				F_\vp(t) \cdot   \nabla_p  ) h 	  	 	\> \\
				& 		
				\vp(t)  =   e^{- i t \Delta/2 } \vp_0 
				- i 
				\int_0^t e^{ - i (t-s ) \Delta/2 } 
				V_f (s) 
				\vp(s) \d s  
			\end{align} 
			\item $(f, \vp) (0)  =  (f_0 , \vp_0). $
		\end{itemize}
		We say that a solution is local-in-time if $I \neq \R$,
		and global-in-time if $ I = \R$. 
		
		\vspace{2mm}
		
		Similarly as for the Hartree-Hartree equation, 
		since we consider here only  regular potentials 
		we do not present a proof of the following well-posedness theorem.
		Rather, we refer the reader to
		\cite[Appendix A]{BenedikterPortaSaffirioSchlein2016}
		for a similar result whose proof can be adapted to our problem.

		\begin{theorem}
			\label{theorem wp 2}
			Let $(f_0, \vp_0) \in 
			L_+^1(\R^d) \times L^2(\R^d)
			$
			and assume that
			$\nabla V \in \mathrm{Lip}(\R^d; \R)$. 
			Then, the following statements hold true .
			\begin{enumerate}[leftmargin=0.7cm]
				\item 
				{\rm (Global well-posedness)}
				There exists 
				$(f,\vp) 
				\in L^\infty
				(\R ; 
				L_+^1(\R^d) \times L^2(\R^d)
				)
				$ 
				a
				unique global 
				weak-mild solution 
				to the   Vlasov-Hartree equation \eqref{vlasov hartree 2}.
				Furthermore, there is continuity with respect to the initial data.

				\vspace{1mm}	
				\item 
				{\rm (Characteristics/Unitary Representation)}
				Let 
				$
				(f,\vp)
				$
				be the global weak-mild solution
				to the Vlasov-Hartree equation.
				We denote by 
				$\Phi_\vp(t) : \R^{2d}\rightarrow \R^{2d}$
				volume-preserving diffeomorphism, 
				corresponding to  	the solution map of the ODE 
				\begin{align}
					\label{ODE}
					\begin{cases}
						& 			 d/dt  \, x (t) = p (t) 		\\
						& d/dt \, p (t)	 = F_\vp(t,  x(t) )	\\
						&  (x,v) (0) = (x_0 , p_0 )
					\end{cases}  \  ; 
				\end{align}
				and we denote by 
				$( U_f (t,s)  )_{t,s\in \R}$
				the two-parameter
				family of unitary transformations defined through 
				\begin{equation}
					\begin{cases}
						& i \partial_t U_f (t,s)
						= 
						\big(
						- (1/2)\Delta 
						+ 
						V_f (t) 
						\big)
						U_f(t,s) \\ 
						& U_f (t,t) = \1 
					\end{cases} \ . 
				\end{equation}
				Then, for all $t \in \R$
				there holds
				\begin{align}
					f(t) = f_0 \big(
					\Phi_\vp^{-1} (t)
					\big) \\ 
					\vp(t) = U_f(t,0) \vp_0  \ . 
				\end{align}
				In particular, $\|	f(t)	\|_{L^1} = \|f_0	\|_{L^1}$, 
				$f(t) \geq0$
				and
				$\|	 \vp(t)\|_{L^2} = \| \vp_0	\|_{L^2}, $
				for all $ t \in \R $.

				\vspace{1mm}

				\item 
				{\rm (Strong solutions)}
				Assume additionally that
				$\|	   (1 + |x| + |p|)   \< \nabla_{x,p }\> 	 f_0\|_{L^1}$
				and
				$\|	 \vp_0	\|_{H^2}$	
				are finite. 
				Then, 
				the solution map 
				is differentiable 
				$(f, \vp ) \in 
				C^1(\R; 		L^1_+(\R^{2d})\times L^2(\R^d 		)
				)$
				and the Vlasov-Hartree equation \eqref{vlasov hartree 2}
				holds in the strong sense. 
			\end{enumerate}

		\end{theorem}

		\begin{remark}
			\label{remark grow}
			Let $(f,\vp)$ be the weak-mild solution
			of the Vlasov-Hartree equation extracted from Theorem \ref{theorem wp 2}, 
			and let $\Phi_\vp(t) $
			be the solution map of the associated ODE. 
			Then, a Gr\"onwall argument shows 
			that there exists a map 
			$\kappa : \R \mapsto (0,\infty)$ 
			such that the following bound is satisfied
			\begin{equation}
				| \Phi_\vp(t,z) | 
				\leq \kappa (t) |z| \ , \qquad \forall t\in \R , \ \forall z \in \R^{2d} \ . 
			\end{equation}
		\end{remark}

		\section{Calculation of the Infinitesimal Generator}
		\label{appendix inf generator}
		In this section, we give 
		more details of the calculation
		of the infinitesimal generator
		of the 
		fluctuation   dynamics $\U(t,s)$, 
		introduced in \eqref{unitary U}. 
		This is the time-dependent self-adjoint operator
		$\L(t)$ on $\F$ defined through the equations 
		\begin{equation}
			i \hbar \partial_t \U(t,s)
			= 
			\mathcal L (t) \U(t,s) \ , 
			\qquad
			\U(t,t)=\1 
		\end{equation}
		where $ t,s\in \R$. 
		In what follows, 
		we let 
		$(\omega,\vp)$  be a strong  solution 
		of the Hartree-Hartree equation \eqref{hartree hartree}, 
		so that the unitary maps
		$\W_t$ and $\calR_t$
		defined in \eqref{definition W and R}
		are differentiable with respect to $t\in\R$. 
		Note that the final result of this section is contained in Proposition \ref{prop appendix}
		and requires no $H^2$ regularity  
		of the solution $(\omega,\vp)$.
		Thus, 
		an approximation argument 
		shows that the result also holds for 
		mild solutions--we leave the details to the reader. 
		
		\vspace{2mm}
		
		Our purpose here is to given an explicit representation of $\L(t)$ in terms
		of creation- and annihilation-
		operators. 
		As a first step in our calculation, we see that the unitarity of the maps easily imply that 
		$\L(t)$ is the contribution of three terms. Namely, for all $ t \in \R$ there holds 
		\begin{equation}
			\label{generator 1} 
			\L(t) 
			=
			\i \hbar \partial_t \calR_t^*  \calR_t \otimes \1 
			+
			\1 \otimes \i \hbar \partial_t \W_t^* \W_t 
			+ 
			( \calR_t \otimes \W_t 		)^*
			\H 
			(\calR_t \otimes \W_t ) \ . 
		\end{equation}
		We now proceed to  calculate each term separately. 
		
		\subsection{Calculation of $\partial_t \calR^*_t \calR_t $}
		Let us first recall some facts and notations 
		that we have   introduced
		in Section \ref{section second quantization}.
		Namely, denoting by $\omega_t \equiv \omega(t) = \sum_{i=1}^M \ket{\phi_i(t)}    \bra{\phi_i(t)}$
		the fermion component of the solution of the Hartree-Hartree equation \eqref{hartree hartree}, 
		we let
		$u_t \equiv u(t)$ and $  v_t  \equiv v(t)$
		be the bounded operators on $L^2(\R^d)$
		defined as
		$u_t =    1 - \omega_t 
		$
		and
		$
		v_t 
		= 
		\sum_{i=1}^M 
		\ket{ \overline{ \phi_i(t) }} 
		\bra{ \phi_i(t)  } . 
		$
		The kernels of these operators define the  distributions 
		$	u_{t,x} (y)  \equiv  u_t (y,x)$
		and 
		$				v_{t,x} (y) \equiv v_t (y,x) $,
		for all $t\in \R$ and $x,y\in \R^d$. 
		Finally, 
		let us recall that we have introduced 
		in \eqref{hamiltonians} the  operator $h_F(t)$ 
		as   the one-particle   Hamiltonian, 
		driving the fermion dynamics. 
		In particular, 
		one may verify that for all $t\in \R$ there holds 
		\begin{align}
			\label{u and v equations}			
			\i \hbar \partial_t u_t   = [ h_F (t) , u_t  ]   
			\qquad 
			\t{and}
			\qquad 
			\i \hbar \partial_t \overline{ v_t }   
			& =
			h_F(t) \overline{v_t}
			+
			\overline{v_t} \overline{ h_F(t) }  \  .
		\end{align}

		\vspace{2mm}

		In order to calculate the first term in the expansion \eqref{generator 1}, 
		we start with the following auxiliary Lemma.
		The  proof is a simplification of the argument contained in 
		\cite[Proposition 3.1]{BenedikterJaksicPortaSaffirioSchlein2016}
		for   mixed states.
		Note that a similar calculation has been carried out in 
		\cite{BenedikterDesio2022} for pure states, 
		which unfortunately is not 
		precise enough for our interests.

		\begin{lemma}
			\label{lemma generator R}
			Assume that for all $j=1, \ldots, M $, 
			the orbitals are differentiable 
			in $L^2$, 
			$t\mapsto \phi_j(t)$. 
			Then, we find that for all $t\in \R$
			there holds 
			\begin{equation}
				\i \hbar \partial_t 
				\calR_t^* \, \calR_t 
				=
				\mathcal{S}_F(t) \1 
				+ 
				\int_{\R^{2d}} \calC  (t; x,y) a_x^* a_y \d x \d y 
				+
				\frac{1}{2}
				\Bigg( 
				\int_{ \R^{2d}}
				\mathcal{D}(t ; x,y) a_x^* a_y^* \, dx dy 
				+ h.c
				\Bigg)  \ . 
			\end{equation}
			Here, $\calC (t)$ and $\calD (t) $ are  the operators on $L^2(\R^d)$
			given 
			by 
			\begin{equation}
				\calC (t)  =
				(i \hbar \partial_t u_t) u_t
				+ 
				(i\hbar \partial_t \overline{v_t}) v_t  
				\qquad \t{and} \qquad 
				\calD (t)   = 
				(i \hbar \partial_t u_t) \overline{v_t}
				+ 
				(i\hbar \partial_t \overline{v_t}) \overline{u_t}       \ , 
			\end{equation}
			and $\calS_F(t)\in \R$ is the scalar term.
			\begin{align}
				\textstyle
				\mathcal{S}_F(t)   = 
				- \sum_{j=1}^M
				\<   
				\phi_j(t) 
				, 
				i \hbar \partial_t
				\phi_j(t) \>_{L^2} \ . 
			\end{align}
		\end{lemma}
				\begin{proof}[Sketch of proof]
							For simplicity we write
					$\calR_t = \calR$, 
					$u_t = u$
					and
					$v_t = v$. 
					We start with the following observation: 
					for  $f\in L^2(\R^d)$  
					we obtain the conjugation relations
			\begin{equation}
				\calR^* a^* (f) \calR 
				= 
				a^*(uf)
				+ 
				a(\overline{vf}) \ . 
			\end{equation}
			and thus, taking the time derivative on both sides
			and using $(d/dt \calR )^* \calR = - \calR^* (d/dt \calR)$
			we obtain  
			\begin{equation}
				[ (d/dt \, \calR^*) \calR , \calR^* a^*(f) \calR    ] = 
				a^*\Big(  \frac{du}{dt} f \big) 
				+
				a\Big(
				\frac{d  \overline v}{dt} 
				\overline{f}
				\Big) \ . 
			\end{equation}  
			We plug in again  the conjugation relations back in the previous equation and multiply both sides with $i$
			to to obtain
			\begin{equation}
				\label{lemma b2 eq 1}
				[ 
				( i \, d/dt \, \calR^*) \calR
				,
				a^*(uf)
				+ 
				a(\overline{vf}) 
				]
				=
				a^*\Big(  i \frac{du}{dt} f \big) 
				- 
				a\Big(
				i \frac{d  \overline v}{dt}  
				\overline{f}
				\Big) \ . 
			\end{equation}
			The previous equation 
			determines
			the self-adjoint 
			operator 
			$( i \, d/dt \, \calR^*) \calR$
			up to a scalar term, 
			and implies that 
			is quadratic in creation-
			and annihilation operators.
			Thus, 
			there exists operators 
			$C$ and $D$
			on 
			$L^2(\R^d )$
			with operator kernels
			$C(x,y)$
			and $D(x,y)$, 
			and a scalar $S_F \in \R $
			such that 
			\begin{equation}
				( i \, d/dt \, \calR^*) \calR
				= 
				S_F  \1 
				+ 
				\int_{\R^{2d}} C(x,y) a_x^* a_y \d x \d y 
				+ 
				\frac{1}{2}
				\Bigg(
				\int_{\R^{2d}}
				D(x,y) a_x^* a_y^* \d x \d y 
				+ h.c 
				\Bigg)  \ . 
			\end{equation}
			Here,   we   assume without loss of
			 generality that $\overline{C(y,x)} = C(x,y)$
			and
			$D(y,x) =  - D(x,y)$. 
			Let us first compute that operator contributions. 
			A lengthy but straightforward calculation using the CAR 
			allows us to compute the commutators 
			\begin{align}
				\label{lemma b2 eq 2}
				[ ( i \, d/dt \, \calR^*) \calR , a^*(uf)
				+ 
				a(\overline{vf}) 
				]
				= 
				a^*(C uf )
				- a(C \overline{v} \overline{f})
				+ a^*(D v f )
				- 
				a (D \overline u \overline f) \  . 
			\end{align}
			Thus, we compare the right hand sides of  \eqref{lemma b2 eq 1} and \eqref{lemma b2 eq 2}
			to obtain, as operators on $L^2(\R^d)$, the following equations 
			\begin{align}
				\begin{cases}
					C u + D v &  = i  du/dt  \\
					C \overline v + D \overline u & = i d \overline v / dt  \  \\     
				\end{cases}
				\iff
				(C \ D )
				\begin{pmatrix}
					u & \overline v \\ 
					v & \overline u \\ 
				\end{pmatrix}
				= i d/dt ( u \ \overline v  ) \ . 
			\end{align}
			In particular, the matrix  containing $u$ and $v$
			is a unitary map on $L^2(\R^d) \oplus L^2(\R^d)$, 
			which we denote by $\nu$. 
			In fact, using the relations
			$u^* = u $, 
			$v^* = \overline v$
			it is easy to verify that  $\nu^* = \nu$.
			Thus we invert the last equation and solve for the operators  $C$ and $D$
			to find
			\begin{equation}
				C =  i  du/dt   \, u + i d \overline v /dt  \, v 
				\qquad
				\t{and}
				\qquad 
				D = i du/dt   \, \overline v + i d \overline v /dt \, \overline  u \ . 
			\end{equation}
			Finally,  we multiply with $\hbar$
			and identify 
			$\calC = \hbar C$
			and
			$\calD = \hbar D$. 
			As for the scalar contribution, we consider its vacuum expectation value to obtain 
			\begin{align}
				S_F(t) 
				=  \<  \Omega_F , i d\calR^* /dt \, \calR  \Omega_F\>_{\F_F}
				= - \<     \calR \Omega_F , i  d\calR / dt \, \Omega_F    \>_{\F_F} \ .
			\end{align}
			Let us now calculate the right hand side.
			To this end, we write in terms of the orbitals
			$\omega(t) = \sum_{j=1}^M \ket{\phi_j(t)} \bra{\phi_j(t)}$
			the vector
			$\calR \Omega_F  = a^*( \phi_1) \cdots a^*(\phi_M)\Omega_F$.
			A straightforward calculation
			using the CAR and the orthogonality relations 
			$\<\phi_i, \phi_j\>_{L^2}= \delta_{i,j}$ now implies
			\begin{equation}
				S_F(t)
				= - \sum_{j=1}^M \< \phi_j (t),  i d \phi_j(t)/dt \>_{L^2} \ . 
			\end{equation}
			This finishes the proof once we multiply with $\hbar$
			and identify $\calS = \hbar S $. 
		\end{proof}
		
		Let us now give an   explicit 
		representation 
		of the term that we just calculated. 
		Namely, in view of \eqref{u and v equations}
		it is easy to verify that
		\begin{equation}
			\mathcal C (t) 
			= 
			h_F(t) 
			- 
			u_t h_F(t) u_t
			+ 
			\overline{v_t}   \overline{h_F(t)}   v_t 
			\qquad 
			\t{and}
			\qquad 
			\mathcal D (t)  
			= 
			- u_t h_F(t) \overline{v_t}
			+ 
			\overline{v_t} \overline{h_F(t)} \overline{u_t} 
		\end{equation}
		In particular, the second term may be simplified. Namely,  in view of the relations 
		$h_F (t)  = h_F^*(t) $, 
		$ u_t = u^*_t$ and $ \bar v _t  = v_t^*  $, one may check that
		$ (vhu)^* = u h \bar v	$. 
		Consequently, using the  anti-commutation relation $\{a_x^* ,  a_y^*\} = 0 $, 
		one  finds that 
		\begin{equation}
			\int_{ \R^d \times \R^d }  (  \overline{v h u}  )  (x,y) a_x^* a_y^* \d x \d y
			= 
			- 
			\int_{ \R^d \times \R^d } 
			( u h \overline v)(x,y)
			a_x^* a_y^* \d x \d y \ . 
		\end{equation}
		We can then put the calculation from the above lemma in the following form, 
		which is the final result of our calculation
		(we omit the time labels for convenience) 
		\begin{equation}
			( 	i \hbar \partial_t \calR^*)   \calR 
			=
			- \tr \big(
			h_F \omega 
			\big) 
			+ 
			\d \Gamma_F \Big[
			h_F 
			- 
			u  h_F u 
			+ 
			\overline{v}   \overline{h_F}   v 
			\Big] 
			- 
			\Bigg( 				
			\int_{ \R^{2d} } 
			[ u h_F  \bar{v} ](x,y) a_x^* a_y^*   
			\d x \d y  
			+ h.c 
			\Bigg)
		\end{equation}

		\subsection{Calculation of $\partial_t \W^*_t \W_t $}
		Let us now calculate the second contribution of \eqref{generator 1}.
		The time derivatives of Weyl operators that are parametrized
		by a   field $t\mapsto \alpha (t) \in L^2(\R^d)$
		can be regarded  as classical result, 
		and we record it in the following lemma. 
		For reference, see
		\cite[Lemma 3.1]{GinibreVelo1974}.

		\begin{lemma}
			Assume that
			$t \mapsto \vp_t \in L^2(\R^d)$ 
			is   differentiable.
			Then, for all $t\in \R$ there holds 
			\begin{equation}
				i \hbar \partial_t 
				\W_t^* \W_t
				= 
				N  \mathrm{Im}
				\< \vp_t ,  \hbar  \partial_t \vp_t \> 
				-
				\sqrt N 
				\Big(
				b^* ({ i \hbar \partial_t \vp_t}   ) 
				+
				b  ( i \hbar \partial_t \vp_t) 
				\Big)  \ . 
			\end{equation}  
		\end{lemma}
		
		We shall use the fact that $\vp(t)$
		solves the Hartree-Hartree equation \eqref{hartree hartree}. 
		Namely, we consider on $H^2(\R^d)$
		the time-dependent bosonic   Hamiltonian 
		$h_B(t)$
		defined in \eqref{hamiltonians} and we conclude that for
		all $t\in \R$
		there holds 
		\begin{equation}
			i \hbar \partial_t 
			\W_t^* \W_t
			= 
			-
			N \mathrm{Re}
			\< \vp_t , h_B(t)  \vp_t \> 
			-
			\sqrt N 
			\Big(
			b^* ( h_B(t)  \vp_t   ) 
			+
			b  ( h_B(t)  \vp_t) 
			\Big)  \ .
		\end{equation}
		
		\subsection{Calculation of $\calR^*_t \W^*_t \H \W_t \calR_t $}		
		
		Our next task it to compute 
		$
		( \calR_t \otimes \W_t 		)^*
		\H 
		(\calR_t \otimes \W_t ) \ . 
		$
		To this end, we shall use extensively
		the conjugation relations 
		for particle-hole transformations   (see Lemma \ref{lemma fermion properties}) 
		\begin{align}
			\calR_t^* a^*_x \calR_t 
			= 
			a^* ( u_{t,x}) 
			+ 
			a( \overline{  v_{t,x} } )		
			\qquad 
			\t{and}
			\qquad 
			\calR_t^* a_x \calR_t 
			= 
			a ( u_{t,x}) 
			+ 
			a^*( \overline{  v_{t,x} } ) \ , 
		\end{align}
		for all $(t,x) \in \R\times \R^d$, and similarly for Weyl operators (see Lemma \ref{lemma properties bosons}) 
		\begin{align}	 
			\W_t^* b_x  \W_t
			=
			b_x +  	\sqrt  N \vp_t(x)		 
			\qquad 
			\t{and}
			\qquad 
			\W_t^* b^*_x  \W_t
			=
			b^*_x +  \sqrt N 	\overline{ 	\vp_t(x)  } \  .
		\end{align}
		where $(t,x) \in \R\times \R^d$. 
		A lengthy but otherwise straightforward calculation
		using these conjugation relations, 
		together with $u^*=u$
		and
		$v^*=\overline v$, 
		yields the following result. 
		\begin{lemma}
			\label{lemma generator conjugation}
			The following holds for all $t \in \R$
			\begin{align}
				( \calR_t \otimes &  \W_t 		)^*					\nonumber 
				\H 		
				(\calR_t \otimes \W_t )  \\
				& 	=
				\bigg( 
				N \|	  \nabla \vp_t	\|_{L^2}^2 		+ \t{Tr }( -\Delta\omega_t ) 	
				\nonumber 
				+ 
				\lambda NM
				\int_{ \R^{2d} } 
				\rho_F(t,x )
				V(x-y) 
				\rho_B(t,y)
				\, dx dy 
				\bigg)  \, \1 \otimes \1 
				\\
				&  + 
				\Bigg(						\nonumber 
				\d \Gamma_F 
				\big[   
				u_t h_F(t) u_t 
				-
				\overline{v_t} h_F(t) v_t 
				\big] 
				+ 
				\int_{ \R^{2d} } 
				( u_t h_F(t) \bar v _t)(x,y) a_x^* a_y^* dx \, dy
				+ \t{h.c}
				\Bigg)
				\otimes \1 			 \\
				&   
				+  
				\1 \otimes
				\bigg(
				\d \Gamma_B 
				\big[ h_B(t) \big]
				+
				\sqrt N \big(	b^* (h_B(t) \vp_t )		+ 	b(h_B(t) \vp_t )	\big) 	
				\bigg) 
				\nonumber 			\\
				&  +  \lambda \sqrt  N \L_{2,1} (t) + \lambda \L_{2,2} (t) \ . 
			\end{align}
			Here, we denote
			\begin{align}
				\nonumber 
				\L_{2,1} (t) 
				& = 
				\int_{ \R^d \times \R^d  } V(x - y ) a^*(u_{t,x}) a( u_{t,x}) \otimes \big(		 \vp_t (y) b_y^* + h.c	\big)
				\ \d x \d y 	\\
				\nonumber 
				& - 
				\int_{ \R^d \times \R^d  } V(x - y ) a^*( \overline{v _{t,x}}) a(\overline{v_{t,x}}) \otimes
				\big(		 \vp_t(y) b_y^* + h.c	\big) \ \d x \d y 	\\
				\nonumber 
				& + 
				\int_{ \R^d \times \R^d  } V(x - y ) a^*(  {u_{t,x}}) a^*(\overline{v_{t,x}}) 
				\otimes
				\big(		 \vp_t(y) b_y^* + h.c	\big)  \ \d x \d y 	\\
				& + 	 
				\int_{ \R^d \times \R^d  } V(x - y ) a(  \overline{v _{t,x}}    ) a( u_{t,x} ) \otimes 
				\big(		 \vp_t(y) b_y^* + h.c	\big)	
				\ \d x \d y 	  
			\end{align}
			and 
			\begin{align}
				\nonumber 
				\L_{2,2} (t) 
				& = 
				\int_{ \R^d \times \R^d  } V(x - y ) 
				a^*(  u_{t,x}) a(    u_{t,x} ) 
				\otimes b_y^*b_y
				\ \d x \d y 	\\
				\nonumber 
				& - 
				\int_{ \R^d \times \R^d  } V(x - y ) 
				a^*( \overline{v _{t,x}}) 
				a(\overline{v_{t,x}}) 
				\otimes b_y^*  b_y
				\ \d x \d y 	\\
				\nonumber 
				& +
				\int_{ \R^d \times \R^d  } V(x - y ) 
				a^*(  {u_{t,x}}) 
				a^*(\overline{  v_{t,x} }) 
				\otimes b_y^*b_y
				\ \d x \d y 	\\
				& + 		 
				\int_{ \R^d \times \R^d  } V(x - y ) a(  \overline{v _{t,x}}    ) a( u_{t,x}  ) \otimes b_y^*b_y
				\ \d x \d y 	  \ . 
			\end{align} 
		\end{lemma}

		\subsection{Putting everything together}
		We put together the last three lemmas to find the following explicit representation of the generator $\L(t)$. 
		
		\begin{proposition}
			\label{prop appendix}
			Let $\U(t,s)$ be the unitary transformation defined in \eqref{unitary U}, and let $\L(t)$ be its inifinitesimal generator. 
			Then, $\L(t)$ admits the following representation 
			\begin{equation}
				\L(t) 
				=
				S(t) \1 \otimes \1 
				+ 
				\d \Gamma [h_F(t) ] \otimes \1 
				+
				\1 \otimes \d \Gamma [h_B(t)]
				+
				\lambda \sqrt N \L_{2,1}(t) 
				+
				\lambda \L_{2,2}(t)  \ . 
			\end{equation}
			Here
			the scalar term is
			$S(t) = -  \lambda NM \int_{\R^{2d}}V(x-y) \rho_F(t,x) \rho_B(t,y) dxdy$, 
			$h_F(t)$ and $h_B(t)$ are the 1-particle Hamiltonians defined in 
			\eqref{hamiltonians}, and
			the operators $\L_{2,1}(t) $ and $\L_{2,2}(t)$ are defined in Lemma \ref{lemma generator conjugation}. 
		\end{proposition}


\end{document}